\documentclass[letterpaper, USenglish, 11pt, fleqn]{article}
\pdfoutput=1

\usepackage{microtype}
\usepackage[svgnames]{xcolor}
\usepackage[
	pdfusetitle,
	unicode=true,
	colorlinks=true,
	linkcolor=NavyBlue,
	citecolor=NavyBlue,
	urlcolor=NavyBlue,
]{hyperref}
\usepackage{pdfcomment}
\usepackage[
	margin=1in
]{geometry}
\usepackage{amsmath,amsthm}
\usepackage{newtxtext}
\usepackage{newtxmath}
\usepackage{csquotes}
\usepackage[shortlabels]{enumitem}
\usepackage[
	capitalize
]{cleveref}
\usepackage{caption}
\usepackage{subcaption}
\usepackage{authblk}
\usepackage{xargs}
\usepackage{xspace}

\usepackage{mathtools}
\usepackage{thmtools, thm-restate}
\usepackage{bm}
\usepackage{tikz} %
\usetikzlibrary{automata, graphs,positioning,chains,arrows,decorations.pathmorphing,calc,shapes}

\tikzset{
	/semifill/left/.initial=none,
	/semifill/right/.initial=none,
	semifill/.style={path picture={
		\pgfqkeys{/semifill}{#1}
		\fill[\pgfkeysvalueof{/semifill/left}] 
		(path picture bounding box.south) rectangle
		(path picture bounding box.north west);
		\fill[\pgfkeysvalueof{/semifill/right}] 
		(path picture bounding box.south) rectangle
		(path picture bounding box.north east);
	}},
}

\tikzset{
	inter/.style={inner sep=1.5pt, draw=black},
	UP/.style={inter, densely dotted},
	UA/.style={inter, inner sep=1pt, ellipse},
	UM/.style={inter},
	US/.style={inter, inner sep=1pt, densely dotted, ellipse},
	Red/.style={fill=red!40},
	Green/.style={fill=green!40},
	Blue/.style={fill=blue!40},
	Yellow/.style={fill=yellow!40},
	RedYellow/.style={semifill={left=red!40, right=yellow!40}},
	YellowRed/.style={semifill={left=yellow!40, right=red!40}},
	RedGreen/.style={semifill={left=red!40, right=green!40}},
	GreenRed/.style={semifill={left=green!40, right=red!40}},
	BlueRed/.style={semifill={left=blue!40, right=red!40}},
	RedBlue/.style={semifill={left=red!40, right=blue!40}},
} %
\newenvironment{claimproof}{\begin{proof}}{\end{proof}}      
\declaretheorem[
	name=Definition,
	numberwithin=section,
	style=definition,
	qed={\ensuremath{\lrcorner}}
]{definition}

\declaretheorem[
	name=Example, 
	sibling=definition,
	style=definition,
	qed={\ensuremath{\lrcorner}}
]{example}
\let\endExample\qedhere
		
\declaretheorem[
	name=Remark, 
	sibling=definition,
	style=definition
]{remark}

\declaretheorem[
	name=Observation, 
	sibling=definition,
	style=definition
]{observation}

\declaretheorem[
	name=Corollary, 
	sibling=definition
]{corollary}

\declaretheorem[
	name=Lemma, 
	sibling=definition
]{lemma}

\declaretheorem[
	name=Fact, 
	sibling=definition
]{fact}

\declaretheorem[
	name=Claim, 
	sibling=definition,
	style=definition,
	qed={\ensuremath{\lrcorner}}
]{claim}
\let\endClaim\qedhere

\declaretheorem[
	name=Claim,
	numbered=no,
	style=definition,
	qed={\ensuremath{\lrcorner}}
]{claim*}

\declaretheorem[
	name=Proposition,
	sibling=definition,
	style=definition,
	qed={\ensuremath{\lrcorner}}
]{proposition}
\let\endProposition\qedhere

\declaretheorem[
	name=Theorem,
	sibling=definition,
]{theorem}

\renewcommand{\epsilon}{\varepsilon}

  \newcommand{\exrel}[1]{\textsl{#1}}
 \newcommand{\exdata}[1]{\texttt{#1}}
 
 \tikzset{
 	defaultstyle/.style={
 		>=stealth, 
 		semithick, 
 		auto,
 		initial text= {},
 		initial distance= {3mm},
 		accepting distance= {3mm}}}

\newcommand{\exsigma}{\sigma}
\newcommand{\exsigmaSimple}{\widehat{\sigma}}

 \newcommand{\exdb}{D}
 \newcommand{\exdbSimple}{\widehat{D}}

\newcommand{\nc}[1]{\newcommand{#1}}
\newcommand{\rnc}[1]{\renewcommand{#1}}

\nc{\myparagraph}[1]{\paragraph{#1}}
\rnc{\leq}{\ensuremath{\leqslant}}
\rnc{\geq}{\ensuremath{\geqslant}}

\rnc{\le}{\leq}
\rnc{\ge}{\geq}

\nc{\isdef}{\coloneqq}
\nc{\deff}{\isdef}

\nc{\set}[1]{\ensuremath{\{#1\}}}
\nc{\smallset}[1]{\ensuremath{\{#1\}}}
\nc{\setsize}[1]{\ensuremath{|#1|}}
\nc{\Setsize}[1]{\ensuremath{\big|#1\big|}}
\nc{\Set}[1]{\ensuremath{\big\{#1\big\}}}
\nc{\setc}[2]{\set{#1 \, : \, #2}}
\nc{\Setc}[2]{\Set{#1 \, : \, #2}}

\nc{\aufgerundet}[1]{\ensuremath{\lceil #1 \rceil}}
\nc{\abgerundet}[1]{\ensuremath{\lfloor #1 \rfloor}}

\nc{\dcup}{\ensuremath{\dot\cup}}

\nc{\ov}[1]{\ensuremath{\overline{#1}}}

\nc{\NN}{\ensuremath{\mathbb{N}}}
\nc{\NNpos}{\ensuremath{\NN_{\geq 1}}}
\nc{\RR}{\ensuremath{\mathbb{R}}}
\nc{\RRpos}{\ensuremath{\RR_{\geq 0}}}
\nc{\QQ}{\ensuremath{\mathbb{Q}}}
\nc{\QQpos}{\ensuremath{\QQ_{\geq 0}}}

\nc{\und}{\ensuremath{\wedge}}
\nc{\Und}{\ensuremath{\bigwedge}}
\nc{\oder}{\ensuremath{\vee}}
\nc{\Oder}{\ensuremath{\bigvee}}
\nc{\nicht}{\ensuremath{\neg}}
\nc{\impl}{\ensuremath{\to}}
\nc{\gdw}{\ensuremath{\leftrightarrow}}

\nc{\Semijoin}{\ensuremath{\ltimes}}

\nc{\free}{\ensuremath{\textrm{\upshape free}}}
\nc{\quant}{\ensuremath{\textrm{\upshape quant}}}
\nc{\ar}{\ensuremath{\operatorname{ar}}}

\nc{\Structure}[1]{\ensuremath{\mathcal{#1}}}
\nc{\A}{\Structure{A}}
\nc{\B}{\Structure{B}}
\nc{\C}{\Structure{C}}

\nc{\isom}{\ensuremath{\cong}}

\nc{\querycont}{\ensuremath{\sqsubseteq}}
\nc{\eval}[2]{\ensuremath{#1(#2)}}
\nc{\semantik}[1]{\ensuremath{\left\llbracket#1\right\rrbracket}}
\nc{\sem}[1]{{\semantik{#1}}}
\nc{\bigsem}[1]{\ensuremath{\big\llbracket #1 \big\rrbracket}}
\nc{\smallsem}[1]{\ensuremath{\llbracket #1 \rrbracket}}

\nc{\CanDB}[1]{\ensuremath{\A_{#1}}} %
\nc{\CanTup}[1]{\ensuremath{t_{#1}}} %

\newcommand{\queryphi}{\varphi}

\newcommand{\relS}{S} %
 
\newcommand{\relT}{T} %

\newcommand{\relE}{E} %

\nc{\Vars}{\ensuremath{\textrm{\upshape vars}}}
\nc{\vars}{\Vars}
\nc{\Cons}{\ensuremath{\textrm{\upshape cons}}}
\nc{\cons}{\Cons}
\nc{\atoms}{\ensuremath{\textrm{\upshape atoms}}}
\nc{\Atoms}{\atoms}
\nc{\Atom}{\myatom}
\nc{\Adom}{\ensuremath{\textrm{\upshape adom}}}
\nc{\adom}[1]{\ensuremath{\Adom(#1)}} %
\nc{\dom}[1]{\ensuremath{\textrm{\upshape dom}(#1)}} %

\nc{\Witness}{\ensuremath{W}} %
\nc{\myAtoms}{\ensuremath{\textrm{\upshape atm}}}

\newcommand{\poly}{\operatorname{\textit{poly}}}

\newcommand{\DBone}[1]{}

\newcommand{\bigoh}{O}
\newcommand{\bigOh}{\bigoh}

\newcommand{\parent}{\pointerfont{parent}}

\nc{\arrayfont}[1]{\ensuremath{\texttt{#1}}}

\newcommand{\size}[1]{\ensuremath{|\!|#1|\!|}}
\nc{\dbsize}[1]{{\ensuremath{|#1|}}}
\nc{\querysize}[1]{\size{#1}}

\nc{\card}[1]{\ensuremath{|#1|}}

\newcommand{\assign}{\ensuremath{\alpha}}

\nc{\insertp}{\textsc{Insert}}
\nc{\cleanup}{\textsc{cleanUp}}
\nc{\cleanups}{\textsc{cleanUp'}}

\nc{\Dom}{\ensuremath{\textbf{dom}}}
\nc{\Var}{\ensuremath{\textbf{var}}}
\nc{\schema}{\ensuremath{\sigma}}
\nc{\DB}{\ensuremath{D}} %
\nc{\DBstrich}{\ensuremath{D'}} %
\nc{\DBstart}{\ensuremath{{\DB_0}}} %
\nc{\DBempty}{\ensuremath{{\DB_{\emptyset}}}} %

\nc{\DS}{\ensuremath{\mathtt{D}}} %

\rnc{\phi}{\queryphi}

\nc{\UpdateFont}[1]{\ensuremath{\textsf{#1}}}
\nc{\Delete}{\UpdateFont{delete}}
\nc{\Insert}{\UpdateFont{insert}}
\nc{\Update}{\UpdateFont{update}}

\nc{\AlgoFont}[1]{\ensuremath{\textbf{#1}}}
\nc{\PREPROCESS}{\AlgoFont{preprocess}}
\nc{\INIT}{\AlgoFont{init}}
\nc{\UPDATE}{\AlgoFont{update}}
\nc{\ENUMERATE}{\AlgoFont{enumerate}}
\nc{\COUNT}{\AlgoFont{count}}
\nc{\ANSWER}{\AlgoFont{answer}}
\nc{\TEST}{\AlgoFont{test}}

\nc{\EOE}{\texttt{\upshape EOE}\xspace} %

\nc{\preprocessingtime}{\ensuremath{t_{\operatorname{prep}}}}
\nc{\delaytime}{\ensuremath{t_{\operatorname{delay}}}}
\nc{\preprocessingtimefunc}{\ensuremath{f_{\operatorname{prep}}}}
\nc{\delaytimefunc}{\ensuremath{f_{\operatorname{delay}}}}
\nc{\indexingtime}{\ensuremath{t_{\operatorname{index}}}}
\nc{\indexingtimefunc}{\ensuremath{f_{\operatorname{index}}}}
\nc{\booltime}{\ensuremath{t_{\booltask}}}
\nc{\countingtime}{\ensuremath{t_{\counttask}}}
\nc{\booltimefunc}{\ensuremath{f_{\booltask}}}
\nc{\countingtimefunc}{\ensuremath{f_{\counttask}}}

\nc{\preprocessingtimehat}{\ensuremath{\hat{t}_p}}
\nc{\inittimehat}{\ensuremath{\hat{t}_i}}
\nc{\delaytimehat}{\ensuremath{\hat{t}_d}}
\nc{\updatetimehat}{\ensuremath{\hat{t}_u}}
\nc{\answertimehat}{\ensuremath{\hat{t}_a}}
\nc{\countingtimehat}{\ensuremath{\hat{t}_c}}
\nc{\testingtimehat}{\ensuremath{\hat{t}_t}}

\nc{\phiBTypical}{\ensuremath{\phi'_{\relS\text{-}\relE\text{-}\relT}}}
\nc{\phiJTypical}{\ensuremath{\phi_{\relS\text{-}\relE\text{-}\relT}}}
\nc{\phiET}{\ensuremath{\phi_{\relE\text{-}\relT}}}

\nc{\restrict}[2]{\ensuremath{{#1}_{|#2}}}
\nc{\extend}[3]{\ensuremath{{#1}\tfrac{#3}{#2}}}
\nc{\emptyassign}{\ensuremath{\emptyset}}
\nc{\Assign}[2]{\ensuremath{\frac{#2}{#1}}}

\nc{\vroot}{\ensuremath{\varv_{\textsl{root}}}}

\nc{\pointerfont}[1]{\textit{#1}}

\nc{\varitem}[1]{\ensuremath{v^{#1}}}
\nc{\assitem}[1]{\ensuremath{\assign^{#1}}}
\nc{\constitem}[1]{\ensuremath{a^{#1}}}
\nc{\parentitem}[1]{\ensuremath{\parent^{#1}}}
\nc{\childitem}[2]{\ensuremath{\pointerfont{child}^{#1}_{#2}}}
\nc{\llist}[2]{\ensuremath{\mathcal{L}_{#2}^{#1}}}
\nc{\startlist}{\ensuremath{\mathcal{L}_{\text{\upshape start}}}\xspace}
\nc{\nextlistitem}[1]{\ensuremath{\pointerfont{next-listitem}^{#1}}}
\nc{\prevlistitem}[1]{\ensuremath{\pointerfont{prev-listitem}^{#1}}}
\nc{\countitem}[1]{\ensuremath{C_{\textit{below}}^{#1}}}
\nc{\desc}[1]{\ensuremath{\text{desc}}}

\nc{\Null}{\ensuremath{0}}

\nc{\arrayA}{\arrayfont{A}}
\nc{\arrayB}{\arrayfont{B}}
\nc{\arrayC}{\arrayfont{C}}
\nc{\arrayE}{\arrayfont{E}}

\nc{\ITEMS}{\mathcal{I}}

\nc{\NIL}{\textsc{nil}}

\nc{\TupleSet}{\ensuremath{\mathcal{T}}}
\nc{\ResultSet}{\ensuremath{\mathcal{R}}}
\nc{\SkipArrayNext}[1]{\ensuremath{\mathsf{skip}[#1].\mathsf{next}}}
\nc{\SkipArrayPrev}[1]{\ensuremath{\mathsf{skip}[#1].\mathsf{prev}}}
\nc{\AlgoA}{\ensuremath{\mathbb{A}}}
\nc{\nil}{\texttt{nil}\xspace}
\nc{\SkipStart}{\ensuremath{\mathsf{sk{-}start}}}
\nc{\tups}{\ensuremath{\ov{s}}}
\nc{\prozvisit}{\ensuremath{\textsc{Visit}}}
\nc{\prozvisitrev}{\ensuremath{\textsc{Visit}^{-1}}}
\nc{\tut}{\ensuremath{t}}
\nc{\enumprev}{\ensuremath{\vartriangleleft}}
\nc{\SkipLast}{\ensuremath{\mathsf{sk{-}last}}}

\nc{\lllist}{\ensuremath{\mathcal{L}}}
\nc{\pllist}{\ensuremath{\mathcal{L}^+}}
\nc{\milist}{\ensuremath{\mathcal{L}^-}}
\nc{\cilist}{\ensuremath{\mathcal{L}^\circ}}
\nc{\numitmpl}{\ensuremath{+{-}\text{on}{-}\text{path}}}
\nc{\numitmmi}{\ensuremath{-{-}\text{on}{-}\text{path}}}
\nc{\numitmci}{\ensuremath{\circ{-}\text{on}{-}\text{path}}}
\nc{\DBnew}{\ensuremath{\DB_{\text{new}}}}
\nc{\DBold}{\ensuremath{\DB_{\text{old}}}}
\nc{\liitmpl}{\ensuremath{\mathcal{L}^{+{-}\text{on}{-}\text{path}}}}
\nc{\liitmmi}{\ensuremath{\mathcal{L}^{-{-}\text{on}{-}\text{path}}}}
\nc{\liitmci}{\ensuremath{\mathcal{L}^{\circ{-}\text{on}{-}\text{path}}}}
\nc{\ITEMSres}[1]{\ensuremath{\ITEMS|_{#1}}}

\nc{\prVisit}{\textsc{Visit}}
\nc{\prVisitRes}{\textsc{VisitRes}}
\nc{\prEnumWithItem}{\textsc{EnumWithItem}}
\nc{\prFindItems}{\textsc{FindItems}}

\nc{\SUBW}{\ensuremath{\textit{subw}}} %
\nc{\ADW}{\ensuremath{\textit{adw}}} %
\nc{\fcSUBW}{\ensuremath{\textit{fc-subw}}} %
\nc{\fcFHW}{\ensuremath{\textit{fc-fhw}}} %
\nc{\fcGHW}{\ensuremath{\textit{fc-ghw}}} %

\nc{\TD}{\ensuremath{\textit{TD}}} %
\nc{\FDecom}{\ensuremath{\textit{F}}} %

\nc{\emptytuple}{\ensuremath{()}}
\nc{\emptyword}{\ensuremath{\varepsilon}}

\nc{\proj}{\ensuremath{\pi}}
\nc{\projgen}{\ensuremath{\uppi}}
\nc{\select}{\ensuremath{\sigma}}

\nc{\FD}{\ensuremath{\delta_{\textit{fd}}}} %
\nc{\IND}{\ensuremath{\delta_{\textit{ind}}}} %
\nc{\INDtilde}{\ensuremath{\tilde{\delta}_{\textit{ind}}}}%
\nc{\SD}{\ensuremath{\delta_{\textit{sd}}}} %
\nc{\CC}{\ensuremath{\delta_{\textit{cc}}}}%

\nc{\DEP}{\ensuremath{\delta}} %
\nc{\CONSTR}{\ensuremath{\Sigma}} %

\nc{\qSET}{\ensuremath{q_{\textit{S-E-T}}}}
\nc{\pSET}{\ensuremath{p_{\textit{S-E-T}}}}

\nc{\qET}{\ensuremath{q_{\textit{E-T}}}}

\nc{\Ans}{\ensuremath{\textit{Ans}}}

\nc{\query}{\ensuremath{Q}}
\nc{\qatom}{\ensuremath{\myatom}}

\nc{\HG}{\ensuremath{\mathcal{H}}} %
\nc{\Nodes}{\ensuremath{V}} %
\nc{\Edges}{\ensuremath{E}} %

\nc{\HD}{\ensuremath{\textit{H}}} %
\nc{\Tree}{\ensuremath{T}}
\nc{\rootedTree}{\ensuremath{\hat{\Tree}}}

\nc{\treenode}{\ensuremath{t}} %
\nc{\treenodeparent}{\ensuremath{p}} %
\nc{\parentnode}{\treenodeparent} 
\nc{\treeroot}{\ensuremath{r}} %
\nc{\Bag}{\ensuremath{\textit{bag}}}
\nc{\Cover}{\ensuremath{\textit{cover}}}

\nc{\FHD}{\ensuremath{\textit{F}}} %
\nc{\Weight}{\ensuremath{\textit{weight}}} %

\nc{\Width}{\ensuremath{\textit{width}}}
\nc{\GHW}{\ensuremath{\textit{ghw}}} %
\nc{\FHW}{\ensuremath{\textit{fhw}}} %

\nc{\freetreenodes}{\ensuremath{U}} %

\nc{\prunedTree}{\ensuremath{\tilde{\Tree}}}
\nc{\prunedSchema}{\ensuremath{\tilde{\schema}}}
\nc{\prunedDB}{\ensuremath{\tilde{\DB}}}
\nc{\prunedDBold}{\ensuremath{\prunedDB_{\textit{old}}}}
\nc{\prunedDBnew}{\ensuremath{\prunedDB_{\textit{new}}}}
\nc{\prunedQuery}{\ensuremath{\tilde{\query}}}

\nc{\Start}{\ensuremath{\texttt{fetch-first}}}
\nc{\Next}{\ensuremath{\texttt{fetch-next}}}
\nc{\TestTuple}{\ensuremath{\texttt{test}}}
\nc{\PositionCursor}{\ensuremath{\texttt{position-cursor}}}
\nc{\AccessJth}{\ensuremath{\texttt{access}}}
\nc{\RankTuple}{\ensuremath{\texttt{rank}}}

\nc{\Mapping}[1]{\ensuremath{\tilde{#1}}} %
\nc{\MappingR}{\ensuremath{\tilde{R}}}

\nc{\Algo}[1]{\ensuremath{\textsc{#1}}}

\nc{\True}{\ensuremath{\textsl{true}}}
\nc{\False}{\ensuremath{\textsl{false}}}

\nc{\Yes}{\True}
\nc{\No}{\False}

\nc{\QueryClass}{\ensuremath{\mathcal{Q}}}
\nc{\fcACQNoSigma}{\ensuremath{\textsf{\upshape{fc-ACQ}}}}
\nc{\fcACQ}{\ensuremath{\textsf{\upshape{fc-ACQ}}[\sigma]}}
\nc{\fcACQOne}{\ensuremath{\textsf{\upshape{fc-ACQ}}[\sigmaOne]}}
\nc{\fcACQci}{\ensuremath{\textsf{\upshape{fc-ACQ}}[\cisigma]}}
\nc{\fcACQSimple}{\ensuremath{\textsf{\upshape{fc-ACQ}}[\sigmaSimple]}}
\nc{\fcACQBinary}{\ensuremath{\textsf{\upshape{fc-ACQ}}[\sigmaBinary]}}

\nc{\ciDSimple}{\ensuremath{{\widehat{D}_{\operatorname{col}}}}}

\nc{\AllDBs}[1]{\ensuremath{\textsf{DB}[#1]}}

\nc{\taskdescription}[1]{\texttt{#1}}
\nc{\booltask}{\taskdescription{bool}}
\nc{\counttask}{\taskdescription{count}}
\nc{\enumtask}{\taskdescription{enum}}

\nc{\IndexingProblem}[2]{\ensuremath{\textup{\textsc{IndexingEval}}(#1,#2)}}
\nc{\IndexingProblemGeneral}{\IndexingProblem{\sigma}{\QueryClass}}
\nc{\IndexingProblemOurs}{\IndexingProblem{\sigma}{\fcACQ}}

\nc{\valuation}{\ensuremath{\nu}}
\nc{\val}{\ensuremath{\valuation}}

\nc{\sigmaOne}{\ensuremath{\bar{\sigma}}}
\nc{\DOne}{\ensuremath{\bar{D}}}
\nc{\GOne}{\ensuremath{\bar{G}}}
\nc{\VOne}{\ensuremath{\bar{V}}}
\nc{\EOne}{\ensuremath{\bar{E}}}
\nc{\QOne}{\ensuremath{\bar{Q}}}
\nc{\vl}{\ensuremath{\textnormal{\textsf{nl}}}}
\nc{\el}{\ensuremath{\textnormal{\textsf{el}}}}
\nc{\VLabels}{\ensuremath{L_{\VOne}}}
\nc{\ELabels}{\ensuremath{L_{\EOne}}}
\nc{\MyPlus}{\ensuremath{\rightarrow}}
\nc{\MyMinus}{\ensuremath{\leftarrow}}

\nc{\sigmaBinary}{\ensuremath{\sigma'}}
\nc{\dbBinary}{\ensuremath{{D'}}}
\nc{\QBinary}{\ensuremath{Q'}}

\nc{\ArS}[1]{\ensuremath{\textsl{A}_{#1}}}
\nc{\StpS}[1]{\ensuremath{\textsl{S}_{#1}}}

\nc{\tup}[1]{{\ensuremath{\bar{#1}}}}
\nc{\tset}{\operatorname{set}}

\nc{\at}{\tup{a}}
\nc{\bt}{\tup{b}}
\nc{\ct}{\tup{c}}
\nc{\dt}{\tup{d}}

\nc{\tD}{{\ensuremath{\mathbf{D}}}}

\nc{\projection}{\tup{p}}
\nc{\Projections}{{\ensuremath{\mathbf{\Pi}}}}
\nc{\projectionAlt}{\tup{q}}

\nc{\myTuples}{\ensuremath{\mathbf{T}}}
\nc{\myTuplesSym}{\ensuremath{\mathbf{T}'}}

\nc{\sigmaSimple}{\ensuremath{{\widehat{\sigma}}}}
\nc{\dbSimple}{\ensuremath{{\widehat{D}}}}
\nc{\RSimple}{\ensuremath{E}}
\nc{\QSimple}{\ensuremath{\widehat{Q}}}
\nc{\CSimple}{\ensuremath{\widehat{C}}}
\nc{\nuSimple}{\ensuremath{\widehat{\nu}}}

\nc{\img}{\ensuremath{\textrm{img}}}

\nc{\Neighbors}[5]{\ensuremath{{#1}_{#2}^{#3}({#4},{#5})}}
\nc{\NSucc}[3]{\Neighbors{N}{\rightarrow}{#1}{#2}{#3}}
\nc{\NPred}[3]{\Neighbors{N}{\leftarrow}{#1}{#2}{#3}}
\nc{\hatNSucc}[3]{\Neighbors{\widehat{N}}{\rightarrow}{#1}{#2}{#3}}
\nc{\hatNPred}[3]{\Neighbors{\widehat{N}}{\leftarrow}{#1}{#2}{#3}}

\nc{\Numbers}[5]{\ensuremath{{#1}_{#2}^{#3}({#4},{#5})}}
\nc{\numSucc}[3]{\Numbers{\#}{\rightarrow}{#1}{#2}{#3}}
\nc{\numPred}[3]{\Numbers{\#}{\leftarrow}{#1}{#2}{#3}}
\nc{\numNeigh}[3]{\Numbers{\#}{d}{#1}{#2}{#3}}
\nc{\hatnumSucc}[3]{\Numbers{\widehat{\#}}{\rightarrow}{#1}{#2}{#3}}
\nc{\hatnumPred}[3]{\Numbers{\widehat{\#}}{\leftarrow}{#1}{#2}{#3}}

\nc{\Numb}[1]{\ensuremath{n_{#1}}}

\nc{\col}{\ensuremath{\textnormal{\textsf{col}}}}
\nc{\colAlt}{\col'}

\nc{\Dual}[1]{\ensuremath{\widetilde{#1}}}

\nc{\DSD}{\ensuremath{\textsf{\upshape{DS}}_D}}

\nc{\fcr}{\ensuremath{f_{\operatorname{cr}}}}

\nc{\cisigma}{\ensuremath{\sigma_{\operatorname{col}}}}
\nc{\ciD}{\ensuremath{{D_{\operatorname{col}}}}}
\nc{\ciDn}{\ensuremath{{D^n_{\operatorname{col}}}}}
\nc{\ciDStrich}{\ensuremath{{D'_{\operatorname{col}}}}}
\nc{\ciQ}{\ensuremath{{Q_{\operatorname{col}}}}}
\nc{\ciQStrich}{\ensuremath{{Q'_{\operatorname{col}}}}}

\nc{\elmt}[3]{{\ensuremath{e^{#1}_{({#2},{#3})}}}}

\nc{\myEnum}{\ensuremath{\textsc{Enum}}}

\nc{\Parent}{\ensuremath{\textit{p}}}
\nc{\Children}{\ensuremath{\textit{ch}}}

\nc{\DataStructure}[2]{\textsf{\upshape{DS}}_{D,Q}}

\nc{\myatom}{\ensuremath{\alpha}}

\nc{\Tuples}{\ensuremath{\textit{Tup}}}

\nc{\Plus}{\ensuremath{{\texttt{f}}}}
\nc{\Minus}{\ensuremath{{\texttt{b}}}}

\nc{\rootnode}{\ensuremath{\textit{root}}}

\nc{\vtup}[1]{\ensuremath{\bar{x}_{#1}}}

\nc{\singvarv}[1]{\ensuremath{\textsf{v}_{#1}}}
\nc{\singvarw}[1]{\ensuremath{\textsf{w}_{#1}}}

\nc{\Hom}{\ensuremath{\textup{Hom}}}

\nc{\tupind}[2]{\ensuremath{{#1}^{(#2)}}}
\nc{\atind}[1]{\tupind{\at}{#1}}
\nc{\btind}[1]{\tupind{\bt}{#1}}
\nc{\atOne}{\atind{1}}
\nc{\btOne}{\btind{1}}
\nc{\atI}{\atind{i}}
\nc{\btI}{\btind{i}}
\nc{\atIs}{\atind{i+1}}
\nc{\btIs}{\btind{i+1}}
\nc{\atIp}{\atind{i-1}}
\nc{\btIp}{\btind{i-1}}
\nc{\atJ}{\atind{j}}
\nc{\btJ}{\btind{j}}
\nc{\atM}{\atind{m}}
\nc{\btM}{\btind{m}}

\nc{\AOne}{\ensuremath{A^{(1)}}}
\nc{\BOne}{\ensuremath{B^{(1)}}}
\nc{\AJ}{\ensuremath{A^{(j)}}}
\nc{\BJ}{\ensuremath{B^{(j)}}}
\nc{\AI}{\ensuremath{A^{(i)}}}
\nc{\BI}{\ensuremath{B^{(i)}}}
\nc{\AIs}{\ensuremath{A^{(i+1)}}}
\nc{\BIs}{\ensuremath{B^{(i+1)}}}
\nc{\AIp}{\ensuremath{A^{(i-1)}}}
\nc{\BIp}{\ensuremath{B^{(i-1)}}}
\nc{\AM}{\ensuremath{A^{(m)}}}
\nc{\BM}{\ensuremath{B^{(m)}}}

\nc{\Nsl}[1]{\ensuremath{N'_{#1}}}

\renewcommand*{\mid}{\, : \,} %

\newcommand*{\fd}{f_{\downarrow}}

\newcommand{\numNc}[3]{\#_{#1}({#2}, {#3})}
\newcommand*{\numN}[2]{\numNc{}{#1}{#2}}
\newcommand*{\N}[2]{N({#1}, {#2})}

\newcommand{\wat}{w_{\at}}
\newcommand{\wbt}{w_{\bt}}

\newcommandx*{\mset}[3][1={},3={\,}]{\ensuremath{
	#1\{\!\!#1\{#3 {#2} #3#1\}\!\!#1\}
}}
\newcommand{\slices}{\mathcal{S}}
\newcommand*{\slice}{\tup{s}}
\newcommand*{\sliceAlt}{\tup{t}}

\newcommand*{\pslices}{\pi_{\slices}}
\newcommand*{\pproj}[1][\at, \bt]{\pi_{#1}}

\newcommand*{\stp}{\mathsf{stp}}

\newcommand{\hati}{\hat{\imath}}

\renewcommandx*{\set}[3][1={},3={\,}]{\ensuremath{#1\{#3 {#2} #3#1\}}}

\newcommand*{\rcrcols}{C_R}

\tikzset{nudge/.code args={#1}{%
  \pgfkeysalso{transform canvas={xshift=#1}}%
}} %

\title{Structural Indexing of Relational Databases for the Evaluation of Free-Connex Acyclic Conjunctive Queries}

\date{}

\author[1]{Cristian Riveros}
\author[2]{Benjamin Scheidt}
\author[2]{Nicole Schweikardt}
\hypersetup{
	pdfauthor={Cristian Riveros, Benjamin Scheidt and Nicole Schweikardt},
	pdfkeywords={
		query evaluation,
		index structures,
		constant-delay enumeration,
		counting,
		conjunctive queries,
		free-connex ACQs.
	}
}

\affil[1]{
	Pontificia Universidad Católica de Chile\\
	Chile\\
	\href{mailto:criveros@ing.puc.cl}{criveros@ing.puc.cl}
}
\affil[2]{
	Humboldt-Universität zu Berlin\\
	Germany\\
	$\{$
		\href{mailto:benjamin.scheidt@hu-berlin.de}{benjamin.scheidt}, 
		\href{mailto:schweikn@hu-berlin.de}{schweikn} 
	$\}$%
	@hu-berlin.de
}

\begin{document}

\maketitle

\begin{abstract}
	We present an index structure to boost the evaluation of free-connex acyclic conjunctive queries (fc-ACQs) over relational databases.
	The main ingredient of the index associated with a given database $D$ is an auxiliary database $\ciD$.
	Our main result states that for any fc-ACQ $Q$ over $D$, we can count the number of answers of $Q$ or enumerate them with constant delay after a preprocessing phase that takes time linear in the size of $\ciD$.

	Unlike previous indexing methods based on values or order (e.g., B+ trees), our index is based on structural symmetries among tuples in a database, and the size of $\ciD$ is related to the number of colors assigned to~$D$ by Scheidt and Schweikardt's \enquote{relational color refinement} (2025).
	In the particular case of graphs, this coincides with the minimal size of an equitable partition of the graph.
	For example, the size of $\ciD$ is logarithmic in the case of binary trees and constant for regular graphs.
	Even in the worst-case that $D$ has no structural symmetries among tuples at all, the size of $\ciD$ is still linear in the size of $D$.

	Given that the size of $\ciD$ is bounded by the size of $D$ and can be much smaller (even constant for some families of databases), our index is the first foundational result on indexing internal structural symmetries of a database to evaluate all fc-ACQs with performance potentially strictly smaller than the database size.
\end{abstract}

\paragraph*{Related version} This paper supersedes the preprint \href{https://arxiv.org/abs/2405.12358}{arXiv:2405.12358}~\cite{Riveros2024} by the same authors that only considered the special case of \emph{binary} schemas.

\section{Introduction}%
\label{sec:introduction}
An important part of database systems are \emph{index structures} that provide efficient access to the stored data and help to accelerate query evaluation.
Such index structures typically rely on hash tables or balanced trees such as B-trees or B$^+$-trees, which boost the performance of value queries~\cite{ramakrishnan2003database}.
Another recent example is indices for worst-case optimal join algorithms~\cite{ngo2018worst}.
For example, \emph{Leapfrog Triejoin}~\cite{DBLP:conf/icdt/Veldhuizen14}, a simple worst-case optimal algorithm for evaluating multiway-joins on relational databases, is based on so-called trie iterators for boosting the access under different join orders.
These trie indices have recently improved in the use of time and space~\cite{ArroyueloHNRRS21}.
Typically, index structures are not constructed for supporting the evaluation of a single query, but for supporting the evaluation of an entire class of queries.
This paper presents a novel kind of index structure to boost the evaluation of free-connex acyclic conjunctive queries (fc-ACQs).
Unlike previous indexing methods based on values or order, our index
is based on structural symmetries among tuples in a database.

Acyclic conjunctive queries (ACQs) were introduced in~\cite{DBLP:journals/jacm/BeeriFMY83,DBLP:journals/siamcomp/BernsteinG81} and have since then received a lot of attention in the database literature.
From Yannakakis' seminal paper~\cite{Yannakakis1981} it is known that the result of every fixed ACQ $Q$ over a database $D$ can be computed in time linear in the product of the database size $\dbsize{D}$ and the output size $|\sem{Q}(D)|$.
For the particular subclass fc-ACQ of \emph{free-connex} ACQs, it is even known to be linear in the sum $\dbsize{D}+|\sem{Q}(D)|$.
The notion of \emph{free-connex} ACQs was introduced in the seminal paper by Bagan, Durand, and Grandjean~\cite{Bagan.2007}, who improved Yannakakis' result as follows.
For any database $D$, during a preprocessing phase that takes time linear in $\dbsize{D}$, a data structure can be computed that allows to enumerate, without repetition, all the tuples of the query result $\sem{Q}(D)$, with only a \emph{constant} delay between outputting any two tuples.
Note that the above running time statement suppresses factors that depend on $Q$, and the data structure computed in the preprocessing phase is designed for the particular query $Q$.
To evaluate a new query $Q'$, one has to start a \emph{new} preprocessing phase that, again, takes time linear in $\dbsize{D}$.

Our main contribution is a new index structure that is based on the \emph{structural symmetries} among tuples in the database.
Upon input of a database $D$ of an arbitrary schema $\schema$, the index can be built in time $O(\dbsize{D}{\cdot}\log\dbsize{D})$ in the worst-case, and for many $D$ the time is only $O(\dbsize{D})$.
The main ingredient of our index is an auxiliary database $\ciD$.
Our main result states that for any fc-ACQ $Q$ over $D$, we can count the number of answers of $Q$ or enumerate them with constant delay after a preprocessing phase that takes time $O(\dbsize{\ciD})$.
Compared to the above-mentioned result by Bagan, Durand, and Grandjean~\cite{Bagan.2007}, this accelerates the preprocessing time from $O(\dbsize{D})$ to $O(\dbsize{\ciD})$.

The size of $\ciD$ is related to the number of colors assigned to~$D$ by Scheidt and Schweikardt's \emph{relational color refinement} \cite{ScheidtSchweikardt_MFCS25}.
In the particular case of graphs, this coincides with the minimal size of an \emph{equitable partition} of the graph.
For example, the size of $\ciD$ is logarithmic in the case of binary trees and constant for regular graphs.
Even in the worst-case that $D$ has no structural symmetries among tuples at all, the size of $\ciD$ is still linear in the size of $D$.
Given that $\dbsize{\ciD}$ is bounded by $\dbsize{D}$ and can be
much smaller (even constant for some families of databases), our index is the first foundational result on indexing internal structural symmetries of a database to evaluate all fc-ACQs with performance potentially strictly smaller than the database size. 

Proving our main result relies on two main ingredients. The first is to reduce the evaluation of fc-ACQs on databases $D$ over an arbitrary, fixed schema $\schema$ to the evaluation of fc-ACQs on node-labeled graphs.
We achieve this by showing that (1) $D$ can be translated into a node-labeled graph
$\dbSimple$ in linear time, (2) any fc-ACQ $Q$ over $D$ can be translated in linear time into a query $\QSimple$ over $\dbSimple$, and (3) there is a linear time computable bijection between the answer tuples of $\QSimple$ on $\dbSimple$ and the answer
tuples of $Q$ on $D$.
All this has to be carried out in such a way that $\QSimple$ is also \emph{free-connex acyclic} and, moreover, without introducing additional structural symmetries into
$\dbSimple$ that had not been present in the original database $D$ --- ensuring both is a major technical challenge.

The second main ingredient is to apply the well-known \emph{color refinement} algorithm (CR, for short) to the node-labeled graph $\dbSimple$.
CR is a simple and widely used subroutine for graph isomorphism testing algorithms (see e.g.~\cite{BBG-ColorRefinement,grohe_color_2021} for an overview and~\cite{CFI-paper,Arvind2017,Kiefer2020,Kiefer2021} for details on its expressibility).
Its result is a particular coloring of the vertex set of $\dbSimple$.
The construction of our index structure and, in particular, the auxiliary database $\ciD$ are based on this coloring.
Our result relies on a close connection between the colors computed by CR and the homomorphisms from ACQs to the database.
In recent years, this connection between colors and homomorphisms from tree-like structures has been successfully applied in different contexts~\cite{Dvorak2010,grohe_dimension_2014,Grohe2020a,Kayali2022,Bollen2023,Deeds2024,Goebel2024,ScheidtSchweikardt_MFCS25}.
Notions of index structures that are based on concepts of bisimulations
(which produce results similar to CR) and geared towards conjunctive query evaluation have been proposed and empirically evaluated, e.g., in~\cite{Picalausa2014,Picalausa2012}.
But to the best of our knowledge, the present paper is the first to use CR to produce an index structure that guarantees efficient \emph{constant-delay enumeration} and \emph{counting} and considers databases and fc-ACQs of \emph{arbitrary relational schemas}.

The rest of this paper is organized as follows.
Section~\ref{sec:preliminaries} fixes the basic notation concerning databases and queries, and it formalizes the general setting of \emph{indexing for query evaluation}.
Section~\ref{sec:ACQs} provides the necessary background on fc-ACQs and formally states our main theorem.
Section~\ref{sec:Reductions} reduces the problem from arbitrary schemas to node-labeled graphs.
Section~\ref{sec:OneBinaryRelation} describes our index structure and shows how it can be utilized to enumerate and count the results of fc-ACQs.
Section~\ref{sec:final} proves our main theorem, provides details on the size of $\ciD$, and points out directions for future research.
Many proof details have been deferred to an appendix. 

\section{Preliminaries and Formalization of Indexing for Query Evaluation}%
\label{sec:preliminaries}

We write $\NN$ for the set of non-negative integers, and we let $\NNpos \deff \NN \setminus \set{0}[]$.
For $m, n \in \NN$ we let $[m,n] \deff \setc{i \in \NN}{m \leq i\leq n}$ and $[n] \deff [1,n]$.

Whenever $G$ denotes a graph (directed or undirected),
we write $V(G)$ and $E(G)$ for the set of nodes and the set of edges of $G$.
Given a set $U \subseteq V(G)$, the subgraph of $G$ \emph{induced} by $U$
(for short: $G[U]$) is the graph $G'$ such that $V(G') = U$ and $E(G') = \set{ (u,v) \in E(G) \mid u,v \in U }$.
A \emph{connected component} of an undirected graph is a maximal connected subgraph of $G$.
A \emph{forest} is an undirected acyclic graph; and a \emph{tree} is a connected~forest.

We usually write $\at = (a_1, \ldots, a_r)$ to denote an $r$-tuple for some arity $r \in \NN$ and write $a_i$ to denote its $i$-th component (for $i \in [r]$).
Note that there is only one tuple of arity 0, namely, the \emph{empty tuple} denoted as $\emptytuple$.
Given a function $f\colon X \to Y$ and an $r$-tuple $\tup{x}$ of elements in $X$,
we write $f(\tup{x})$ for the $r$-tuple $(f({x}_1), \ldots, f({x}_r))$.
For $S \subseteq X^r$ we let $f(S) \isdef \setc{f(\tup{x})}{\tup{x} \in S}$.

A \emph{schema} $\sigma$ is a finite, non-empty set of relation symbols,
where each $R \in \sigma$ is equipped with a fixed arity $\ar(R) \in \NNpos$.
A schema $\sigma$ is called \emph{binary} if every $R \in \sigma$ has arity $\ar(R) \leq 2$.
A schema is called a \emph{schema for node-labeled graphs} 
if it consists of one binary relation symbol $E$ and, in addition to
that, a finite number of unary relation symbols.

We fix a countably infinite set $\Dom$ for the \emph{domain} of potential database entries, which we also call \emph{constants}.
A \emph{database} $D$ of schema $\sigma$ ($\sigma$-db, for short)
is a tuple of the form $D = (R^D)_{R \in \sigma}$, where $R^D$ is a finite subset of $\Dom^{\ar(R)}$.
The \emph{active domain} $\adom{D}$ of $D$ is the smallest subset $S$ of $\Dom$
such that $R^D \subseteq S^{\ar(R)}$ for all $R \in \sigma$.
The \emph{size} $\dbsize{D}$ of $D$ is defined as the total number of tuples in $D$, i.e.,
$\dbsize{D} = \sum_{R\in\sigma} |R^D|$.
A \emph{node-labeled graph} is a $\sigma$-db $D$, where $\sigma$ is a
schema for node-labeled graphs, and $E^D$ is \emph{symmetric},
i.e., for all tuples $(a,b) \in E^D$, also $(b,a) \in E^D$.

A $k$-ary \emph{query} (for $k \in \NN$) of schema $\sigma$ ($\sigma$-query, for short)
is a syntactic object $Q$ which is associated with a function $\sem{Q}$ that maps every $\sigma$-db $D$ to a finite subset of $\Dom^k$.
\emph{Boolean} (\emph{non-Boolean}) queries are $k$-ary queries for $k=0$ ($k \geq 1$).
Note that there are only two relations of arity 0, namely $\emptyset$ and $\set{\emptytuple}[]$.
For a Boolean query $Q$, we write $\sem{Q}(D) = \Yes$ to indicate that $\emptytuple \in \sem{Q}(D)$,
and we write $\sem{Q}(D)=\No$ to indicate that~$\sem{Q}(D)=\emptyset$.
\smallskip

\noindent
In this paper we will focus on the following evaluation tasks for a given $\sigma$-db $D$:
\begin{description}
	\item[Boolean query evaluation:]
	Upon input of a Boolean $\sigma$-query $Q$, decide if $\sem{Q}(D) = \Yes$;
	\item[Non-Boolean query evaluation:]
	Upon input of a $\sigma$-query $Q$, compute the relation $\sem{Q}(D)$;
	\item[Counting query evaluation:]
	Upon input of a $\sigma$-query $Q$, compute the number $|\sem{Q}(D)|$.
\end{description}
Concerning the second task, we are mainly interested in finding an \emph{enumeration algorithm} for computing the tuples in $\sem{Q}(D)$.
Such an algorithm consists of two phases: the \emph{preprocessing phase} and the \emph{enumeration phase}.
In the preprocessing phase, the algorithm is allowed to do arbitrary preprocessing to build a suitable data structure.
In the enumeration phase, the algorithm can use this data structure to enumerate all tuples in $\sem{Q}(D)$ followed by an end-of-enumeration message $\EOE$.
We require here that each tuple is enumerated exactly once (i.e., without repetitions).
The \emph{delay} is the maximum time that passes between the start of the enumeration phase and the first output,
between the output of two consecutive tuples, and between the last tuple and $\EOE$.

For our algorithms we use the RAM-model with a uniform cost measure.
In particular, storing and accessing elements of $\Dom$ requires $O(1)$ space and time.
This assumption implies that, for any $r$-ary relation $R^D$,
we can construct in time $O(r \cdot |R^D|)$ an index that allows to enumerate $R^D$ with $O(1)$ delay and to test for a given $r$-tuple $\tup{c}$ whether $\tup{c} \in R^D$ in time $O(r)$.
Furthermore, this implies that given any finite partial function $f\colon A \to B$,
we can build a \emph{lookup table} in time $O(|\dom{f}|)$,
where $\dom{f} \isdef \set{ x \in A \mid f(x) \text{ is defined} }$,
and have access to $f(a)$ in constant time.
\subsection*{Indexing for Query Evaluation}%
\label{sec:indexing}
We close Section~\ref{sec:preliminaries} by formalizing the setting of \emph{indexing for query evaluation}
for the tasks of Boolean ($\booltask$), non-Boolean ($\enumtask$), and counting ($\counttask$) query evaluation for a given class $\QueryClass$ of queries over a fixed schema $\sigma$.
We present here the general setting; later, we will instantiate it for a specific class of queries.
The scenario is as follows:
As input we receive a $\sigma$-db $D$.
We perform an \emph{indexing phase} in order to build a suitable data structure $\DSD$.
This data structure shall be helpful to efficiently evaluate \emph{any} query $Q \in \QueryClass$.
In the \emph{evaluation phase} we have access to $\DSD$.
As input, we receive arbitrary queries $Q \in \QueryClass$ and one of the three task descriptions $\booltask$, $\enumtask$, or $\counttask$, where $\booltask$ is only allowed in case that $Q$ is a Boolean query.
The goal is to solve this query evaluation task for $Q$ on $D$.

This scenario resembles what happens in real-world database systems,
where indexes are built to ensure efficient access to the information stored in the database,
and subsequently these indexes are used for evaluating various input queries.
We formalize the problem as:
\begin{center}
	\framebox{
	\begin{tabular}{rl}
		\textbf{Problem:}
			\!\!\!\!\!\!&
			$\IndexingProblemGeneral$
			\\\hline\vspace{-3mm}\\
		\textbf{Indexing:}
			\!\!\!\!\!\!&
			$\left\{\text{
				\begin{tabular}{rl}
					\textbf{input:}
						&\!\!\!
						a $\sigma$-db $D$\\
					\textbf{result:}
						&\!\!\!
						a data structure $\DSD$
				\end{tabular}
			}\right.$%
			\\[3.5mm]\hline\vspace{-3mm}\\
		\!\!
		\textbf{Evaluation:}
			\!\!\!\!\!\!&
			$\left\{\!\!\!
				\text{
				\begin{tabular}{rl}
					\textbf{input:}
						&\!\!\!
						\parbox[t]{.615\linewidth}{
							a $\sigma$-query $Q \in \QueryClass$, and a task
							description in $\set{ \booltask, \enumtask, \counttask }$
						}\\
					\textbf{output:}
						&\!\!\!
						\parbox[t]{.615\linewidth}{
							the correct answer solving the given task for $\sem{Q}(D)$
						}
				\end{tabular}
				}
			\right.$
			\!\!\!\!
	\end{tabular}
	}
\end{center}
\smallskip

The \emph{indexing time} is the time used for building the data structure $\DSD$;
it only depends on the input database $D$.
The time it takes to answer a Boolean query $Q$ on $D$ or for counting the number of result tuples of a query $Q$ on $D$ depends on~$Q$ and the particular properties of the data structure $\DSD$.
We measure these times by functions that provide an upper bound on the time taken to solve the task by utilizing the data structure $\DSD$.
Concerning the task $\enumtask$, we measure the preprocessing time and the delay by two functions that provide upper bounds on the time taken for preprocessing and the delay, respectively, when using the data structure $\DSD$ to enumerate $\sem{Q}(D)$.

Note that for measuring running times we use $\DSD$ (and not $\dbsize{\DSD}$ or $\dbsize{D}$) because we want to allow the running time analysis to be more fine-grained
than just depending on the size of $\DSD$ or $D$.
Our main result is a solution for $\IndexingProblemGeneral$ where $\QueryClass$ is the class $\fcACQ$ of all \emph{free-connex acyclic conjunctive queries} of an arbitrary schema $\sigma$.

\section{Free-Connex Acyclic CQs and  Formulation of this Paper's Main Theorem}%
\label{sec:ACQs}
Before presenting a formal statement of this paper's main theorem, we provide the necessary background concerning \emph{free-connex acyclic conjunctive queries} ($\fcACQNoSigma$).

We fix a countably infinite set $\Var$ of \emph{variables} such that $\Var \cap \Dom = \emptyset$.
An \emph{atom} $\myatom$ of schema $\sigma$ is of the form $R(x_1, \ldots, x_r)$ where $R \in \sigma$, $r=\ar(R)$, and $x_1, \ldots, x_r \in \Var$.
We write $\vars(\myatom)$ for the set of variables occurring in $\myatom$,
and we let $\ar(\myatom) \deff \ar(R)$ be the \emph{arity} of $\myatom$.
Let $\ell\in\NN$.
An $\ell$-ary \emph{conjunctive query} (CQ) of schema $\sigma$ is of the form 
\ $\Ans(z_1,\ldots,z_\ell) \leftarrow \myatom_1, \ldots, \myatom_d$, \
where $d \in \NNpos$, $\myatom_j$ is an atom of schema $\sigma$ for every $j \in [d]$,
and $z_1, \ldots, z_\ell$ are $\ell$ pairwise distinct variables in $\bigcup_{j \in [d]}\vars(\myatom_j)$.
The expression to the left (right) of $\leftarrow$ is called the \emph{head} (\emph{body}) of the query.
We let
$\atoms(Q) \isdef \smallset{\myatom_1,\ldots,\myatom_d}$,
$\vars(Q) \isdef \bigcup_{j \in [d]} \vars(\myatom_j)$,
and $\free(Q) \isdef \smallset{z_1,\ldots,z_\ell}$.
The \emph{(existentially) quantified} variables are the elements in
$\quant(Q) \deff \vars(Q) \setminus \free(Q)$.
A CQ $Q$ is called \emph{Boolean} if $\free(Q)=\emptyset$, and it is called \emph{full}
(or, \emph{quantifier-free}) if $\quant(Q) = \emptyset$.
The \emph{size} $|Q|$ of the query is defined as $|\atoms(Q)|$,
while $\querysize{Q}$ is defined to be the length of a reasonable representation of $Q$;
to be concrete we let $\querysize{Q}$ be the sum of the query's arity $\ell = |\free(Q)|$ and the sum of the arities of all the atoms of $Q$.
The semantics are defined as usual (cf.\ \cite{AHV-Book}):
A \emph{valuation} $\val$ for $Q$ is a mapping $\val\colon \vars(Q) \to \Dom$.
A \emph{homomorphism} from $Q$ to a $\sigma$-db $D$ is a valuation $\val$ for $Q$ such that for every atom $R(x_1, \ldots, x_r) \in \atoms(Q)$ we have $(\val(x_1), \ldots, \val(x_r)) \in R^D$.
We let $\Hom(Q,D)$ be the set of all homomorphisms from $Q$ to $D$.
The \emph{query result}  of a CQ $Q$ with head $\Ans(z_1,\ldots,z_\ell)$ on the $\sigma$-DB $D$ is defined as the set of tuples
$\sem{Q}(D) \isdef \set{(\val(z_1),\ldots,\val(z_\ell)) \mid \val \in\Hom(Q,D)}$.

The \emph{hypergraph} $H(Q)$ of a CQ $Q$ is defined as follows.
Its vertex set is $\vars(Q)$, and it contains a hyperedge $\vars(\myatom)$ for every $\myatom\in\atoms(Q)$.
The \emph{Gaifman graph} $G(Q)$ of $Q$ is the undirected simple graph with vertex set $\vars(Q)$, and it contains the edge $\smallset{x,y}$ whenever $x,y$ are distinct variables such that $x, y \in \vars(\myatom)$ for some $\myatom \in \atoms(Q)$.

\emph{Acyclic} CQs and \emph{free-connex acyclic} CQs are standard notions studied in
the database theory literature
(cf.~\cite{DBLP:journals/jacm/BeeriFMY83,
	DBLP:journals/siamcomp/BernsteinG81,
	DBLP:journals/jcss/GottlobLS02,
	Bagan.2007,
	AHV-Book};
see~\cite{BGS-tutorial} for an overview).
A CQ $Q$ is called \emph{acyclic} if its hypergraph $H(Q)$ is \emph{$\alpha$-acyclic}, i.e.,
there exists an undirected tree $T = (V(T), E(T))$ (called a \emph{join-tree} of $H(Q)$ and of $Q$)
whose set of nodes $V(T)$ is precisely the set of hyperedges of $H(Q)$,
and where for each variable $x \in \vars(Q)$ the set $\set{ t \in V(T) \mid x \in t }$ induces a connected subtree of $T$.
A CQ $Q$ is \emph{free-connex acyclic} if it is acyclic \emph{and}
the hypergraph obtained from $H(Q)$ by adding the hyperedge $\free(Q)$ is $\alpha$-acyclic.
Note that any CQ $Q$ that is either Boolean or full is free-connex acyclic iff it is acyclic.
However, $\Ans(x,z) \leftarrow R(x,y), R(y,z)$ is an example of a query that is acyclic,
but not free-connex acyclic.
For the special case of \emph{binary} schemas, there is a particularly simple characterization of (free-connex) acyclic CQs (see Appendix~\ref{app:proof:prop:binaryfcacqs} for a proof):
\begin{restatable}[Folklore]{proposition}{propBinaryFcAcqs}\label{prop:binaryfcacqs}
	A CQ $Q$ of a \emph{binary} schema $\sigma$ is \emph{acyclic} iff its Gaifman graph $G(Q)$ is acyclic.
	The CQ $Q$ is \emph{free-connex acyclic} if, and only if, $G(Q)$ is acyclic and the following statement is true:\
	for every connected component $C$ of $G(Q)$,
	either $\free(Q) \cap V(C)=\emptyset$ or the subgraph of $C$ induced by the set $\free(Q) \cap V(C)$ is connected.
\end{restatable}
We write $\fcACQ$ to denote the set of all free-connex acyclic CQs of
schema $\sigma$.
Note that Proposition~\ref{prop:binaryfcacqs} does not generalize to arbitrary, non-binary schemas $\sigma$
(see Appendix~\ref{appendix:BinaryToArbitraryIsNontrivial} for an example of a query $Q\in\fcACQ$ whose Gaifman graph is not acyclic).

In the following, we discuss an important result that will be crucial for the main result of this paper.
Yannakakis' seminal result \cite{Yannakakis1981} states that Boolean ACQs $Q$ can be evaluated in time $O(\dbsize{D})$.
Bagan, Durand, and Grandjean's seminal paper~\cite{Bagan.2007} showed that for any (non-Boolean) fc-ACQ $Q$,
the set $\sem{Q}(D)$ can be enumerated with constant delay after $O(\dbsize{D})$ preprocessing time.
The above statements refer to data complexity, i.e., running time components that depend on the query are hidden in the O-notation.
Several proofs of (and different algorithms for) Bagan, Durand and Grandjean's theorem
are available in the literature~\cite{Bagan.2007,
	Bagan_PhD,
	BraultBaron_PhD,
	DBLP:journals/tods/OlteanuZ15,
	DynamicYannakakis2017,
	DBLP:journals/pvldb/IdrisUVVL18,
	DBLP:journals/sigmod/IdrisUVVL19};
all of them focus on data complexity.
For this paper we need a more refined statement that takes into account the combined complexity of the problem, which is implicit in~\cite{BGS-tutorial}~(see Appendix~\ref{app:proof:thm:BGS-enum} for details).
\begin{restatable}{theorem}{thmBGSenum}\label{thm:BGS-enum}
	For every schema $\sigma$ there is an enumeration algorithm that receives as input
	a $\sigma$-db $D$ and a query $Q\in\fcACQ$ and that computes within preprocessing time
	$O(|Q|{\cdot}\dbsize{D})$ a data structure for enumerating $\sem{Q}(D)$ with delay $O(|\free(Q)|)$.
\end{restatable}
Our main result provides a solution for the problem $\IndexingProblem{\sigma}{\fcACQ}$ for any schema~$\sigma$.
The data structure $\DSD$ that we build for a given database $D$ during the indexing phase will provide a new, auxiliary database $\ciD$, which is potentially much smaller than $D$.
It will allow us to improve the preprocessing time provided by Theorem~\ref{thm:BGS-enum} to $O(|Q|{\cdot}\dbsize{\ciD})$.
Specifically, the following is the main theorem of the paper.
\begin{theorem}\label{thm:maintheorem}
	For every schema $\sigma$, the problem $\IndexingProblemOurs$ can be solved with indexing time $O(|D|\cdot\log|D|)$ constructing a new database $\ciD$,
	such that afterwards, every Boolean acyclic query $Q$ posed against $D$ can be answered in time $O(|Q|{\cdot}|\ciD|)$.
	Furthermore, for every query $Q \in \fcACQ$  we can enumerate the tuples in $\sem{Q}(D)$ with delay $O(|\free(Q)|)$ after preprocessing time $O(|Q|{\cdot}|\ciD|)$,
	and we can compute the number $|\sem{Q}(D)|$ of result tuples in time $O(|Q|{\cdot}|\ciD|)$.
\end{theorem}

The rest of this paper is devoted to proving Theorem~\ref{thm:maintheorem} and to also providing insights in the size of $\ciD$.
In Section~\ref{sec:Reductions}, we reduce the problem from arbitrary schemas $\sigma$ to schemas $\sigmaSimple$ for node-labeled graphs.
In Section~\ref{sec:OneBinaryRelation} we describe our index structure and show how to utilize it to enumerate and count the results of fc-ACQs posed against node-labeled graphs.
In Section~\ref{sec:final} we combine the results of the previous two sections into the proof of Theorem~\ref{thm:maintheorem},
and we also provide details on the size of $\ciD$, and how it relates to the internal structural symmetries of the original $\sigma$-database.

\section{Reducing the Problem from Arbitrary Schemas to Node-Labeled Graphs}%
\label{sec:Reductions}
This section reduces the evaluation of fc-ACQs on databases over an arbitrary, fixed schema $\sigma$ to the evaluation of fc-ACQs on node-labeled graphs,
while ensuring that the translation is conducted in linear time, does not introduce new symmetries into the database, and preserves free-connex acyclicity of the queries.
We proceed it two steps: first, from arbitrary schemas to binary schemas (Section~\ref{sec:ArbitrarySchemas}),
and afterwards from binary schemas to node-labeled graphs (Section~\ref{sec:ReductionToOneBinaryRelation}).
\subsection{From Arbitrary Schemas to Binary Schemas}%
\label{sec:ArbitrarySchemas}

This subsection is devoted to proving the following theorem.
\begin{theorem}\label{thm:ReductionArbSchemaToBinarySchema}
	For any arbitrary schema $\sigma$ with $k \deff \ar(\sigma)$, there exists a binary schema $\sigmaBinary$ of size $|\sigma| + 2k^2 + k + 1$, such that the following is true:
	\begin{enumerate}[(1)]
		\item\label{item:one:ReductionArbToBinary}
		upon input of a $\sigma$-db $D$, we can compute in time $2^{\bigOh(k\log k)}{\cdot} \dbsize{D}$ a $\sigmaBinary$-db $\dbBinary$,
		\item\label{item:two:ReductionArbToBinary}
		upon input of any query $\query\in\fcACQ$, we can compute in time $\bigOh(\size{Q})$ a query $\QBinary \in \fcACQBinary$ with $|\free(\QBinary)| < 2 {\cdot} |\free(\query)|$, such that
		\item\label{item:three:ReductionArbToBinary}
		there is a bijection $f\colon \smallsem{\QBinary}(\dbBinary) \to \smallsem{\query}(D)$.
		Furthermore, when given a tuple $\at \in \smallsem{\QBinary}(\dbBinary)$, the tuple $f(\at) \in \smallsem{\query}(D)$ can be computed in time $\bigOh(|\free(\query)| \cdot k)$.
	\end{enumerate}
\end{theorem}
\noindent
When adopting the same view as in the formulation of Theorem~\ref{thm:BGS-enum}, the schema $\sigma$ is fixed and the value $k = \ar(\sigma)$ is subsumed in the $\bigOh$-notation.
Thus, in Theorem~\ref{thm:ReductionArbSchemaToBinarySchema}, the running times simplify to $\bigOh(\dbsize{D})$ in statement~\ref{item:one:ReductionArbToBinary} and $\bigOh(|\free(\query)|)$ in statement~\ref{item:three:ReductionArbToBinary}.

The rest of this subsection is dedicated to presenting the main ingredients of the proof of Theorem~\ref{thm:ReductionArbSchemaToBinarySchema}
(see Appendix~\ref{appendix:FromArbitraryToBinary} for the missing details).
We first pick the binary schema $\sigmaBinary$ and show how to build the database $D'$ from $D$
(statement \ref{item:one:ReductionArbToBinary}).
Afterwards, we show how to construct the fc-ACQ $Q'$ from a given fc-ACQ $Q$
(statement \ref{item:two:ReductionArbToBinary}).
Finally, we present the bijection $f$ from the outputs in $\smallsem{\QBinary}(\dbBinary)$ to the outputs in $\smallsem{\query}(D)$ and sketch its correctness
(statement \ref{item:three:ReductionArbToBinary}).

We use the following notation.
For any $r \in \NN$ and any $r$-tuple $\at = (a_1,\ldots,a_r)$ we let $\tset(\at) \deff \smallset{a_1, \ldots, a_r}$, and for $i \in [r]$ we let $\proj_i(\bar{a}) \deff a_i$.
For $r \in \NNpos$, every $r$-tuple $\at = (a_1, \ldots, a_r)$, every $m \in [0,r]$ and every tuple $(j_1, \ldots, j_m)$ of pairwise distinct elements $j_1, \ldots, j_m \in [r]$, we let $\projgen_{(j_1,\ldots,j_m)}(\at) \deff (a_{j_1},\ldots,a_{j_m})$, and we call this tuple \emph{a projection of $\at$}.
In particular, for $m = 0$ and the empty tuple $()$ we have $\projgen_{()}(\at) = ()$;
and for $i \in [r]$ we have $\projgen_{(i)}(\at) = (a_i) = (\proj_i(\at))$.
We let $\Projections(\at)$ be the set of all projections of $\at$, i.e., 
\begin{equation*}
	\Projections(\at) \;\;=\;\; \set[\big]{ 
		\projgen_{(j_1,\ldots,j_m)}(\at) \mid
			m \in [0,r],\; j_1, \ldots, j_m\ \text{are pairwise distinct elements in} [r] }.
\end{equation*}
For a $\sigma$-db $D$, we let $\tD \deff \bigcup_{R \in \sigma} R^{D}$ be the set of tuples that occur in some relation of $D$.
We let $\Projections(\tD) \deff \bigcup_{\dt \in \tD} \Projections(\dt)$ be the set of projections of tuples present in $D$.

\paragraph{\ref{item:one:ReductionArbToBinary} Choosing $\sigmaBinary$ and Constructing the $\sigmaBinary$-db $\dbBinary$.}
Let $\sigma$ be an arbitrary schema, and let $k \deff \ar(\sigma) = \max\setc{\ar(R)}{R \in \sigma}$.
Let $\sigmaBinary$ consist of unary  symbols $U_R$ for all $R \in \sigma$, unary symbols $\ArS{i}$ for all $i \in [0,k]$, and binary symbols $E_{i,j}$ and $F_{i,j}$ for all $i, j \in [k]$.
Clearly, $|\sigmaBinary| = |\sigma| + 2k^2 + k + 1$.

Now, let $D$ be an arbitrary $\sigma$-db.
Our $\sigmaBinary$-db $\dbBinary$ that represents $D$ is defined as follows.
For each $\dt \in \tD$ we introduce a new node $w_{\dt}$, and for each $\projection \in \Projections(\tD)$ we introduce a new node $v_{\projection}$.
We let
\begin{align*}
	(U_R)^{\dbBinary} \;\;&\deff\;\;
		\setc{w_{\at}}{\at\in R^D},\quad
		\text{for every $R \in \sigma$},\\
	(\ArS{i})^{\dbBinary} \;\;&\deff\;\;
		\setc{v_{\projection}}{\projection \in \Projections(\tD),\ \ar(\projection) = i},\quad
		\text{for every $i \in [0,k]$},
\end{align*}
and for all $i, j \in [k]$ we let
\begin{align*}
	(E_{i,j})^{\dbBinary} \;\;&\deff\;\;
		\set[\big]{
			(w_{\dt}, v_{\projection})
			\mid
			\dt \in \tD,\
			\projection \in \Projections(\dt),\
			i \leq \ar(\dt),\
			j \leq \ar(\projection),\
			\proj_i(\dt) = \proj_j(\projection)},\\
	(F_{i,j})^{\dbBinary} \;\;&\deff\;\;
		\bigl\{\,
			\makebox[\widthof{$(w_{\dt}, v_{\projection})$}][l]{$(v_{\projection},v_{\projectionAlt})$}
			\mid
			\begin{aligned}[t]
				&\projection, \projectionAlt \in \Projections(\tD),\
				i \leq \ar(\projection),\
				j \leq \ar(\projectionAlt),\
				\proj_i(\projection) = \proj_j(\tup{q}), \text{ and } \\
				&\bigl(
					\tset(\projection) \subseteq \tset(\projectionAlt)
					\;\ \text{or}\;\
					\tset(\projection) \supseteq \tset(\projectionAlt)
				\bigr)
				\,\bigr\}.
			\end{aligned}
\end{align*}
Intuitively, in the new database $\dbBinary$, we represent each $R$-tuple $\at$ by a node $w_{\at}$ in $(U_R)^{\dbBinary}$ and any projection $\projection$ of arity $i$ by a node $v_{\projection}$ in $(A_i)^{\dbBinary}$.
To relate $w_{\at}$ with $v_{\projection}$ we use the binary relation $(E_{i,j})^{\dbBinary}$ that models that the $i$-component of $\at$ is equal to the $j$-component of $\projection$.
Finally, $(F_{i,j})^{\dbBinary}$ relates equal values between projections of tuples in $D$ whenever their domains are contained.

\begin{restatable}{claim}{constructDbBinary}\label{claim:Construct-dbBinary}
Upon input of a $\sigma$-db $D$, the $\sigmaBinary$-db $\dbBinary$ can be constructed in time $2^{O(k\cdot \log k)} \cdot \dbsize{D}$.
\end{restatable}
\begin{proof}[Proof sketch]
	Note that for every $\at\in\tD$ with $r\deff \ar(\at)$ we have
	\(
		|\Projections(\at)|
		\ \leq \
		\sum_{m = 0}^{r} {r \choose m } \cdot m!
		\ \leq \ 
		r \cdot r!
		\ \leq \
		k \cdot k!\,.
	\)
	Thus, $|\Projections(\tD)| \leq k \cdot k! \cdot |\tD| \leq k \cdot k! \cdot \dbsize{D} = 2^{O(k\cdot \log k)} \cdot \dbsize{D}$.
	The condition in the second line of the definition of $(F_{i,j})^{\dbBinary}$ is crucial in order to guarantee that we can indeed construct this relation in time $2^{O(k\cdot \log k)}{\cdot} \dbsize{D}$
	(when omitting this condition, we would end up with a factor $\dbsize{D}^2$ instead of~$\dbsize{D}$).
	Details on the (brute-force) construction of $\dbBinary$ can be found in Appendix~\ref{appendix:ReductionFromArbitraryToBinary}.
\end{proof}

\paragraph{\ref{item:two:ReductionArbToBinary} Constructing the fc-ACQ $\QBinary$.}
Our next aim is to translate queries $\query \in \fcACQ$ into suitable $\sigmaBinary$-queries $\QBinary$.
We want $\QBinary$ to be in $\fcACQBinary$, and we want to ensure that there is an easy to compute bijection $f$ that maps the tuples in $\bigsem{\QBinary}(\dbBinary)$ onto the tuples in $\sem{\query}(D)$.
The main challenge here is to define $\QBinary$ in such a way that it indeed is \emph{free-connex acyclic}
(cf.\ Appendix~\ref{appendix:BinaryToArbitraryIsNontrivial}).

Our translation is based on the well-known characterization of the free-connex acyclic queries via the following notion.
A \emph{free-connex generalized hypertree decomposition of width~1} (fc-1-GHD, for short) of a CQ $\query$ is a tuple $H = (T, \Bag, \Cover, \Witness)$ such that
\begin{enumerate}[(1)]
	\item
	$T = (V(T), E(T))$ is a finite undirected tree,
	\item
	$\Bag$ is a mapping that associates with every $t \in V(T)$ a set $\Bag(t) \subseteq \Vars(\query)$ such that
	\begin{enumerate}[(a), topsep=0pt]
		\item
		for each atom $\Atom \in \Atoms(\query)$ there exists a $t \in V(T)$ such that $\Vars(\Atom) \subseteq \Bag(t)$, and
		\item
		for each variable $y \in \Vars(\query)$ the set $\Bag^{-1}(y) \deff \setc{t \in V(T)}{y \in \Bag(t)}$ induces a connected subtree of $T$
		(this condition is called \emph{path condition}),
	\end{enumerate}
	\item
	$\Cover$ is a mapping that associates with every $t \in V(T)$ an atom $\Cover(t) \in \Atoms(\query)$ such that $\Bag(t) \subseteq \Vars(\Cover(t))$,
	\item
	$\Witness \subseteq V(T)$ such that $\Witness$ induces a connected subtree of $T$, and $\free(\query) = \bigcup_{t \in \Witness} \Bag(t)$.
	The set $\Witness$ is called a \emph{witness} for the free-connexness of $H$.
\end{enumerate}
By $\size{H}$ we denote the size of a reasonable representation of $H$.
An fc-1-GHD $H$ is called \emph{complete} if for every $\Atom \in \Atoms(\query)$ there is a $t \in V(T)$ such that $\Vars(\Atom) = \Bag(t)$ and $\Atom = \Cover(t)$.

\begin{restatable}{proposition}{propComputeGHD}\label{prop:ComputeGHD}
	For every $\query \in \fcACQ$, in time $\bigOh(\size{\query})$ one can compute a complete fc-1-GHD $H = (T, \Bag, \Cover, \Witness)$ of $\query$ such that $|\Witness| < 2 {\cdot} |\free(\query)|$ and for all edges $\smallset{t,p} \in E(T)$ we have $\Bag(t) \subseteq \Bag(p)$ or $\Bag(t) \supseteq \Bag(p)$.
\end{restatable}
\begin{proof}[Proof sketch]
	Using a result of Bagan \cite{Bagan_PhD} and then performing standard modifications, we can construct in time $\bigOh(\size{\query})$ a complete fc-1-GHD $H = (T, \Bag, \Cover, \Witness)$ of $\query$ with $|\Witness| \leq |\free(\query)|$.
	We subdivide every edge $\smallset{t,p}$ of $T$ that violates the condition
	\enquote{$\Bag(t) \subseteq \Bag(p)$ or $\Bag(t) \supseteq \Bag(p)$}
	by introducing a new node $n_{\smallset{t,p}}$, letting $\Bag(n_{\smallset{t,p}}) \deff \Bag(t) \cap \Bag(p)$ and
	$\Cover(n_{\smallset{t,p}}) \deff \Cover(t)$ and, in case that $\smallset{t,p} \subseteq \Witness$, inserting $n_{\smallset{t,p}}$ into $\Witness$.
	See Appendix~\ref{appendix:ReductionFromArbitraryToBinary} for details.
\end{proof}

Given a query $\query \in \fcACQ$, we use Proposition~\ref{prop:ComputeGHD} to compute a complete fc-1-GHD $H = (V, \Bag, \Cover, \Witness)$ of $\query$ such that $|\Witness| < 2 {\cdot} |\free(\query)|$ and for all edges $\smallset{t,p} \in E(T)$ we have $\Bag(t) \subseteq \Bag(p)$ or $\Bag(t) \supseteq \Bag(p)$.
We fix an arbitrary order $<$ on $\Vars(\query)$.
For every $t \in V(T)$ we let $\vtup{t}$ be the $<$-ordered tuple formed from the elements of $\Bag(t)$
(i.e., $\ar(\vtup{t}) = |\Bag(t)|$ and $\tset(\vtup{t}) = \Bag(t)$).
For every atom $\myatom \in \Atoms(\query)$, we fix a node $t_\myatom \in V(T)$ such that $\Vars(\myatom) = \Bag(t_\myatom)$ and  $\myatom = \Cover(t_\myatom)$
(such a node $t_\myatom$ exists because $H$ is complete),
and we let $\myAtoms(T) \deff \setc{t_\myatom}{\myatom \in \Atoms(\query)}$.
Finally, we fix an arbitrary list $t_1, \ldots, t_{|\Witness|}$ of all nodes in $\Witness$
(this list is empty if $W = \emptyset$).
If $\Witness \neq \emptyset$, choose $\rootnode$ to be an arbitrary node in $\Witness$;
otherwise $\Witness = \emptyset$ and $\free(\query) = \emptyset$, and we let $\rootnode$ be an arbitrary node of $T$.
Let $\rootedTree$ be the oriented version of $T$ where $\rootnode$ is the root of $T$.

For every $t \in V(T)$, the query $\QBinary$ uses a new variable $\singvarv{t}$, and if $t \in \myAtoms(T)$, it additionally uses a new variable $\singvarw{t}$.
The $\sigmaBinary$-query $\QBinary$ is defined as $\Ans(\singvarv{t_1}, \ldots, \singvarv{t_{|\Witness|}}) \leftarrow \Und_{\Atom \in \Atoms(\QBinary)} \Atom$, where $\Atoms(\QBinary)$ is constructed as follows.
We initialize $\Atoms(\QBinary)$ to be the empty set $\emptyset$, and for every node $t \in V(T)$ we insert into $\Atoms(\QBinary)$ the atom
\begin{itemize}
	\item $\ArS{|\Bag(t)|}(\singvarv{t})$.
\end{itemize}
In case that $t \in \myAtoms(T)$, let $\myatom$ be the particular atom of $Q$ such that $t = t_\myatom$, let $R(\tup{z}) \deff \myatom$, and insert into $\Atoms(\QBinary)$ the additional atoms
\begin{itemize}
	\item $U_R(\singvarw{t})$, and
	\item
	$E_{i,j}(\singvarw{t},\singvarv{t})$ for all those $i, j \in [k]$ such that $i \leq \ar(\tup{z})$, $j \leq \ar(\vtup{t})$, $\proj_i(\tup{z}) = \proj_j(\vtup{t})$.
\end{itemize}
Furthermore, for arbitrary $t \in V(T)$, in case that $t$ is \emph{not} the root node of $\rootedTree$, let $p$ be the parent of $t$ in $\rootedTree$ and insert into $\Atoms(\QBinary)$ also the atoms
\begin{itemize}
	\item
	$F_{i,j}(\singvarv{t},\singvarv{p})$ for all those $i, j \in [k]$ with $i \leq \ar(\vtup{t})$, $j \leq \ar(\vtup{p})$ where $\proj_i(\vtup{t}) = \proj_j(\vtup{p})$.
\end{itemize}
This completes the construction of $\QBinary$.
\begin{restatable}{claim}{claimConstructQBinary}\label{claim:ConstructQBinary}
	$\QBinary \in \fcACQBinary$, and given $H$, the query $\QBinary$ can be constructed in time $\bigOh(\size{H})$.
\end{restatable}
\begin{proof}[Proof sketch]
	Achieving the claimed running time is obvious.
	To show that $\QBinary \in \fcACQBinary$, we modify the tree $T$ of $H$ by renaming every node $t \in V(T)$ into $\singvarv{t}$, and for every $t \in \myAtoms(T)$ by adding to $\singvarv{t}$ a new leaf node called $\singvarw{t}$.
	It is easy to see that the resulting tree yields exactly the Gaifman graph $G(\QBinary)$ of $\QBinary$, when deleting those edges $\smallset{\singvarv{t},\singvarv{t'}}$ where $\Bag(t) \cap \Bag(t') = \emptyset$.
	From Proposition~\ref{prop:binaryfcacqs} we then obtain that $\QBinary$ is an acyclic query;
	and since $\free(\QBinary) = \setc{\singvarv{t}}{t \in \Witness}$, it can easily be verified that $\QBinary \in \fcACQBinary$.
	See Appendix~\ref{appendix:ReductionFromArbitraryToBinary} for details.
\end{proof}

\paragraph{\ref{item:three:ReductionArbToBinary} The Bijection $f$ Between Outputs.}
Our definition of the $\sigmaBinary$-db $\dbBinary$ is similar to the definition of the binary structure $\mathcal{H}_{D}$ for the $\sigma$-db $D$ provided in \cite[Definition~3.4]{ScheidtSchweikardt_MFCS25} ---
however, with subtle differences that are crucial for our proof of Theorem~\ref{thm:ReductionArbSchemaToBinarySchema}:
We use all the projections $\projection \in \Projections(D)$, while \cite{ScheidtSchweikardt_MFCS25} only uses \enquote{slices}
(i.e., projections $\projection$ where $\setsize{\tset(\projection)} = \ar(\projection) \neq 0$),
and we use additional unary relations $\ArS{i}$, for $i \in [0,k]$, to label the nodes $v_{\projection}$ associated with projections $\projection$ of arity $i$.
This enables us to assign a variable $\singvarv{t}$ of $\QBinary$ with a node $v_{\projection}$ of $\dbBinary$, ensure that $\projection$ has the same arity as the
variable tuple $\vtup{t}$, and assign the $\ell$-th variable in the tuple $\vtup{t}$ with the $\ell$-th entry of the tuple $\projection$ (for all $\ell$).
Furthermore, we use additional binary relations $F_{i,j}$ (for $i, j \in [k]$) which enable us to ensure consistency between these assignments when considering different nodes $t,t'$ of $T$.
Below, we show that this induces a bijection $\beta$ between $\Hom(\QBinary,\dbBinary)$ and $\Hom(\query,\DB)$ and also a bijection $f$ between $\smallsem{\QBinary}(\dbBinary)$ and $\smallsem{\query}(D)$.

We fix the following notation.
For every  $y \in \Vars(\query)$ choose an arbitrary node $t_y \in V(T)$ with $y \in \Bag(t_y)$ and, moreover, if $y \in \free(\query)$ then $t_y \in \Witness$
(this is possible because $H$ is an fc-1-GHD of $Q$).
Let $j_y \in [|\Bag(t_y)|]$ such that $y = \proj_{j_y}(\vtup{t_y})$
(i.e., $y$ occurs as the $j_y$-th entry of the tuple $\vtup{t_y}$).
For the remainder of this proof, these items $t_y$ and $j_y$ will remain fixed for each $y \in \Vars(\query)$.

Now consider an arbitrary homomorphism $h \in \Hom(\QBinary,\dbBinary)$, and let $\beta(h)\colon \Vars(\query) \to \Adom(D)$ be defined as follows.
Consider an arbitrary $y \in \Vars(\query)$ and note that $\ArS{|\Bag(t_y)|}(\singvarv{t_y}) \in \Atoms(\QBinary)$.
From $h \in \Hom(\QBinary,\dbBinary)$, we thus obtain that $h(\singvarv{t_y}) \in (\ArS{|\Bag(t_y)|})^{\dbBinary}$.
According to the definition of $\dbBinary$, there is a (unique) tuple $\projection_{h,y} \in \Projections(\tD)$ such that $h(\singvarv{t_y}) = v_{\projection_{h,y}}$ and, furthermore, $\ar(\projection_{h,y}) = |\Bag(t_y)|$.
We define $\beta(h)(y) \deff \proj_{j_y}(\projection_{h,y})$, i.e., $\beta(h)$ is defined to map the variable $y$ to the $j_y$-th entry of the tuple $\projection_{h,y}$.
Clearly, this defines a mapping $\beta(h)\colon \Vars(\query) \to \Adom(D)$.
\begin{restatable}{claim}{claimPropertiesOfBeta}\label{claim:PropertiesOfBeta}~
	\begin{enumerate}[(a)]
		\item\label{item:BetahIsConsistent}
		Let $h \in \Hom(\QBinary,\dbBinary)$ and $h' \deff \beta(h)$.
		For all $t \in V(T)$, we have $h(\singvarv{t}) = v_{h'(\vtup{t})}$ and, for every $t \in \myAtoms(T)$ with $R(\tup{z})\deff \Cover(t)$, we have $h(\singvarw{t}) = w_{h'(\tup{z})}$ and $h'(\tup{z}) \in R^D$.
		\item\label{item:BetahIsAHomomorphism}
		For all $h \in \Hom(\QBinary, \dbBinary)$ we have:\ $\beta(h) \in \Hom(\query,D)$.
		\item\label{item:betaIsSufficientlyInjective}
		For all $h_1, h_2 \in \Hom(\QBinary,\dbBinary)$ and all $t \in V(T)$, the following is true:

		If $h_1(\singvarv{t}) \neq h_2(\singvarv{t})$, then there is a variable $y \in \Bag(t)$ such that $\beta(h_1)(y) \neq \beta(h_2)(y)$.

		If $t \in \myAtoms(T)$ and $h_1(\singvarw{t}) \neq h_2(\singvarw{t})$, then there is a $y \in \Bag(t)$ such that $\beta(h_1)(y) \neq \beta(h_2)(y)$.
		\item\label{item:betaIsSurjective}
		The mapping $\beta\colon \Hom(\QBinary,\dbBinary) \to \Hom(\query,D)$ is bijective.\endClaim
	\end{enumerate}
\end{restatable}
\begin{proof}[Proof sketch]
	\ref{item:BetahIsConsistent} is shown by closely inspecting $\dbBinary$, $\QBinary$ and the particular fc-1-GHD $H$.

	\smallskip\noindent
	\ref{item:BetahIsAHomomorphism} and \ref{item:betaIsSufficientlyInjective} easily follow from \ref{item:BetahIsConsistent} (and \ref{item:BetahIsAHomomorphism} relies on the fc-1-GHD $H$ being \emph{complete}).

	\smallskip\noindent
	For \ref{item:betaIsSurjective}, the injectivity of $\beta$ immediately follows from \ref{item:betaIsSufficientlyInjective}.
	For proving that $\beta$ is surjective, consider an arbitrary $h'' \in \Hom(\query,D)$.
	We have to find an $h \in \Hom(\QBinary,\dbBinary)$ with $h'' = \beta(h)$.
	Based on $h''$, we define a mapping $h\colon \Vars(\QBinary) \to \Adom(\dbBinary)$ as follows.
	For every $t \in V(T)$, let $R_t(\tup{z}_t) \deff \Cover(t)$.
	Since $h'' \in \Hom(\query, D)$ and $R_t(\tup{z}_t) \in \Atoms(\query)$, we have $\at_t \deff h''(\tup{z}_t) \in (R_t)^D \subseteq \tD$.
	From $\tset(\tup{z}_t) \supseteq \Bag(t)$ we obtain that $h''(\vtup{t}) \in \Projections(\at_t) \subseteq \Projections(\tD)$.
	We let $h(\singvarv{t}) \deff v_{h''(\vtup{t})}$, for every $t \in V(T)$, and $h(\singvarw{t}) \deff w_{h''(\tup{z}_t)}$, for every $t \in \myAtoms(T)$.
	This defines a mapping $h\colon \Vars(\QBinary) \to \Adom(\dbBinary)$.
	By a close inspection of the atoms of $\QBinary$ we can show that indeed $h \in \Hom(\QBinary,\dbBinary)$
	(for this, we rely on the fact that the fc-1-GHD $H$ satisfies ``$\Bag(t) \subseteq \Bag(p)$ or $\Bag(t) \supseteq \Bag(p)$'' for all its edges $\smallset{t,p}$).
	Afterwards, by \ref{item:BetahIsConsistent} it is easy to see that $h'' = \beta(h)$.
	This proves that $\beta$ is surjective and completes the proof of Claim~\ref{claim:PropertiesOfBeta}\ref{item:betaIsSurjective}.
	See Appendix~\ref{appendix:ReductionFromArbitraryToBinary} for details.
\end{proof}

Parts \ref{item:betaIsSurjective} and \ref{item:betaIsSufficientlyInjective} of Claim~\ref{claim:PropertiesOfBeta} imply that there exists a bijection $f\colon \smallsem{\QBinary}(\dbBinary) \to \smallsem{\query}(D)$.
A closer inspection shows that, when given a tuple $\at \in \smallsem{\QBinary}(\dbBinary)$, the tuple $f(\at) \in \smallsem{\query}(D)$ can be computed in time $O(|\free(Q)|{\cdot}k)$
(see Appendix~\ref{appendix:ReductionFromArbitraryToBinary}).
This, finally, completes the proof of Theorem~\ref{thm:ReductionArbSchemaToBinarySchema}.
\subsection{From Binary Schemas to Node-Labeled Graphs}%
\label{sec:ReductionToOneBinaryRelation}
This subsection is devoted to proving the following theorem.
\begin{theorem}\label{thm:ReductionBinarySchemaToBinarySchemaWithOnlyOneBinaryRelation}
	For any binary schema $\sigma$ there exists a schema $\sigmaSimple$ for node-labeled graphs, such that $|\sigmaSimple| = |\sigma| + 3$, and
	\begin{enumerate}[(1)]
		\item\label{item:one:ReductionBinaryToOnlyOneBinary}
		upon input of a $\sigma$-db $D$, we can compute in time $\bigOh(\dbsize{D})$ a node-labeled graph $\dbSimple$ of schema $\sigmaSimple$, and
		\item\label{item:two:ReductionBinaryToOnlyOneBinary}
		upon input of any query $\query\in\fcACQ$, we can compute in time $\bigOh(\size{Q})$ a query $\QSimple \in \fcACQSimple$ with $|\free(\QSimple)| < 3 \cdot |\free(\query)|$, such that
	 	\item\label{item:three:ReductionBinaryToOnlyOneBinary}
		 there is a bijection $g\colon \smallsem{\QSimple}(\dbSimple) \to \smallsem{\query}(D)$.
		 Furthermore, when given a tuple $\at \in \smallsem{\QSimple}(\dbSimple)$, the tuple $g(\at) \in \smallsem{\query}(D)$ can be computed in time $\bigOh(|\free(\query)|)$.
	\end{enumerate}
\end{theorem}
\noindent
Similar to Theorem~\ref{thm:ReductionArbSchemaToBinarySchema}, we present the proof details ordered by statements \ref{item:one:ReductionBinaryToOnlyOneBinary} to \ref{item:three:ReductionBinaryToOnlyOneBinary}.

\paragraph{\ref{item:one:ReductionBinaryToOnlyOneBinary} Choosing $\sigmaSimple$ and Constructing the $\sigmaSimple$-db $\dbSimple$.}
Let $\sigma$ be an arbitrary binary schema and let $\sigma_{|2}$ be the set of all binary relation symbols in $\sigma$.
Let $\sigmaSimple$ be the schema that contains only a single binary relation symbol $\RSimple$, all the unary relation symbols of $\sigma$, a unary relation symbol $U_F$ for every $F\in\sigma_{|2}$, and two further new unary relation symbols $V$ and $W$.
Clearly, $|\sigmaSimple|=|\sigma|+3$.

We represent any $\sigma$-db $D$ by a $\sigmaSimple$-db $\dbSimple$ as follows.
We let $V^{\dbSimple}\deff \adom{D}$.
For every unary relation symbol $X\in\sigma$ let $X^{\dbSimple}\deff X^D$.
We initialize $\RSimple^{\dbSimple}$, $W^{\dbSimple}$, and $(U_F)^{\dbSimple}$ for all $F\in\sigma_{|2}$ as the empty set $\emptyset$.
Let $\myTuples \deff \bigcup_{F\in\sigma_{|2}} F^D$ and let $\myTuplesSym$ be the symmetric closure of $\myTuples$, i.e., $\myTuplesSym = \myTuples \cup \setc{ (b,a) }{(a,b) \in \myTuples}$.
For every tuple $(a,b) \in \myTuplesSym$, we choose a new element $w_{ab}$ in $\Dom$ and
insert it into $W^{\dbSimple}$.
Furthermore, we insert into $\RSimple^{\dbSimple}$ symmetric edges between $a$ and $w_{ab}$ and between $w_{ab}$ and $w_{ba}$;
see Figure~\ref{fig:BinToGraph_gadgets}.

\begin{figure}[b]
	\centering
	\captionsetup[subfigure]{justification=centering}
	\begin{subfigure}[t]{0.39\textwidth}
		\centering
		\begin{tikzpicture}[
			yscale=.8,
			every node/.style={},
			V/.style={draw=blue, circle, minimum size=19pt},
			W/.style={draw=red, rectangle},
			every loop/.style={looseness=5}
		]
			\node[V] (a) at (0,0) {$a$};
			\node[W, right=of a] (waa) {$w_{aa}$};
			\draw (a) -- (waa);
			\draw (waa) edge[in=-25,out=25,loop] ();
		\end{tikzpicture}
		\caption{Edges inserted for $(a,a)\in \myTuplesSym$.}%
		\label{fig:BinToGraph_gadget_loop}
	\end{subfigure}
	\hfill
	\begin{subfigure}[t]{0.59\textwidth}
		\centering
		\begin{tikzpicture}[
			yscale=.8,
			every node/.style={},
			V/.style={draw=blue, circle, minimum size=19pt},
			W/.style={draw=red, rectangle}
		]
			\node[V] (a) at (0,0) {$a$};
			\node[W, right=of a] (wab) {$w_{ab}$};
			\node[W, right=of wab] (wba) {$w_{ba}$};
			\node[V, right=of wba] (b) {$b$};
			\draw (a) -- (wab) -- (wba) -- (b);
		\end{tikzpicture}
		\caption{Edges inserted for $a \neq b$ and $(a,b), (b,a) \in \myTuplesSym$.}%
		\label{fig:BinToGraph_gadget_edge}
	\end{subfigure}
	\hfill
	\caption{%
		Circled nodes are in $V^{\dbSimple}$, boxed nodes in $W^{\dbSimple}$.
		A self-loop represents the edge $(w_{aa},w_{aa})$ in $\RSimple^{\dbSimple}$;
		an undirected edge between two nodes $x$ and $y$ represents edges in both directions, i.e., $(x,y)$ and $(y,x)$ in $\RSimple^{\dbSimple}$.
	}%
	\label{fig:BinToGraph_gadgets}
\end{figure}
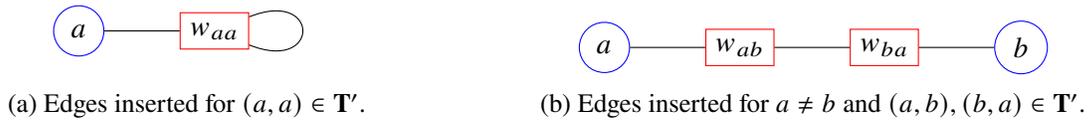

Note that every node in $W^{\dbSimple}$ has exactly one neighbor in $V^{\dbSimple}$ and one neighbor in $W^{\dbSimple}$.
We complete the construction of $\dbSimple$ by looping through all $F\in\sigma_{|2}$ and all $(c,d) \in F^D$ and inserting $w_{cd}$ into $(U_F)^{\dbSimple}$.
Note that $\dbSimple$ is a node-labeled graph  that can be constructed in time $O(|D|)$.
\begin{figure}[h!tbp]
	\centering
\begin{tikzpicture}[
	defaultstyle,
	level distance=2.25cm,
	sibling distance=2.5cm,
	mnode/.style={draw, rounded corners, minimum width=6mm},
	edge from parent/.style={draw=none},
	enode/.style={inner sep=1mm, sloped, at start, font={\tiny }},
	medge/.style={->}
]
	\node[mnode] (PS) at (0,0) {\exdata{PS}}
	child {
		node[mnode] (LM) {\exdata{LM}} 
		child { node[mnode] (18m) {\exdata{18m}} }
		child { node[mnode] (DrS) {\exdata{Dr.S}}}	
	}
	child {
		node[mnode] (MM) {\exdata{MM}}
		child { (DrS) }
		child { node[mnode] (34m) {\exdata{34m}} }
	};

	\node[UP] (PSLM) at ($(PS)!.33!(LM)$) {$w_{\exdata{\tiny PL}}$};
	\node[UA] (LMPS) at ($(PS)!.66!(LM)$) {$w_{\exdata{\tiny LP}}$};
	\draw (PS) edge (PSLM);
	\draw (PSLM) edge (LMPS);
	\draw (LMPS) edge (LM);

	\node[UP] (PSMM) at ($(PS)!.33!(MM)$) {$w_{\exdata{\tiny PM}}$};
	\node[UA] (MMPS) at ($(MM)!.33!(PS)$) {$w_{\exdata{\tiny MP}}$};
	\draw (PS) edge (PSMM);
	\draw (PSMM) edge (MMPS);
	\draw (MMPS) edge (MM);

	\node[US] (LM18m) at ($(LM)!.33!(18m)$) {$w_{\exdata{\tiny L1}}$};
	\node[] (18mLM) at ($(18m)!.33!(LM)$) {$w_{\exdata{\tiny 1L}}$};
	\draw (LM) edge (LM18m);
	\draw (LM18m) edge (18mLM);
	\draw (18mLM) edge (18m);

	\node[UM] (LMDrS) at ($(LM)!.33!(DrS)$) {$w_{\exdata{\tiny LD}}$};
	\node[] (DrSLM) at ($(DrS)!.33!(LM)$) {$w_{\exdata{\tiny DL}}$};
	\draw (LM) edge (LMDrS);
	\draw (LMDrS) edge (DrSLM);
	\draw (DrSLM) edge (DrS);

	\node[UM] (MMDrS) at ($(MM)!.33!(DrS)$) {$w_{\exdata{\tiny MD}}$};
	\node[] (DrSMM) at ($(DrS)!.33!(MM)$) {$w_{\exdata{\tiny DM}}$};
	\draw (MM) edge (MMDrS);
	\draw (MMDrS) edge (DrSMM);
	\draw (DrSMM) edge (DrS);

	\node[US] (MM34m) at ($(MM)!.33!(34m)$) {$w_{\exdata{\tiny M3}}$};
	\node[] (34mMM) at ($(34m)!.33!(MM)$) {$w_{\exdata{\tiny 3M}}$};
	\draw (MM) edge (MM34m);
	\draw (MM34m) edge (34mMM);
	\draw (34mMM) edge (34m);
\end{tikzpicture} %
	\caption{Representation of $\exdbSimple$ from Example~\ref{example:running1}.}\label{fig:appendix_example:rep:graph}
\end{figure}
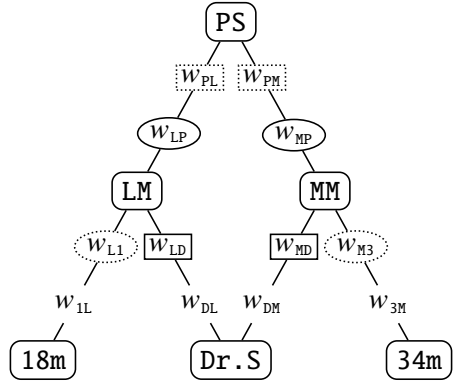
\begin{example}\label{example:running1}
	Consider the following example about movies and actors
	taken from~\cite[{F}ig. 2]{AnglesABHRV17}.
	The schema $\exsigma$ has binary relation symbols
	\exrel{Plays}, \exrel{ActedBy}, \exrel{Movie}, and \exrel{Screentime}
	(denoted by \exrel{P}, \exrel{A}, \exrel{M}, and \exrel{S}),
	and the database $\exdb$ has the following relations and tuples:
	\smallskip%
	\begin{center}%
	\begin{tabular}{cccc}
		\begin{tabular}{r|cc}
			\exrel{P} & &  \\ \hline
			& \exdata{PS} & \exdata{LM} \\
			& \exdata{PS} & \exdata{MM}
		\end{tabular} \qquad
		&
		\begin{tabular}{r|cc}
			\exrel{A} & &  \\ \hline
			& \exdata{LM} & \exdata{PS} \\
			& \exdata{MM} & \exdata{PS}
		\end{tabular} \qquad
		&
		\begin{tabular}{r|cc}
			\exrel{M} & &  \\ \hline
			& \exdata{LM} & \exdata{Dr.S} \\
			& \exdata{MM} & \exdata{Dr.S} 
		\end{tabular} \qquad
		&
		\begin{tabular}{r|cc}
			\exrel{S} & &  \\ \hline
			& \exdata{LM} & \exdata{18m} \\
			& \exdata{MM} & \exdata{34m}
		\end{tabular}
	\end{tabular}
	\end{center}
	\smallskip
	where Peter Sellers (\exdata{PS}) is an actor who plays as Lionel Mandrake (\exdata{LM}) 
	and Merkin Muffley (\exdata{MM}) in the same movie ``Dr.\ Strangelove'' (\exdata{Dr.S}).
	Each character appears 18 minutes (\exdata{18m}) and 34 minutes (\exdata{34m}) 
	in the movie, respectively.

	The corresponding schema $\exsigmaSimple$ consists of a single
	binary relation symbol $\RSimple$ and the unary relation symbols $V$, $W$,
	$U_{\exrel{P}}$,
	$U_{\exrel{A}}$,
	$U_{\exrel{M}}$,
	$U_{\exrel{S}}$.
	The according $\exsigmaSimple$-db $\exdbSimple$ is depicted in Figure~\ref{fig:appendix_example:rep:graph}.
	Elements of the form $w_{ab}$ are abbreviated using the first character of $a$ and $b$, e.g., instead of $w_{\exdata{\tiny PSLM}}$ we write $w_{\exdata{\tiny PL}}$. Since these elements are exactly those in $W^{\exdbSimple}$, and all other elements belong to $V^{\exdbSimple}$, we do not indicate them further.
	Because $\RSimple^{\exdbSimple}$ is symmetric, we draw it using undirected edges.
	The unary relations $U_{\exrel{P}}^{\exdbSimple}$, $U_{\exrel{A}}^{\exdbSimple}$, $U_{\exrel{M}}^{\exdbSimple}$, $U_{\exrel{S}}^{\exdbSimple}$, are depicted as
	\pdftooltip{\tikz{\node[UP] at (0,0) {\phantom{$w_{\exdata{\tiny PL}}$}};}}{a dotted rectangle},
	\pdftooltip{\tikz{\node[UA, inner sep=0.5pt] at (0,0) {\phantom{$w_{\exdata{\tiny PL}}$}};}}{a solid ellipse},
	\pdftooltip{\tikz{\node[UM] at (0,0) {\phantom{$w_{\exdata{\tiny PL}}$}};}}{a solid rectangle},
	\pdftooltip{\tikz{\node[US, inner sep=0.5pt] at (0,0) {\phantom{$w_{\exdata{\tiny PL}}$}};}}{a dotted ellipse}, respectively.
\end{example} 
\paragraph{\ref{item:two:ReductionBinaryToOnlyOneBinary}  Constructing the fc-ACQ $\QSimple$.}
Our next aim is to translate  queries $Q\in\fcACQ$ into suitable $\sigmaSimple$-queries $\QSimple$.
We want $\QSimple$ to be in $\fcACQSimple$, and we want to ensure that there is an easy to compute bijection $g$ that maps the tuples in $\bigsem{\QSimple}(\dbSimple)$ onto the tuples in $\sem{Q}(D)$.

Consider an arbitrary query $Q\in\fcACQ$.
Consider the Gaifman graph $G(Q)$ of $Q$ and recall from Proposition~\ref{prop:binaryfcacqs} that $G(Q)$ is a forest, and for every connected component $C$ of $G(Q)$, the subgraph of $C$ induced by the set $\free(Q) \cap V(C)$ is connected or empty.
For each connected component $C$ of $G(Q)$, we orient the edges of $C$ as follows.
If $\free(Q) \cap V(C) \neq \emptyset$, then choose an arbitrary vertex $r_C \in \free(Q)\cap V(C)$ as the \emph{root} of $C$;
otherwise choose an arbitrary vertex $r_C \in V(C)$ as the root of $C$.
Orient the edges of $C$ to be directed \emph{away} from the root $r_C$.
Let $\vec{G}(Q)$ be the resulting oriented version of $G(Q)$.
Furthermore, let $S$ be the set of all variables $x$ of $Q$ such that $F(x,x) \in \Atoms(Q)$ for some $F \in \sigma_{|2}$.
We construct the $\sigmaSimple$-query $\QSimple$ as follows.

We initialize $\atoms(\QSimple)$ to consist of all the \emph{unary} atoms of $Q$.
For every $x \in \vars(Q)$, we add to $\atoms(\QSimple)$ the unary atom $V(x)$.
For every $x \in S$, we use a new variable $z_{xx}$ and add to $\Atoms(\QSimple)$ the atoms $W(z_{xx})$, $\RSimple(x, z_{xx})$, and $\RSimple(z_{xx}, z_{xx})$.
For every directed edge $(x,y)$ of $\vec{G}(Q)$, we use two new variables $z_{xy}$ and $z_{yx}$ and insert into $\atoms(\QSimple)$ the atoms $W(z_{xy})$, $W(z_{yx})$, $\RSimple(x, z_{xy})$, $\RSimple(z_{xy}, z_{yx})$ and $\RSimple(z_{yx}, y)$.
Finally, for every atom of $Q$ of the form $F(u,v)$
(with $F\in\sigma_{|2}$ and $u, v \in \Vars(Q)$),
we add the atom $U_F(z_{uv})$ to $\atoms(\QSimple)$.
The head of $\QSimple$ is obtained from the head of $Q$ by appending to it the variables $z_{xy}$ and $z_{yx}$ for all edges $(x,y)$ of $\vec{G}(Q)$ where $x$ and $y$ both belong to $\free(Q)$.
It is not difficult to verify that $|\free(\QSimple)| < 3 \cdot |\free(Q)|$ and that $\QSimple$ indeed is an fc-ACQ
(see Appendix~\ref{appendix:ReductionToOneBinaryRelation}).
\begin{example}\label{example:running:continued}
	Consider the schemas $\exsigma$ and $\exsigmaSimple$ from Example~\ref{example:running1}, and let $Q$ be the query
	\[
		\Ans(x,y_1) \; \leftarrow \; \exrel{A}(x,y_1),\; \exrel{A}(x,y_2),\;\exrel{P}(y_2,x).
	\]
	Consider the oriented version $\vec{G}(Q)$ of the Gaifman graph of $Q$ obtained by choosing $x$ as the root.
	Then, $\QSimple$ is the following query:
	\begin{align*}
		\Ans(x,y_1,z_{xy_1},z_{y_1x})
		\ \leftarrow \
		&V(x), V(y_1), V(y_2), U_{\exrel{A}}(z_{xy_1}), U_{\exrel{A}}(z_{xy_2}), U_{\exrel{P}}(z_{y_2x}),\\
		&\RSimple(x,z_{xy_1}), \RSimple(z_{xy_1},z_{y_1x}), \RSimple(z_{y_1x},y_1), W(z_{xy_1}), W(z_{y_1x}),\\
		&\RSimple(x,z_{xy_2}), \RSimple(z_{xy_2},z_{y_2x}), \RSimple(z_{y_2x},y_2), W(z_{xy_2}), W(z_{y_2x}).\endExample
	\end{align*}
\end{example}

\paragraph{\ref{item:three:ReductionBinaryToOnlyOneBinary} The Bijection $g$ Between Outputs.}
Along the definition of $\dbSimple$ and $\QSimple$ one can verify the
following (for a proof, see Appendix~\ref{appendix:ReductionToOneBinaryRelation}).
\begin{restatable}{claim}{claimReductionBinaryToGraph}%
	\label{claim:ReductionBinaryToGraph}~
	\begin{enumerate}[(a)]
	\item\label{claim:reduction:three}
	For every  $\nu \in \Hom(Q,D)$, the following mapping $\nuSimple$ is a homomorphism from $\QSimple$ to $\dbSimple$: \
	for all $x \in \vars(Q)$ let $\nuSimple(x) \deff \nu(x)$;
	for all $x \in S$ let $\nuSimple(z_{xx}) \deff w_{aa}$ for $a \deff \nu(x)$;
	and for all edges $(x,y)$ of $\vec{G}(Q)$ let $\nuSimple(z_{xy}) \deff w_{ab}$ and $\nuSimple(z_{yx}) \deff w_{ba}$ for $a \deff \nu(x)$ and $b \deff \nu(y)$.
 \item\label{claim:reduction:QSimpleToQ}
	For every homomorphism $\nuSimple$ from $\QSimple$ to $\dbSimple$, the following holds:
	\begin{enumerate}[(i), topsep=0pt]
		\item\label{claim:reduction:one}
		For every edge $(x,y)$ of $\vec{G}(Q)$, for $a \deff \nuSimple(x)$ and $b \deff \nuSimple(y)$ we have $a, b \in \adom{D}$ and $\nuSimple(z_{xy}) = w_{ab}$ and $\nuSimple(z_{yx}) = w_{ba}$.
		For every $x \in S$ we have $a \deff \nuSimple(x) \in \adom{D}$ and $\nuSimple(z_{xx}) = w_{aa}$.
		\item\label{claim:reduction:two}
		The mapping $\nu$ with $\nu(x) \deff \nuSimple(x)$, for all $x \in \vars(Q)$, is a homomorphism from $Q$ to $D$.\endClaim
	\end{enumerate}
\end{enumerate}
\end{restatable}
From Claim~\ref{claim:ReductionBinaryToGraph} we obtain that the mapping $\beta$ that maps every homomorphism $\nuSimple \in \Hom(\QSimple,\dbSimple)$ to the restriction of $\nuSimple$ to the set $\vars(Q)$, is a bijection between $\Hom(\QSimple, \dbSimple)$ and $\Hom(Q,D)$.
This, in particular, implies that $\sem{Q}(D) = g(\bigsem{\QSimple}(\dbSimple))$, where $g$ projects every tuple in $\bigsem{\QSimple}(\dbSimple)$ to the first $|\free(Q)|$ components of this tuple.
Furthermore, as an immediate consequence of item~\ref{claim:reduction:one} of part~\ref{claim:reduction:QSimpleToQ} of Claim~\ref{claim:ReductionBinaryToGraph} we obtain that for all tuples $t, t' \in \bigsem{\QSimple}(\dbSimple)$ with $t \neq t'$ we have: $g(t) \neq g(t')$.
This yields statement~\ref{item:three:ReductionBinaryToOnlyOneBinary} and completes the proof of Theorem~\ref{thm:ReductionBinarySchemaToBinarySchemaWithOnlyOneBinaryRelation}.

\section{Solving the Problem for Node-Labeled Graphs}%
\label{sec:OneBinaryRelation}
In this section we prove Theorem~\ref{thm:maintheorem} for the special case where the schema $\sigma$ consists of one binary symbol $E$ and a finite number of unary symbols, and where the relation $E^D$ of the given $\sigma$-db $D$ is symmetric.
We can think of $D$ as an undirected, node-labeled graph that may contain self-loops.
\subsection{The Indexing Phase}%
\label{sec:mainresult}
This subsection describes the indexing phase of our solution.
As input we receive a node-labeled graph $D$ of schema $\sigma$ (i.e., $E^D$ is symmetric).
We build a data structure $\DSD$ which we call the \emph{color-index};
it is an index structure for $D$ that supports efficient evaluation of all queries in $\fcACQ$.
\paragraph*{Encoding loops.}
We start by labeling self-loops in $E^D$ by a new unary relation symbol $L \not\in \sigma$.
Let~$\sigmaOne \isdef \sigma \cup \set{ L }[]$.
We turn $D$ into a $\sigmaOne$-db $\DOne$ by letting $L^{\DOne} \isdef \set{(v) \mid (v,v) \in E^D}$ and $R^{\DOne} \isdef R^D$ for every $R \in \sigma$.
Note that self-loops are still also present in the relation $E^{\DOne}$, and $L^{\DOne}$ is a redundant representation of these self-loops.
Clearly, the size of $\DOne$ is linear in the size of $D$.

Recall that $D$, and thus also $\DOne$, can be thought of as node-labeled undirected graphs that may have self-loops (because $E^D = E^{\DOne}$ is symmetric).
In the following it will be easier to think of $\DOne$ as such and use the usual notation associated with graphs.
I.e., we represent $\DOne$ as the node-labeled undirected graph (with self-loops) $\GOne \isdef (\VOne, \EOne, \vl)$ defined as follows:
$\VOne = \Adom(D)$, $\EOne$ consists of all undirected edges $\set{ v, w }[]$ such that $(v,w) \in E^{\DOne}$---note that this implies $\EOne$ may contain singleton sets representing loops---and for every node $v \in \VOne$ we let $\vl(v) \isdef \set{ U \in \sigmaOne \mid \ar(U)=1,\ (v) \in U^{\DOne} }$.
\paragraph*{Color refinement.}
We apply to $\GOne$ a suitable variant of the well-known \emph{color refinement} algorithm.
A high-level description of this algorithm, basically taken from~\cite{BBG-ColorRefinement}, is as follows.
The algorithm classifies the nodes of $\GOne$ by iteratively refining a coloring of the nodes.
Initially, each node $v$ has color $\vl(v)$, and note that $L \in \vl(v)$ iff $v$ has a self-loop.
Then, in each step of the iteration, two nodes $v, w$ that currently have the same color get different refined colors iff for some current color $c$ we have $\numN{v}{c} \neq \numN{w}{c}$.
Here, for any node $u$ of $\GOne$, we let $\numN{u}{c} \isdef |\N{u}{c}|$ where $\N{u}{c}$ denotes the set of all neighbors of $u$ that have color $c$.
Note that if $u$ has a self-loop, then it is a neighbor of $u$ and, in particular, for the current color $d$ of $u$, we have $u \in \N{u}{d}$.
The process stops if no further refinement is achieved, resulting in a so-called \emph{stable coloring} of the nodes.

Formally, we say that color refinement computes a \enquote{coarsest stable coloring that refines $\vl$}, which is defined using the following notation.
Let $S_f$ and $S_g$ be sets and let $f\colon \VOne \to S_f$ and $g\colon \VOne \to S_g$ be functions.
We say $f$ \emph{refines} $g$ iff for all $v,w \in \VOne$ with $f(v) = f(w)$ we have $g(v) = g(w)$.
Further, we say that $f$ \emph{is stable} iff for all $v, w \in \VOne$ with $f(v) = f(w)$, and every color $c \in S_f$ we have: $\numN{v}{c} = \numN{w}{c}$.
Finally, $f$ is a \emph{coarsest stable coloring that refines $g$} iff $f$ is stable, refines $g$, and for every coloring $h\colon \VOne \to S_h$ (for some set $S_h$) that is stable and refines $g$ we have: \,$h$ refines~$f$.
It is well-known (see for example~\cite{Cardon1982,BBG-ColorRefinement,CFI-paper,Immerman1990}) that color refinement can be implemented to run in time $O((|\VOne|+|\EOne|) {\cdot} \log{|\VOne|})$, which yields the following theorem in our setting:
\begin{theorem}[\cite{Cardon1982,BBG-ColorRefinement,CFI-paper,Immerman1990}]\label{thm:ColorRefinement}
	Within time $O((|\VOne|+|\EOne|) \log|\VOne|)$ one can compute a coarsest stable coloring $\col\colon \VOne \to C$ that refines $\vl$ (for a suitably chosen set $C$ with $\img(\col) = C$).
\end{theorem}

In the indexing phase, we apply the algorithm provided by Theorem~\ref{thm:ColorRefinement} to compute $\col$.
Let $\fcr(\DOne)$ denote the time taken for this.
Note that $\fcr(\DOne) \in O(|\DOne| \cdot \log{|\Adom(\DOne)|})$, which is the \emph{worst-case} complexity;
in case that $\DOne$ has a particularly simple structure, the algorithm may terminate already in time $O(|\DOne|)$.
Also note that the number $|C|$ of colors used by $\col$ is the smallest number possible in order to obtain a stable coloring that refines $\vl$, and, moreover, the coarsest stable coloring that refines $\vl$ is \emph{unique} up to a renaming of colors.
Furthermore, the following is true for all $v, w \,{\in}\, \VOne$:\;
if $\col(v) = \col(w)$, then $\vl(v) = \vl(w)$ and $\numN{v}{c} = \numN{w}{c}$ for all $c \in C$.
This lets us define the following notation:
for all $c, c' \in C$ we let $\numN{c}{c'} \isdef \numN{v}{c'}$ for some (and hence for every) $v \in \VOne$ with $\col(v) = c$.
We now proceed to the final step of the indexing phase, in which we use $\col$ to build our \emph{color-index} data structure.
\paragraph{The color-index.}
The data structure $\DSD$ that we build in the indexing phase consists of
\begin{enumerate}[(1)]
	\item\label{item:ciD:one}
	the schema $\sigmaOne$ and the $\sigmaOne$-db $\DOne$;
	\item\label{item:ciD:two}
	a lookup table to access the color $\col(v)$ given a vertex $v \in \VOne$, and
	an (inverse) lookup table to access the set $\setc{v\in\VOne}{\col(v) = c}$ given a $c \in C$, plus the number $\Numb{c}$ of elements in this set;
	\item\label{item:ciD:three}
	a lookup table to access the set $\N{v}{c}$, given a vertex $v \in \VOne$ and a color $c \in C$;
	\item\label{item:ciD:four}
	a lookup table to access the number $\numN{c}{c'}$, given colors $c,c' \in C$;
	\item\label{item:ciD:five}
	the \emph{color database} $\ciD$ of schema $\sigmaOne$ with $\adom{\ciD} = C$ defined as follows:
	\begin{itemize}
		\item
		For each unary relation symbol $U \in \sigmaOne$ we let $U^{\ciD}$ be the set of unary tuples $(c)$ for all $c \in C$ such that there is a $(v) \in U^{\DOne}$ with $\col(v) = c$.
		\item
		We let $E^{\ciD}$ be the set of all tuples $(c,c') \in C \times C$ such that $\numN{c}{c'} > 0$.\\
		Note that $E^{\ciD}$ is symmetric and may contain tuples of the form $(c,c)$.
	\end{itemize}
\end{enumerate}

Note that after constructing $\DOne$ in time $O(|D|)$ and computing $\col$ in time $\fcr(D)$, the remaining components~\ref{item:ciD:two}--\ref{item:ciD:five} of $\DSD$ can be built in total time $O(|D|)$.
Thus, in summary, the indexing phase takes time $\fcr(D)+O(|D|)$.
The color database $\ciD$ has size $O(|D|)$ in the \emph{worst case};
but $|\ciD|$ might be substantially smaller than $|D|$.
We will further discuss this in Section~\ref{sec:SizeOfIndex}.
\subsection{Using the Color-Index to Evaluate fc-ACQs }\label{sec:eval}
This subsection shows how to use the color-index $\DSD$ to evaluate any fc-ACQ $Q$ on $D$.
I.e.\ it describes the \emph{evaluation phase} of our solution of $\IndexingProblemOurs$.
We assume that the color-index $\DSD$ (including the color database $\ciD$) of the given node-labeled graph $D$ of schema $\sigma$ has already been built during the indexing phase.
Let $Q \in \fcACQ$ be an arbitrary query with a head of the form $\Ans(x_1,\ldots,x_k)$ (for $k \geq 0$) that we receive as input during the evaluation phase.
We first explain how the evaluation tasks can be simplified by focusing on \emph{connected} queries.
\paragraph{Connected queries.}
We say that $Q$ is \emph{connected} iff its Gaifman graph is connected.
If $Q$ is not connected,  we can write $Q$ as
$
	\Ans(\tup{z}_1,\ldots,\tup{z}_\ell) \ \leftarrow \
	\tup{\alpha}_1(\tup{z}_1,\tup{y}_1), \ldots, \tup{\alpha}_\ell(\tup{z}_\ell,\tup{y}_\ell)
$
such that each $\tup{\alpha}_i(\tup{z}_i,\tup{y}_i)$ is a sequence of atoms,
$\vars(\tup{\alpha}_i(\tup{z}_i,\tup{y}_i))$ and $\vars(\tup{\alpha}_j(\tup{z}_j,\tup{y}_j))$ are disjoint for $i \neq j$, and for each $i \in [\ell]$ the CQ $Q_i \isdef \Ans(\tup{z}_i) \leftarrow \tup{\alpha}_i(\tup{z}_i,\tup{y}_i)$ is connected.
To simplify notation we assume w.l.o.g.\ that in the head of $Q$ the variables are ordered in the same way as in the list $\tup{z}_1,\ldots,\tup{z}_\ell$.
One can easily check that this decomposition satisfies $\sem{Q}(D) = \sem{Q_1}(D) \times \cdots \times \sem{Q_\ell}(D)$  and $|\sem{Q}(D)| = \prod_{i=1}^{\ell} |\sem{Q_i}(D)|$.
Hence, we can compute $|\sem{Q}(D)|$ by first computing each number $|\sem{Q_i}(D)|$, and then multiplying all values in $O(|Q|)$-time.
Similarly, for enumerating $\sem{Q}(D)$ one can convert constant-delay enumeration algorithms for $Q_1, \ldots, Q_\ell$ into one for $Q$ by iterating in nested loops over the outputs of $\sem{Q_1}(D), \ldots,\allowbreak \sem{Q_\ell}(D)$.
Henceforth, we assume without loss of generality that $Q$ is connected.
\paragraph{Handling loops in $Q$.}
When receiving $Q \in \fcACQ$, the first step is to remove the self-loops of $Q$, i.e., we translate $Q$ into the $\sigmaOne$-query $\QOne$ by replacing every atom of the form $E(x,x)$ in $Q$ with the atom $L(x)$.
Obviously, the Gaifman graph of the CQ remains unchanged (i.e., $G(\QOne) = G(Q)$), $\QOne$~is free-connex acyclic, connected, and \emph{self-loop-free}, i.e., $x \neq y$ for every atom in $\QOne$ of the form $E(x,y)$.
Clearly, $\sem{Q}(D) = \sem{\QOne}(\DOne)$.
Hence, instead of evaluating $Q$ on $D$ we can evaluate $\QOne$ on $\DOne$.
\paragraph{Ordering the variables in $\QOne$.}
To establish an order on the variables in $\QOne$, we start by picking an arbitrary variable $r$ in $\free(\QOne)$, and if $\QOne$ is a Boolean query (i.e., $\free(\QOne) = \emptyset$) we pick $r$ as an arbitrary variable in $\vars(\QOne)$.
This $r$ will be fixed from now on.
Since $Q$ is connected, picking $r$ as the \emph{root node} turns the Gaifman graph $G(\QOne)$ into a rooted tree $T$.
Furthermore, since $\QOne \in \fcACQOne$, it follows from Proposition~\ref{prop:binaryfcacqs} that there exists a strict linear order $<$ on $\vars(\QOne)$ satisfying $x < y$ for all $x \in \free(\QOne)$ and all $y \in \quant(\QOne)$, such that $<$ is compatible with the descendant relation of the rooted tree $T$, i.e., for all $x \in \vars(\QOne)$ and all $y \in \vars(\QOne)$ that are descendants of $x$ in $T$ we have $x < y$.
Such a $<$ can be obtained based on a variant of breadth-first search of $G(\QOne)$ that starts with node $r$ and that prioritizes free over quantified variables using separate queues for free and for quantified variables.
We choose an arbitrary such order $<$ and fix it from now on.
We will henceforth assume that $r = x_1$ and that $x_1 < x_2 < \cdots < x_k$
(by reordering the variables $x_1,\ldots,x_k$ in the head of~$\QOne$ and $Q$, this can be achieved without loss of generality).

For every node $x$ of $T$ we let $\lambda_x$ be the set of unary relation symbols $U \in \sigmaOne$ such that $U(x) \in \Atoms(\QOne)$; we let $\Children(x)$ be the set of its children; and if $x \neq x_1$, we write $\Parent(x)$ to denote the parent of $x$ in $T$.
Upon input of $Q$ we can compute in time $O(|Q|)$ the query $\QOne$, the rooted tree $T$, and a lookup table to obtain $O(1)$-time access to $\lambda_x$ for each $x \in \vars(\QOne)$.
The following lemma summarizes the correspondence between homomorphisms from $\QOne$ to $\DOne$, the labels $\lambda_x$ associated to the nodes $x$ of $G(\QOne)$, and the node-labeled graph $\GOne = (\VOne, \EOne, \vl)$; see
Appendix~\ref{app:proof:homomorphisms} for a proof.

\begin{restatable}{lemma}{lemmaHomomorphisms}\label{lemma:homomorphisms}
	A mapping $\val\colon \vars(\QOne) \to \Adom(\DOne)$ is a homomorphism from $\QOne$ to $\DOne$ if, and only if, $\vl(\val(x)) \supseteq \lambda_x$, for every $x \in \vars(\QOne)$, and $\set{ \val(x),\, \val(y) }[] \in \EOne$ for every edge $\set{ x, y }[]$ of $G(\QOne)$.
\end{restatable} %
\paragraph{Evaluation phase for Boolean queries (task ``\booltask'').}
As a warm-up, we start with the evaluation of \emph{Boolean} queries.
For this task, we assume that $Q$ is a Boolean query, and as described before, we assume w.l.o.g.\ that it is connected.
The following lemma shows that evaluating $Q$ on $D$ can be reduced to evaluating the query $\QOne$ on the color database $\ciD$.
\begin{restatable}{lemma}{lemmaBool}\label{lemma:bool}
	If $Q$ is a Boolean query, then $\sem{Q}(D) = \sem{\QOne}(\DOne) = \sem{\QOne}(\ciD)$.
\end{restatable}
The proof follows from Lemma~\ref{lemma:homomorphisms}, our particular choice of $\ciD$, and the fact that $\col$ is a stable coloring of $\GOne$ that refines $\vl$ (see appendices~\ref{app:eval:observations} and~\ref{app:proof:bool} for details).
For Boolean queries, the algorithm from Theorem~\ref{thm:BGS-enum} enumerates the empty tuple $()$ or nothing at all.
Thus, the combination of Lemma~\ref{lemma:bool} and Theorem~\ref{thm:BGS-enum} yields that we can solve the task by checking whether the empty tuple $()$ is enumerated upon input of $\ciD$ and $\QOne$, which takes time $O(|\QOne| {\cdot} |\ciD|) + O(|\free(\QOne)|)$.
Since $|\free(\QOne)| = \emptyset$ and $|\QOne| \in O(|Q|)$, this solves the task ``$\booltask$'' in time $O(|Q| {\cdot} |\ciD|)$. %
\paragraph{Evaluation phase for non-Boolean queries (task ``\enumtask'').}
We now assume that $\QOne$ is a connected $k$-ary query for some $k \geq 1$.
Recall that the head of $\QOne$ is $\Ans(x_1,\ldots,x_k)$ with $x_1 < \cdots < x_k$, where $<$ is the order we associated with the query $\QOne$.
Further, recall that $T$ is the rooted tree obtained from the Gaifman graph $G(\QOne)$ by selecting $r=x_1$ as its root.
The following technical lemma highlights the connection between $\sem{\QOne}(\DOne)$ and $\sem{\QOne}(\ciD)$ that will allow us to use the color-index for enumerating $\sem{\QOne}(\DOne)$.
To state the lemma, we need the following notation.
For any $i \in [k]$, we say that $(v_1, \ldots, v_i)$ is a \emph{partial output} of $\QOne$ over $\DOne$ of color $(c_1, \ldots, c_k)$ iff there exists an extension $(v_{i+1}, \ldots, v_k)$ such that $(v_1, \ldots, v_i,\allowbreak v_{i+1}, \ldots, v_k) \in \sem{\QOne}(\DOne)$ and $(\col(v_1), \ldots, \col(v_k)) = (c_1, \ldots, c_k)$.

\begin{restatable}{lemma}{lemmaMainLemma}\label{lemma:mainLemma}~
  \begin{enumerate}[(a)]
    \item\label{item:a:lemma:mainLemma}
    For every $(v_1,\ldots,v_k) \in \sem{\QOne}(\DOne)$ we have
    $(\col(v_1), \ldots, \col(v_k)) \in \sem{\QOne}(\ciD)$.
    \item\label{item:b:lemma:mainLemma}
    For all $\tup{c} = (c_1,\ldots,c_k) \in \sem{\QOne}(\ciD)$, and for every $v_1 \in \Adom(\DOne)$ with $\col(v_1) = c_1$,
    $(v_1)$ is a partial output of $\QOne$ over $\DOne$ of color $\tup{c}$.
    Moreover, if $(v_1,\ldots,v_i)$ is a partial output of $\QOne$ over $\DOne$ of color $\tup{c}$ and $x_j = \Parent(x_{i+1})$,
    then $\N{v_j}{c_{i+1}} \neq \emptyset$ and $(v_1, \ldots, v_i, v_{i+1})$ is a partial output of $\QOne$ over $\DOne$ of color $\tup{c}$ for every $v_{i+1} \in \N{v_j}{c_{i+1}}$.
  \end{enumerate}
\end{restatable}

The proof uses Lemma~\ref{lemma:homomorphisms}, our particular choice of $\ciD$,
and the fact that $\col$ is a stable coloring of $\GOne$ that refines $\vl$
(see appendices~\ref{app:eval:observations} and~\ref{app:proof:mainLemma} for details).
For solving the task ``$\enumtask$'', we use Theorem~\ref{thm:BGS-enum} with input $\ciD$ and $\QOne$ to carry out a preprocessing phase in time $O(|\QOne|{\cdot}|\ciD|)$ and then enumerate the tuples in $\sem{\QOne}(\ciD)$ with delay $O(k)$.
Each time we receive a tuple $\tup{c}=(c_1,\ldots,c_k) \in \sem{\QOne}(\ciD)$,
we carry out the following algorithm:
\begin{quote}
  for all $v_1\in \VOne$ with $\col(v_1) = c_1$ do\ $\myEnum((v_1); \tup{c})$.
\end{quote}
\noindent
Here, for all $i \in [k]$, the procedure $\myEnum((v_1,\ldots,v_i);\tup{c})$ is as follows:
\begin{quote}
  if $i=k$ then \textbf{output} $(v_1,\ldots,v_k)$  \\
  else
  \\
  \mbox{ \ \ \ }let $x_j \deff p(x_{i+1})$ \\
  \mbox{ \ \ \ }for all $v_{i+1}\in \N{v_j}{c_{i+1}}$
  do \
  $\myEnum((v_1,\ldots,v_i,v_{i+1});\tup{c})$\\
  endelse.
\end{quote}

\noindent
Clearly, for each fixed $\tup{c} = (c_1, \ldots, c_k) \in \sem{\QOne}(\ciD)$, this outputs,
without repetition, tuples $(v_1,\ldots,v_k) \in \VOne^k$ such that $(\col(v_1), \ldots, \col(v_k)) = \tup{c}$.
From Lemma~\ref{lemma:mainLemma}\ref{item:b:lemma:mainLemma} we obtain that the output contains all tuples $(v_1,\ldots,v_k)$ with color $(\col(v_1),\ldots,\col(v_k))\allowbreak = \tup{c}$ that belong to $\sem{\QOne}(\DOne)$.
Using Lemma~\ref{lemma:mainLemma}\ref{item:a:lemma:mainLemma}, we then obtain that, in total, the algorithm enumerates exactly all the tuples $(v_1,\ldots,v_k)$ in $\sem{\QOne}(\DOne)$.
From Lemma~\ref{lemma:mainLemma}\ref{item:b:lemma:mainLemma} we know that every time we encounter a loop of the form \enquote{for all $v_{i+1} \in \N{v_j}{c_{i+1}}$}, the set $\N{v_j}{c_{i+1}}$ is non-empty.
Thus, by using the lookup tables built during the indexing phase, we obtain that the delay between outputting any two tuples of $\sem{\QOne}(\DOne)$ is $O(k)$.
In summary, this shows that we can enumerate $\sem{\QOne}(\DOne)$ (which is equal to $\sem{Q}(D)$) with delay $O(k)$ after a preprocessing phase that takes time $O(|Q|{\cdot}|\ciD|)$. %
\paragraph{Evaluation phase for counting (task \upshape{``$\counttask$''}).}
If $\QOne$ is a \emph{Boolean} query, the number of answers is  0 or 1, and the result for the task ``$\counttask$'' is obtained by solving the task ``$\booltask$'' as described above.
If $\QOne$ is a $k$-ary query for some $k \geq 1$, we make the same assumptions and use the same notation as for the task ``$\enumtask$''.
For solving the task ``$\counttask$'', we proceed as follows.
For every $c \in C$, let $v_c$ be the first vertex in the lookup table for the set $\set{v \in \VOne \mid \col(v) = c}$, and for every $x \in \vars(\QOne)$ and every $c \in C$ let $f_1(c,x) \isdef 1$ if $\vl(v_c) \supseteq \lambda_x$, and let $f_1(c,x) \isdef 0$ otherwise.
Note that the lookup table is part of the color-index $\DSD$ and that $f_1(c,x)$ indicates whether one (and thus, every) vertex of color $c$ satisfies all the unary atoms of the form $U(x)$ that occur in $\QOne$.
Recall that $\DSD$ contains a lookup table to access the set $\set{v \in \VOne \mid \col(v) = c}$ for every $c \in C$, and we have already computed a lookup table to access $\lambda_x$ for every $x \in \vars(\QOne)$ in this phase.
Thus we can compute, in time $O(|C| {\cdot} |\QOne|)$, a lookup table that gives us $O(1)$-time access to $f_1(c,x)$ for all $c \in C$ and all $x \in \vars(\QOne)$.

Recall that $T$ denotes the rooted tree obtained from the Gaifman graph $G(\QOne)$ by choosing $x_1$ to be its root.
For every node $x$ of $T$ we let $T_x$ be the subtree of $T$ rooted at $x$.
Via a bottom-up pass of $T$ we define the following values for every $c \in C$:
for every \emph{leaf} $x$ of $T$, let $f_{\downarrow}(c,x) \isdef f_1(c,x)$;
for every node $y \neq x_1$ of $T$, let $g(c,y) \isdef \sum_{c' \in C}  f_{\downarrow}(c', y) \cdot \numN{c}{c'}$;
and for every inner node $x$ of $T$, let $f_{\downarrow}(c,x) \isdef f_1(c,x) \cdot \prod_{y \in \Children(x)} g(c,y)$.
It is easy to see that lookup tables providing $O(1)$-time access to $f_{\downarrow}$ and $g$ can be computed in total time $O(|\QOne| {\cdot} |\ciD|)$
(to achieve this, do a bottom-up pass over the edges $\smallset{x,y}$ of $T$ and  note that $g(c,y) = \sum_{c'\,:\,(c,c') \in E^{\ciD}}  f_{\downarrow}(c', y) \cdot \numN{c}{c'}$).

We obtain the following lemma by induction (bottom-up along $T$), see Appendix~\ref{app:proof:counting} for a proof.
\begin{restatable}{lemma}{lemmaCounting}\label{lemma:counting}
  For all $(c,v)\in C\times \VOne$  with $c=\col(v)$, the following is true for all $x \in V(T)$:
  \begin{enumerate}[(a)]
    \item\label{item:a:lemma:counting}
    $f_{\downarrow}(c,x)$ is the number of mappings $\val\colon V(T_x) \to \VOne$ satisfying $\val(x) = v$ and
    \begin{enumerate}[(1), topsep=0pt]
      \item
      for every $x' \in V(T_x)$ we have $\vl(\val(x')) \supseteq \lambda_{x'}$,\; and
      \item
      for every edge $\set{x',y'}[]$ in $T_x$ we have $\set{\val(x'),\val(y')}[] \in \EOne$;
    \end{enumerate}
    \item\label{item:b:lemma:counting}
    for all $y \in \Children(x)$, the value $g(c, y)$ is the number of mappings $\val\colon \set{x}[] \cup V(T_y) \to \VOne$ with $\val(x) = v$ and
    \begin{enumerate}[(1), topsep=0pt]
      \item
      for every $x' \in V(T_y)$ we have $\vl(\val(x')) \supseteq \lambda_{x'}$,\; and
      \item
      for every edge $\set{x',y'}[]$ in $T_x$ with $x',y' \in \set{x}[] \cup V(T_y)$ we have $\set{ \val(x'),\val(y') }[] \in \EOne$.
    \end{enumerate}%
  \end{enumerate}%
\end{restatable}
\noindent
Combining Lemmas~\ref{lemma:counting}(a) and~\ref{lemma:homomorphisms} yields that $\sum_{c \in C}\ n_c \cdot f_{\downarrow}(c,x_1)$ is the number of homomorphisms from $\QOne$ to $\DOne$, where $n_c$ is the number of nodes $v \in \VOne$ with $\col(v) = c$.
Since we have $O(1)$ access to $n_c$ in $\DSD$, the number $\sum_{c \in C}\ n_c \cdot f_{\downarrow}(c,x_1)$ can be computed in time $O(|C|)$.
In the particular case where $\free(Q) = \vars(Q)$, the number of homomorphisms from $\QOne$ to $\DOne$ is precisely the number $|\sem{\QOne}(\DOne)|$.
Hence, this solves the task ``$\counttask$'' in case that
$\free(\QOne) = \vars(\QOne)$ in time $O(|\QOne| {\cdot} |\ciD|)$.

In case that $\free(\QOne) \neq \vars(\QOne)$, note that $\free(\QOne)$ induces a subtree $T'$ of $T$ by Proposition~\ref{prop:binaryfcacqs}.
Via a bottom-up pass of $T'$ we can define the following values for every $c \in C$:
for each leaf $x$ of $T'$ we let $f'_{\downarrow}(c,x) \isdef 1$ if $f_\downarrow(c,x) \geq 1$, and $f'_{\downarrow}(c,x) \isdef 0$ otherwise;
for every node $y \neq x_1$ of $T'$ let $g'(c,y) \isdef \sum_{c' \in C} \ f'_{\downarrow}(c',y)  \cdot \ \numN{c}{c'}$;
and for every inner node $x$ of $T'$, let $f'_{\downarrow}(c,x) \isdef f_1(c,x) \cdot \prod_{y \in \Children(x)} g'(c,y)$.
By the same reasoning as before, we can obtain lookup tables for $f'_{\downarrow}$ and $g'$ in time $O(|\QOne| {\cdot} |\ciD|)$.
By induction (bottom-up along $T'$) one obtains: For all $c \in C$ and all $v \in \VOne$ with $\col(v) = c$, the value $f'_{\downarrow}(c,x_1)$ is the number of tuples $(v_1,\ldots,v_k) \in \sem{\QOne}(\DOne)$ with $v_1=v$.
Thus,
$
	|\sem{\QOne}(\DOne)| \,=\, \sum_{c\in C} \, n_c \cdot f'_{\downarrow}(c,x_1).
$
This number can be computed in time $O(|\QOne|{\cdot}|\ciD|)$ by the same reasoning as before.
In summary, we can compute the number $|\sem{\QOne}(\DOne)|$ in time $O(|\QOne|{\cdot}|\ciD|)$.
  
In summary, this completes the proof of Theorem~\ref{thm:maintheorem} for the special case where $\sigma$ is a schema for node-labeled graphs and where the given $\sigma$-db $D$ is a node-labeled graph (i.e., $E^D$ is symmetric).

\section{Wrapping Up: Proof of Main Theorem, Size of \texorpdfstring{\textit{D}\textsubscript{col}}{D\_col}, and Open Questions}%
\label{sec:SizeOfIndex}\label{sec:final}
In Section~\ref{sec:OneBinaryRelation} we have proved Theorem~\ref{thm:maintheorem} for the special case where $\sigma$ is a schema for node-labeled graphs and $D$ is a node-labeled graph (i.e., $E^D$ is symmetric).
Using this and applying Theorem~\ref{thm:ReductionBinarySchemaToBinarySchemaWithOnlyOneBinaryRelation} yields a proof of Theorem~\ref{thm:maintheorem} for the case where $\sigma$ is an arbitrary binary schema and $D$ is an arbitrary $\sigma$-db.
And using that and applying Theorem~\ref{thm:ReductionArbSchemaToBinarySchema} yields a proof of Theorem~\ref{thm:maintheorem} for the general case where $\sigma$ is an arbitrary relational schema and $D$ is an arbitrary $\sigma$-db:
during the indexing phase, we first translate $\sigma$ and $D$ into $\sigmaBinary$ and $\dbBinary$ according to Theorem~\ref{thm:ReductionArbSchemaToBinarySchema},
then translate these into $\sigmaSimple$ and $\dbSimple$ according to Theorem~\ref{thm:ReductionBinarySchemaToBinarySchemaWithOnlyOneBinaryRelation},
and then carry out the indexing phase for the latter as described in Section~\ref{sec:OneBinaryRelation}.
Among other things, this yields the auxiliary database $\ciD$.

It turns out that the size of $\ciD$ is tightly related to the number of colors that \emph{relational color refinement} (RCR) assigns to the original $\sigma$-db $D$.
RCR was introduced by Scheidt and Schweikardt in \cite{ScheidtSchweikardt_MFCS25}.
It is a generalization of classical color refinement (CR) that works on arbitrary relational structures,
and that is equivalent to CR in the special case that the relational structure is a graph.
The key aspect of RCR for this paper is that the coloring it produces has the equivalent property with respect to acyclic $\sigma$-structures as the coloring produced by CR on a graph with respect to trees.
CR assigns two nodes $u$, $v$ of a graph $G$ different colors if, and only if, there is a tree $T$ with root $r$ for which the number of homomorphisms from $T$ to $G$ mapping $r$ to $u$ differs from the number of homomorphisms from $T$ to $G$ mapping $r$ to $v$ \cite{Dvorak2010}.
Similarly, RCR assigns two tuples $\at$, $\bt$ of a $\sigma$-db $D$ different colors if, and only if, there is an acyclic $\sigma$-db $\C$ with a tuple $\ct$ for which the number of homomorphisms from $\C$ to $D$ mapping $\ct$ to $\at$ differs from the number of homomorphisms from $\C$ to $D$ mapping $\ct$ to $\bt$~\cite{ScheidtSchweikardt_MFCS25}.
The following theorem shows, for fixed arbitrary $\sigma$, that the size of $\adom{\ciD}$ is linear in the number of colors that RCR produces on $D$
(consult Appendix~\ref{app:size-of-index} for details).
\begin{theorem}\label{thm:size-of-index}
	Let $\sigma$ be an arbitrary schema.
	For every $\sigma$-db $D$ we have	$|\Adom(\ciD)| = O(|\rcrcols|)$,
	where $\rcrcols$ is the set of colors produced on $D$ by the relational color refinement of~\cite{ScheidtSchweikardt_MFCS25}.
\end{theorem}

For undirected self-loop-free and unlabeled graphs $G$ it is known that the coarsest stable coloring assigned to $G$ by CR has $\leq k$ colors if, and only if, $G$ has an \emph{equitable partition} of size $\leq k$ (cf., e.g.,~\cite{ScheinermanBook}),
i.e., $V(G)$ can be partitioned into $k$ disjoint sets $V_1, \ldots, V_k$ such that for each $i \in [k]$ the induced subgraph $G[V_i]$ is \emph{regular}
(i.e., all vertices have the same degree) and for all $i,j \in [k]$ with $i \neq j$,
all vertices in $V_i$ have exactly the same number of neighbors in $V_j$.
Thus, for any function $f(n) \in o(n)$ there is a large class of databases whose active domain has size $n$ (for arbitrarily large~$n$) and whose color database has a reduced active domain of size of order only $f(n)$:
let $\mathcal{C}_f$ be the class of all graphs on $n$ nodes (for any $n \in \NN$) that have an equitable partition of size $\leq f(n)$.
The active domain of the color database of each graph $G$ in $\mathcal{C}_f$ on $n$ nodes then has size $\leq f(n)$.
In particular, for $f(n) = 1$, the class $\mathcal{C}_f$ consists of all regular graphs (this includes, for example, all cycles);
their color databases consist of a single color, i.e.\ have size $O(1)$.
The class $\mathcal{C}_{\log(n)}$ contains, among others, all rooted trees of height $\log(n)$ where for each height $i$ all nodes of height $i$ have the same degree;
their color databases have size $O(\log(n))$.

It is known~\cite{babai1980random} that on random graphs, CR assigns with high probability a new color to each node, and hence its color database is not smaller than the graph itself.
But actual databases are usually designed to store structured information and are expected to look substantially different from a random graph.
Empirical results by Kersting et al.~\cite{Kersting2014} verify this expectation:
They compute the number of colors for various datasets, including (among others) a web graph from Google where the ratio between \enquote{number of colors} and \enquote{number of nodes} is $0.4$;
for most of their further results it is between $0.3$ and $0.6$.
\smallskip

\noindent
We close the paper with two questions for future research.
\begin{description}
	\item[Question 1]
	\emph{Can our approach be lifted from fc-ACQs to queries of free-connex generalized hypertree width $\leq k$, for any fixed $k \geq 1$?}
	When restricting attention to binary schemas, we believe that this might be achieved by using a variant of the $k$-dimensional Weisfeiler-Leman algorithm (cf.~\cite{CFI-paper,grohe_color_2021}).
	But lifting it to arbitrary schemas might be difficult, since currently no generalization of RCR from acyclic relational structures to structures of generalized hypertree width $k$ is known.

	\item[Question 2]
	\emph{Can our approach be lifted to a dynamic scenario?}
	The \emph{Dynamic Yannakakis} approach of Idris, Ugarte, and Vansummeren~\cite{DynamicYannakakis2017} lifts Theorem~\ref{thm:BGS-enum} to the scenario,
	where a fixed query $Q$ shall be evaluated against a database $D$ that is frequently updated.
	Can our index structure $\DSD$ be updated in time proportional to the actual difference between the two versions of $\DSD$ before and after the update?
\end{description}

\bibliography{literature}

\begin{thebibliography}{10}

\bibitem{AHV-Book}
Serge Abiteboul, Richard Hull, and Victor Vianu.
\newblock {\em Foundations of Databases}.
\newblock Addison-Wesley, 1995.
\newblock URL: \url{http://webdam.inria.fr/Alice/}.

\bibitem{AnglesABHRV17}
Renzo Angles, Marcelo Arenas, Pablo Barcel{\'{o}}, Aidan Hogan, Juan~L.
  Reutter, and Domagoj Vrgoc.
\newblock Foundations of modern query languages for graph databases.
\newblock {\em {ACM} Comput. Surv.}, 50(5):68:1--68:40, 2017.

\bibitem{ArroyueloHNRRS21}
Diego Arroyuelo, Aidan Hogan, Gonzalo Navarro, Juan~L. Reutter, Javiel
  Rojas{-}Ledesma, and Adri{\'{a}}n Soto.
\newblock Worst-case optimal graph joins in almost no space.
\newblock In {\em SIGMOD}, pages 102--114. {ACM}, 2021.

\bibitem{Arvind2017}
V.~Arvind, Johannes Köbler, Gaurav Rattan, and Oleg Verbitsky.
\newblock Graph {{Isomorphism}}, {{Color Refinement}}, and {{Compactness}}.
\newblock {\em Computational Complexity}, 26(3):627--685, September 2017.
\newblock \href {https://doi.org/10.1007/s00037-016-0147-6}
  {\path{doi:10.1007/s00037-016-0147-6}}.

\bibitem{babai1980random}
L{\'{a}}szl{\'{o}} Babai, Paul Erd{\"{o}}s, and Stanley~M. Selkow.
\newblock Random graph isomorphism.
\newblock {\em {SIAM} J. Comput.}, 9(3):628--635, 1980.
\newblock \href {https://doi.org/10.1137/0209047} {\path{doi:10.1137/0209047}}.

\bibitem{Bagan_PhD}
Guillaume Bagan.
\newblock {\em Algorithmes et complexit{\'{e}} des probl{\`{e}}mes
  d'{\'{e}}num{\'{e}}ration pour l'{\'{e}}valuation de requ{\^{e}}tes logiques.
  (Algorithms and complexity of enumeration problems for the evaluation of
  logical queries)}.
\newblock PhD thesis, University of Caen Normandy, France, 2009.
\newblock URL: \url{https://tel.archives-ouvertes.fr/tel-00424232}.

\bibitem{Bagan.2007}
Guillaume Bagan, Arnaud Durand, and Etienne Grandjean.
\newblock On acyclic conjunctive queries and constant delay enumeration.
\newblock In {\em Proceedings of the 16th Annual Conference of the EACSL,
  CSL'07, Lausanne, Switzerland, September 11--15, 2007}, pages 208--222, 2007.
\newblock \href {https://doi.org/10.1007/978-3-540-74915-8_18}
  {\path{doi:10.1007/978-3-540-74915-8_18}}.

\bibitem{DBLP:journals/jacm/BeeriFMY83}
Catriel Beeri, Ronald Fagin, David Maier, and Mihalis Yannakakis.
\newblock On the desirability of acyclic database schemes.
\newblock {\em J. {ACM}}, 30(3):479--513, 1983.
\newblock \href {https://doi.org/10.1145/2402.322389}
  {\path{doi:10.1145/2402.322389}}.

\bibitem{BBG-ColorRefinement}
Christoph Berkholz, Paul~S. Bonsma, and Martin Grohe.
\newblock Tight lower and upper bounds for the complexity of canonical colour
  refinement.
\newblock {\em Theory Comput. Syst.}, 60(4):581--614, 2017.
\newblock \href {https://doi.org/10.1007/S00224-016-9686-0}
  {\path{doi:10.1007/S00224-016-9686-0}}.

\bibitem{BGS-tutorial}
Christoph Berkholz, Fabian Gerhardt, and Nicole Schweikardt.
\newblock Constant delay enumeration for conjunctive queries: a tutorial.
\newblock {\em {ACM} {SIGLOG} News}, 7(1):4--33, 2020.
\newblock \href {https://doi.org/10.1145/3385634.3385636}
  {\path{doi:10.1145/3385634.3385636}}.

\bibitem{DBLP:journals/siamcomp/BernsteinG81}
Philip~A. Bernstein and Nathan Goodman.
\newblock Power of natural semijoins.
\newblock {\em {SIAM} J. Comput.}, 10(4):751--771, 1981.
\newblock \href {https://doi.org/10.1137/0210059} {\path{doi:10.1137/0210059}}.

\bibitem{Bollen2023}
Jeroen Bollen, Jasper Steegmans, Jan Van Den~Bussche, and Stijn Vansummeren.
\newblock {L}earning {G}raph {N}eural {N}etworks using {E}xact {C}ompression.
\newblock In {\em Proceedings of the 6th Joint Workshop on Graph Data
  Management Experiences \& Systems (GRADES) and Network Data Analytics (NDA)},
  GRADES \& NDA '23, New York, NY, USA, 2023. ACM.
\newblock \href {https://doi.org/10.1145/3594778.3594878}
  {\path{doi:10.1145/3594778.3594878}}.

\bibitem{BraultBaron_PhD}
Johann Brault{-}Baron.
\newblock {\em De la pertinence de l'{\'{e}}num{\'{e}}ration : complexit{\'{e}}
  en logiques propositionnelle et du premier ordre. (The relevance of the list:
  propositional logic and complexity of the first order)}.
\newblock PhD thesis, University of Caen Normandy, France, 2013.
\newblock URL: \url{https://tel.archives-ouvertes.fr/tel-01081392}.

\bibitem{BraultBaron16}
Johann Brault-Baron.
\newblock {H}ypergraph {A}cyclicity {R}evisited.
\newblock {\em {ACM} Comput. Surv.}, 49(3):54:1--54:26, 2016.
\newblock \href {https://doi.org/10.1145/2983573} {\path{doi:10.1145/2983573}}.

\bibitem{CFI-paper}
Jin{-}yi Cai, Martin F{\"{u}}rer, and Neil Immerman.
\newblock An optimal lower bound on the number of variables for graph
  identification.
\newblock {\em Combinatorica}, 12(4):389--410, 1992.
\newblock \href {https://doi.org/10.1007/BF01305232}
  {\path{doi:10.1007/BF01305232}}.

\bibitem{Cardon1982}
A.\ Cardon and Maxime Crochemore.
\newblock {P}artitioning a {G}raph in {$O(|A| \log_2 |V|)$}.
\newblock {\em Theor. Comput. Sci.}, 19:85--98, 1982.
\newblock \href {https://doi.org/10.1016/0304-3975(82)90016-0}
  {\path{doi:10.1016/0304-3975(82)90016-0}}.

\bibitem{Deeds2024}
Kyle~B. Deeds, Diandre Sabale, Moe Kayali, and Dan Suciu.
\newblock {COLOR:} {A} framework for applying graph coloring to subgraph
  cardinality estimation.
\newblock {\em Proc. {VLDB} Endow.}, 18(2):130--143, 2024.
\newblock \href {https://doi.org/10.14778/3705829.3705834}
  {\path{doi:10.14778/3705829.3705834}}.

\bibitem{Dvorak2010}
Zdeněk Dvořák.
\newblock {O}n {R}ecognizing {G}raphs by {N}umbers of {H}omomorphisms.
\newblock {\em Journal of Graph Theory}, 64(4):330--342, 2010.

\bibitem{Goebel2024}
Andreas G{\"{o}}bel, Leslie~Ann Goldberg, and Marc Roth.
\newblock The {W}eisfeiler-{L}eman {D}imension of {C}onjunctive {Q}ueries.
\newblock {\em Proc. {ACM} Manag. Data}, 2(2):86, 2024.
\newblock Proc.~PODS'24.
\newblock \href {https://doi.org/10.1145/3651587} {\path{doi:10.1145/3651587}}.

\bibitem{DBLP:journals/jcss/GottlobLS02}
Georg Gottlob, Nicola Leone, and Francesco Scarcello.
\newblock Hypertree decompositions and tractable queries.
\newblock {\em J. Comput. Syst. Sci.}, 64(3):579--627, 2002.
\newblock \href {https://doi.org/10.1006/jcss.2001.1809}
  {\path{doi:10.1006/jcss.2001.1809}}.

\bibitem{Grohe2020a}
Martin Grohe.
\newblock {W}ord2vec, {N}ode2vec, {G}raph2vec, {X}2vec: {T}owards a {T}heory of
  {V}ector {E}mbeddings of {S}tructured {D}ata.
\newblock In {\em Proceedings of the 39th ACM SIGMOD-SIGACT-SIGAI Symposium on
  Principles of Database Systems}, PODS'20, pages 1--16. ACM, 2020.
\newblock \href {https://doi.org/10.1145/3375395.3387641}
  {\path{doi:10.1145/3375395.3387641}}.

\bibitem{grohe_color_2021}
Martin Grohe, Kristian Kersting, Martin Mladenov, and Pascal Schweitzer.
\newblock Color {Refinement} and {Its} {Applications}.
\newblock In Guy Van~den Broeck, Kristian Kersting, Sriraam Natarajan, and
  David Poole, editors, {\em An {Introduction} to {Lifted} {Probabilistic}
  {Inference}}. The MIT Press, 2021.
\newblock \href {https://doi.org/10.7551/mitpress/10548.003.0023}
  {\path{doi:10.7551/mitpress/10548.003.0023}}.

\bibitem{grohe_dimension_2014}
Martin Grohe, Kristian Kersting, Martin Mladenov, and Erkal Selman.
\newblock Dimension {Reduction} via {Colour} {Refinement}.
\newblock In Andreas~S. Schulz and Dorothea Wagner, editors, {\em Algorithms -
  {ESA} 2014}, Lecture {Notes} in {Computer} {Science}, pages 505--516, Berlin,
  Heidelberg, 2014. Springer.
\newblock \href {https://doi.org/10.1007/978-3-662-44777-2_42}
  {\path{doi:10.1007/978-3-662-44777-2_42}}.

\bibitem{DynamicYannakakis2017}
Muhammad Idris, Mart{\'{\i}}n Ugarte, and Stijn Vansummeren.
\newblock The {D}ynamic {Y}annakakis algorithm: Compact and efficient query
  processing under updates.
\newblock In {\em Proc.\ 2017 {ACM} International Conference on Management of
  Data ({SIGMOD} Conference 2017), Chicago, IL, USA, May 14--19, 2017}, pages
  1259--1274, 2017.
\newblock \href {https://doi.org/10.1145/3035918.3064027}
  {\path{doi:10.1145/3035918.3064027}}.

\bibitem{DBLP:journals/pvldb/IdrisUVVL18}
Muhammad Idris, Mart{\'{\i}}n Ugarte, Stijn Vansummeren, Hannes Voigt, and
  Wolfgang Lehner.
\newblock Conjunctive queries with inequalities under updates.
\newblock {\em {PVLDB}}, 11(7):733--745, 2018.
\newblock \href {https://doi.org/10.14778/3192965.3192966}
  {\path{doi:10.14778/3192965.3192966}}.

\bibitem{DBLP:journals/sigmod/IdrisUVVL19}
Muhammad Idris, Mart{\'{\i}}n Ugarte, Stijn Vansummeren, Hannes Voigt, and
  Wolfgang Lehner.
\newblock Efficient query processing for dynamically changing datasets.
\newblock {\em {SIGMOD} Record}, 48(1):33--40, 2019.
\newblock \href {https://doi.org/10.1145/3371316.3371325}
  {\path{doi:10.1145/3371316.3371325}}.

\bibitem{Immerman1990}
Neil Immerman and Eric Lander.
\newblock Describing {{Graphs}}: {{A First-Order Approach}} to {{Graph
  Canonization}}.
\newblock In Alan~L. Selman, editor, {\em Complexity {{Theory Retrospective}}:
  {{In Honor}} of {{Juris Hartmanis}} on the {{Occasion}} of {{His Sixtieth
  Birthday}}, {{July}} 5, 1988}, pages 59--81. {Springer}, {New York, NY},
  1990.
\newblock \href {https://doi.org/10.1007/978-1-4612-4478-3_5}
  {\path{doi:10.1007/978-1-4612-4478-3_5}}.

\bibitem{Kayali2022}
Moe Kayali and Dan Suciu.
\newblock Quasi-{{Stable Coloring}} for {{Graph Compression}}: {{Approximating
  Max-Flow}}, {{Linear Programs}}, and {{Centrality}}.
\newblock In {\em Proceedings of the {{VLDB Endowment}}}, volume~16, pages
  803--815, December 2022.
\newblock \href {https://doi.org/10.14778/3574245.3574264}
  {\path{doi:10.14778/3574245.3574264}}.

\bibitem{Kersting2014}
Kristian Kersting, Martin Mladenov, Roman Garnett, and Martin Grohe.
\newblock Power iterated color refinement.
\newblock In Carla~E. Brodley and Peter Stone, editors, {\em Proceedings of the
  Twenty-Eighth {AAAI} Conference on Artificial Intelligence, July 27 -31,
  2014, Qu{\'{e}}bec City, Qu{\'{e}}bec, Canada}, pages 1904--1910. {AAAI}
  Press, 2014.
\newblock \href {https://doi.org/10.1609/AAAI.V28I1.8992}
  {\path{doi:10.1609/AAAI.V28I1.8992}}.

\bibitem{Kiefer2020}
Sandra Kiefer.
\newblock The {{Weisfeiler-Leman Algorithm}}: {{An Exploration}} of {{Its
  Power}}.
\newblock {\em ACM SIGLOG News}, 7(3):5--27, November 2020.
\newblock \href {https://doi.org/10.1145/3436980.3436982}
  {\path{doi:10.1145/3436980.3436982}}.

\bibitem{Kiefer2021}
Sandra Kiefer, Pascal Schweitzer, and Erkal Selman.
\newblock Graphs {{Identified}} by {{Logics}} with {{Counting}}.
\newblock {\em ACM Transactions on Computational Logic}, 23(1):1:1--1:31,
  October 2021.
\newblock \href {https://doi.org/10.1145/3417515} {\path{doi:10.1145/3417515}}.

\bibitem{ngo2018worst}
Hung~Q Ngo, Ely Porat, Christopher R{\'e}, and Atri Rudra.
\newblock Worst-case optimal join algorithms.
\newblock {\em Journal of the ACM (JACM)}, 65(3):1--40, 2018.

\bibitem{DBLP:journals/tods/OlteanuZ15}
Dan Olteanu and Jakub Z{\'{a}}vodn{\'{y}}.
\newblock Size bounds for factorised representations of query results.
\newblock {\em {ACM} Trans. Database Syst.}, 40(1):2:1--2:44, 2015.
\newblock \href {https://doi.org/10.1145/2656335} {\path{doi:10.1145/2656335}}.

\bibitem{Picalausa2014}
François Picalausa, George H.~L. Fletcher, Jan Hidders, and Stijn Vansummeren.
\newblock Principles of {Guarded Structural Indexing}.
\newblock In Nicole Schweikardt, Vassilis Christophides, and Vincent Leroy,
  editors, {\em Proc. 17th {International Conference} on {Database Theory}
  ({ICDT}), {Athens}, {Greece}, {March} 24-28, 2014}, pages 245--256.
  OpenProceedings.org, 2014.
\newblock \href {https://doi.org/10.5441/002/ICDT.2014.26}
  {\path{doi:10.5441/002/ICDT.2014.26}}.

\bibitem{Picalausa2012}
François Picalausa, Yongming Luo, George H.~L. Fletcher, Jan Hidders, and
  Stijn Vansummeren.
\newblock A {Structural Approach} to {Indexing Triples}.
\newblock In Elena Simperl, Philipp Cimiano, Axel Polleres, Oscar Corcho, and
  Valentina Presutti, editors, {\em The {Semantic Web}: {Research} and
  {Applications}}, Lecture {{Notes}} in {{Computer Science}}, pages 406--421,
  Berlin, Heidelberg, 2012. Springer.
\newblock \href {https://doi.org/10.1007/978-3-642-30284-8_34}
  {\path{doi:10.1007/978-3-642-30284-8_34}}.

\bibitem{ramakrishnan2003database}
Raghu Ramakrishnan and Johannes Gehrke.
\newblock {\em Database management systems}, volume~3.
\newblock McGraw-Hill New York, 2003.

\bibitem{Riveros2024}
Cristian Riveros, Benjamin Scheidt, and Nicole Schweikardt.
\newblock {Using Color Refinement to Boost Enumeration and Counting for Acyclic
  CQs of Binary Schemas}.
\newblock {\em CoRR}, abs/2405.12358, 2024.
\newblock \href {https://arxiv.org/abs/2405.12358} {\path{arXiv:2405.12358}},
  \href {https://doi.org/10.48550/ARXIV.2405.12358}
  {\path{doi:10.48550/ARXIV.2405.12358}}.

\bibitem{ScheidtSchweikardt_MFCS25}
Benjamin Scheidt and Nicole Schweikardt.
\newblock Color refinement for relational structures.
\newblock In {\em 50th International Symposium on Mathematical Foundations of
  Computer Science, {MFCS} 2025, August 25-29, 2025, Warsaw, Poland}, volume
  345 of {\em LIPIcs}, pages 88:1--88:19. Schloss Dagstuhl - Leibniz-Zentrum
  f{\"{u}}r Informatik, 2025.
\newblock \href {https://doi.org/10.4230/LIPICS.MFCS.2025.88}
  {\path{doi:10.4230/LIPICS.MFCS.2025.88}}.

\bibitem{ScheinermanBook}
Edward~R. Scheinerman and Daniel~H. Ullman.
\newblock {\em Fractional Graph Theory: A Rational Approach to the Theory of
  Graphs}.
\newblock Wiley-Interscience Series in Discrete Mathematics and Optimization.
  Wiley: John Wiley \& Sons. 211 p., 1997.
\newblock Available at \url{https://www.ams.jhu.edu/ers/books/}.

\bibitem{DBLP:conf/icdt/Veldhuizen14}
Todd~L. Veldhuizen.
\newblock Triejoin: {A} simple, worst-case optimal join algorithm.
\newblock In Nicole Schweikardt, Vassilis Christophides, and Vincent Leroy,
  editors, {\em Proc. 17th International Conference on Database Theory (ICDT),
  Athens, Greece, March 24-28, 2014}, pages 96--106. OpenProceedings.org, 2014.
\newblock \href {https://doi.org/10.5441/002/ICDT.2014.13}
  {\path{doi:10.5441/002/ICDT.2014.13}}.

\bibitem{Yannakakis1981}
Mihalis Yannakakis.
\newblock Algorithms for acyclic database schemes.
\newblock In {\em Very Large Data Bases, 7th International Conference,
  September 9-11, 1981, Cannes, France, Proceedings}, pages 82--94, 1981.

\end{thebibliography}

\clearpage

\appendix

\onecolumn

\section*{APPENDIX}
\section{Details Omitted in Section~\ref{sec:ACQs}}\label{appendix:fcACQ}

\subsection{Proof of Proposition~\ref{prop:binaryfcacqs}}\label{app:proof:prop:binaryfcacqs}
\propBinaryFcAcqs*
Let $\sigma$ be a binary schema, i.e., every $R\in\sigma$ has arity $\ar(R)\leq 2$.
In the following, if $G$ is a (hyper)graph, then we denote by $G - e$ the resulting (hyper)graph
after removing the (hyper)edge $e \in E(G)$ from $G$.
If a vertex becomes isolated due to this procedure, we remove it as well.
Further, we denote by $G + e$ the (hyper)graph resulting from adding the (hyper)edge $e$ to $G$.
First we show that we can assume that $Q$ contains no loops
and only uses predicates of arity exactly 2.
\begin{claim}
	$Q$ is free-connex acyclic iff $Q'$ is free-connex acyclic,
	where \[ \Atoms(Q') \isdef \set{ R(x,y) \mid R(x,y) \in \Atoms(Q),\ x {\neq} y }\]
	and $\free(Q') \isdef \free(Q) \cap \Vars(Q')$.
\end{claim}
\begin{proof}
	\noindent\enquote{$\Longrightarrow$}: \
	Let $Q$ be free-connex acyclic.
	Then $H(Q) + \free(Q)$ has a join-tree $T$.
	We translate $T$ into a join-tree $T'$ for $H(Q')+\free(Q')$ as follows.
	First, we replace the node $\free(Q)$ in $T$ with $\free(Q')$.
	Then, for every hyperedge $e = \set{ x }$,
	we look for a neighbor $f$ of $e$ in $T$ such that $x \in f$.
	If there is none, we let $f$ be an arbitrary neighbor.
	Then we connect all other neighbors of $e$ in $T$ with $f$ and remove $e$ from $T$.
	The resulting graph is still a tree and the set of nodes containing $x$ still
	induces a connected subtree, or it is empty --- but in that case also $x \not\in \Vars(Q')$.
	I.e., now the set of nodes in $T'$ is precisely the set of hyperedges of $H(Q')+\free(Q')$,
	and for every $x \in \Vars(Q')$,
	the set $\set{ t \in V(T') \mid x \in t }$ still induces a connected subtree.
	Hence, $T'$ is a join-tree for $H(Q') + \free(Q')$.

	Since $Q$ is free-connex acyclic, there also exists a join-tree $T$ for $H(Q)$.
	With the same reasoning as above, $T$ can be transformed into a join-tree for $H(Q')$.
	In summary, we obtain that $Q'$ is free-connex acyclic.
	\medskip

	\noindent\enquote{$\Longleftarrow$}: \
	Let $Q'$ be free-connex acyclic.
	Then $H(Q') + \free(Q')$ has a join-tree $T'$.
	We can extend $T'$ to a join-tree $T$ for $H(Q) + \free(Q)$ as follows.
	First, we replace the node $\free(Q')$ in $T'$ with $\free(Q)$.
	For every hyperedge $e$ in $H(Q)$ of the form $e = \set{ x }$
	we add $e$ as a new node to $T'$ and insert the edge $\set{ e, f }$ into $T'$ as follows:
	If $x \not\in \Vars(Q')$ we let $f$ be an arbitrary node in $V(T')$.
	Otherwise, let $f \in V(T')$ be a node such that $x \in f$.
	It is easy to verify that the resulting tree $T$ is a join-tree for $H(Q) + \free(Q)$.
	Since $Q'$ is free-connex acyclic, there also exists a join-tree $T$ for $H(Q')$.
	With the same reasoning as above, $T'$ can be transformed into a join-tree for $H(Q)$.
	In summary, we obtain that $Q$ is free-connex acyclic.
\end{proof}

Due to the above claim, we can assume w.l.o.g.\
that $Q$ contains no self-loops and only uses predicates of arity 2.
Notice that this means that $H(Q) = G(Q)$.
It is well-known that for undirected graphs the notion of
$\alpha$-acyclicity and the standard notion of acyclicity for graphs coincide.
I.e., a graph is $\alpha$-acyclic iff it is acyclic (see~\cite{BraultBaron16} for an overview).
Hence, in our setting, $H(Q)=G(Q)$ is $\alpha$-acyclic,
if and only if it is acyclic, i.e., a forest.
This immediately yields the first statement of Proposition~\ref{prop:binaryfcacqs}.
The second statement of Proposition~\ref{prop:binaryfcacqs} is obtained by the following claim.

\begin{claim}\label{appendix:claim:1forprop}
	$Q$ is free-connex acyclic iff $G(Q)$ is acyclic and the following statement is true:
	(**) For every connected component $C$ of $G(Q)$,
	the subgraph of $C$ induced by the set $\free(Q) \cap V(C)$ is connected or empty.
\end{claim}

\begin{proof}
\noindent\enquote{$\Longrightarrow$}: \
By assumption, $Q$ is free-connex acyclic.
I.e., $H(Q)$ is $\alpha$-acyclic and $H(Q)+\free(Q)$ is $\alpha$-acyclic.
Since $H(Q) = G(Q)$, and since on graphs $\alpha$-acyclicity coincides with acyclicity,
this means that $G(Q)$ is acyclic, i.e., it is a forest.

Since $H(Q)+\free(Q)$ is $\alpha$-acyclic, there exists a join-tree $T$ for $H(Q)+\free(Q)$.
To prove (**) we proceed by induction on the number of nodes in the join-tree of $H(Q) + \free(Q)$.

If $T$ has at most two nodes, $H(Q)$ consists of a single hyperedge,
and this is of the form $\set{ x,y }$ with $x \neq y$.
Therefore, $G(Q)$ consists of a single edge $\set{ x,y }$, and in this case (**) trivially holds.

In the inductive case, let $T$ be a join tree of $H(Q) + \free(Q)$.
Consider some leaf $e$ of $T$ (i.e, $e$ only has one neighbor in $T$)
of the form $e = \set{ x,y }$ and its parent $p$.
Then, $T - e$ is a join-tree for $H(Q') + \free(Q)$ where
$\Atoms(Q') \isdef \set{ R(z_1, z_2) \in \Atoms(Q) \mid \set{ z_1, z_2 } \neq \set{ x, y } }$
and $\free(Q')$ is the set of all variables in $\free(Q)$ that occur in an atom of $Q'$.
It is easy to see that this implies the existence of a join-tree $T'$ for $H(Q') + \free(Q')$.
Since, furthermore, $G(Q')$ is acyclic, the query $Q'$ is free-connex acyclic,
and by induction hypothesis, (**) is true for $G(Q')$.
Note that $G(Q')=G(Q)-e$.
Obviously, $x$ and $y$ are in the same connected component $C$ in $G(Q)$.
We have to consider the relationship between $x$, $y$ and the rest of $C$.

Since $e$ is a leaf in $T$, $x,y \in p$ would imply that $p = \free(Q)$.
Using that we know that
\begin{enumerate}[(i)]
	\item either $y$ only has $x$ as a neighbor in $G(Q)$ or vice-versa, or
	\item $x$ and $y$ are both free variables and both have other neighbors
	than $y$ and $x$, respectively.
\end{enumerate}

\emph{Case (i)}:\; Assume w.l.o.g.\ that $x$ is the only vertex adjacent to $y$ in $G(Q)$.
Then $C' \isdef C - \set{ x,y }$ is a connected component in $G(Q')$, $y \not\in \Vars(Q')$,
i.e., in particular also $y \not\in \free(Q')$.
By induction hypothesis, $\free(Q') \cap V(C')$ induces a connected subgraph on $C'$.
If $y \not\in \free(Q)$ it follows trivially,
that $\free(Q) \cap V(C)$ induces a connected subgraph on $C$.
If $y \in \free(Q)$ and $\free(Q) \cap V(C) = \set{ y }$, this is also trivial.
Otherwise, there exists a $z \in \free(Q) \cap V(C)$ different from $y$.
Since $x,y,z$ are all in the same connected component $C$,
but $x$ is the only vertex adjacent to $y$,
there must be another vertex $w$ that is adjacent to $x$.
Therefore, $x$ is part of another node $\set{ x, w }$ somewhere in $T$,
which by definition means that $x$ must be in $p$.
Since $y, z \in \free(Q)$, $y$ must also be in $p$.
Therefore, $p = \free(Q)$ and that means $x \in \free(Q)$.
Hence, $\free(Q) \cap V(C)$ forms a connected component on $C$ in this case as well.

\emph{Case (ii)}:\; Assume that both variables $x$ and $y$ are free
and both have another vertex adjacent to them in $G(Q)$.
Then $x,y \in \free(Q')$ (i.e., $\free(Q) = \free(Q')$) and removing the edge $\set{ x,y }$ splits
$C$ into two connected components $C_x$, $C_y$.
By induction hypothesis, $\free(Q') \cap V(C_x)$ and $\free(Q') \cap V(C_y)$ induce connected
subgraphs on $C_x$ and $C_y$, respectively, and $x$ (and $y$, resp.) are part of them.
Thus, adding the edge $\set{ x,y }$ establishes a connection between them,
so $\free(Q) \cap V(C)$ also induces a connected subtree on $C$.

Every other connected component $D \neq C$ in $G(Q)$ is also one in $G(Q')$.
Thus, (**) is true for $G(Q)$.
\medskip

\noindent\enquote{$\Longleftarrow$}: \
By assumption, $G(Q)$ is acyclic and (**) is satisfied.
Since $G(Q)=H(Q)$ and $\alpha$-acyclicity coincides with acyclicity on graphs, $Q$ is acyclic.
It remains to show that $H(Q)+\free(Q)$ has a join-tree.

We proceed by induction over the number of vertices in $G(Q)$.
Recall that $Q$ only uses binary relation symbols and that it has no self-loops.
Thus, in the base case, $G(Q)$ consists of two vertices connected by an edge.
Since $H(Q) = G(Q)$, it is easy to see that $H(Q) + \free(Q)$ has a join-tree.

In the inductive case, let $C$ be a connected component of $G(Q)$
that contains a quantified variable, i.e., $\free(Q) \cap V(C) \neq V(C)$.
If no such component exists, $H(Q) + \free(Q)$ has a trivial join-tree by letting
all hyperedges of $H(Q)$ be children of $\free(Q)$ in $T$.

Because of (**) there must be a leaf $x$ in $C$ such that $x \not\in \free(Q)$.
Let $y$ be the parent of $x$, i.e., let $\set{ x,y }$ be an (or rather, the only)
edge in $G(Q)$ that contains $x$.
Let $H' \isdef H(Q) - \set{ x,y }$.
Then $H' = H(Q')$ for the query $Q'$ where
$\Atoms(Q') = \set{ R(z_1, z_2)\in\Atoms(Q) \mid x \not\in \set{ z_1, z_2 } }$
and $\free(Q') = \free(Q) \setminus \set{ x }$.

Let $T'$ be the join-tree that exists by induction hypothesis for
$H(Q') + \free(Q')$, i.e., for $H' + \free(Q) \setminus \set{ x }$.

If $x \not\in \free(Q)$, we can extend $T'$ to a join-tree $T$ for $H(Q)$ by
adding the edge $\set{ x,y }$ as a child node to some node that contains $y$ in $T'$.
We can similarly handle the case $\free(Q) = \set{ x }$
by also inserting $\set{ x }$ as a child of $\set{ x,y }$.

If $x \in \free(Q)$ but $\free(Q) \varsupsetneqq \set{ x }$, it holds that $y \in \free(Q)$,
by the assumption that $\free(Q) \cap V(C)$ induces a connected subgraph.
In this case we can add $\set{ x,y }$ as a child of $\free(Q)$.
\end{proof} %
\subsection{Proof of Theorem~\ref{thm:BGS-enum}}%
\label{app:proof:thm:BGS-enum}
\thmBGSenum*
We provide further definitions that are taken from~\cite{BGS-tutorial}.
We use these definitions for the proof of Theorem~\ref{thm:BGS-enum} below.

A \emph{tree decomposition} (TD, for short) of a CQ $Q$ and its hypergraph $H(Q)$
is a tuple $\TD=(\Tree,\Bag)$, such that:
\begin{enumerate}[(1)]
	\item
	$\Tree=(\Nodes(\Tree),\Edges(\Tree))$ is a finite undirected tree, and
	\item
	$\Bag$ is a mapping that associates with every node $\treenode\in\Nodes(\Tree)$
	a set $\Bag(\treenode)\subseteq \Vars(\query)$ such that
	\begin{enumerate}
		\item
		for each atom $\qatom\in\Atoms(\query)$ there exists $\treenode\in\Nodes(\Tree)$
		such that $\Vars(\qatom)\subseteq\Bag(\treenode)$,
		\item\label{item:pathcondition:treedecomp}
		for each variable $v\in\Vars(\query)$ the set
		$\Bag^{-1}(v)\deff\setc{\treenode\in\Nodes(\Tree)}{v\in\Bag(\treenode)}$
		induces a connected subtree of $\Tree$.
	\end{enumerate}
\end{enumerate}
The \emph{width} of $\TD=(\Tree,\Bag)$ is defined as
$\Width(\TD)=\max_{\treenode\in\Nodes(\Tree)}|\Bag(t)| -1$.

A \emph{generalized hypertree decomposition} (GHD, for short) of a CQ $\query$
and its hypergraph $H(Q)$ is a tuple $\HD=(\Tree,\Bag,\Cover)$ that consists of
a tree decomposition $(\Tree,\Bag)$ of $\query$ and a mapping $\Cover$ that associates with
every node $\treenode\in\Nodes(\Tree)$ a set $\Cover(\treenode)\subseteq\Atoms(\query)$ such that
$\Bag(\treenode)\subseteq\bigcup_{\qatom\in\Cover(\treenode)}\Vars(\qatom)$.
The sets $\Bag(\treenode)$ and $\Cover(\treenode)$ are called the \emph{bag} and the \emph{cover}
associated with node $\treenode\in\Nodes(\Tree)$.
The \emph{width} of a GHD $\HD$ is defined as the maximum number of atoms in a $\Cover$-label
of a node of $\Tree$, i.e., $\Width(\HD)=\max_{\treenode\in\Nodes(\Tree)}|\Cover(t)|$.
The \emph{generalized hypertree width} of a CQ $\query$, denoted $\GHW(\query)$,
is defined as the minimum width over all its generalized hypertree decompositions.

A tree decomposition $\TD=(\Tree,\Bag)$ of a CQ $\query$ is \emph{free-connex} if there is a set
$\freetreenodes \subseteq \Nodes(\Tree)$ that induces a connected subtree of $\Tree$
and that satisfies the condition
$\free(\query) = \bigcup_{\treenode\in\freetreenodes} \Bag(\treenode)$.
Such a set $\freetreenodes$ is called a \emph{witness} for the free-connexness of $\TD$.
A GHD is free-connex if its tree decomposition is free-connex.
The \emph{free-connex generalized hypertree width} of a CQ $\query$, denoted $\fcGHW(\query)$,
is defined as the minimum width over all its free-connex generalized hypertree decompositions.

It is known (\cite{DBLP:journals/jcss/GottlobLS02,Bagan.2007,BraultBaron_PhD};
see also~\cite{BGS-tutorial} for an overview as well as proof details)
that the following is true for every schema $\sigma$ and every CQ $Q$ of schema $\sigma$:
\begin{enumerate}[(I)]
	\item\label{item:acqchara}
	$Q$ is acyclic iff $\GHW(Q)=1$.
	\item\label{item:fcacqchara}
	$Q$ is free-connex acyclic iff $\fcGHW(Q)=1$.
\end{enumerate}

A GHD $\HD=(\Tree,\Bag,\Cover)$ of a CQ $\query$ is called \emph{complete} if,
for each atom $\qatom\in\Atoms(\query)$ there exists a node $\treenode\in\Nodes(\Tree)$
such that $\Vars(\qatom)\subseteq\Bag(\treenode)$ and $\qatom\in\Cover(\treenode)$.

The following has been shown in~\cite{BGS-tutorial} (see Theorem~4.2 in~\cite{BGS-tutorial}):
\begin{theorem}[\cite{BGS-tutorial}]\label{thm:BGS-refined}
	For every schema $\sigma$ there is an algorithm which receives as input a query $Q\in\fcACQ$,
	a complete free-connex width~1 GHD $\HD=(\Tree,\Bag,\Cover)$ of $\query$
	along with a witness $\freetreenodes \subseteq\Nodes(\Tree)$,
	and a $\sigma$-db $D$ and computes within preprocessing time $O(|\Nodes(\Tree)|{\cdot}|D|)$
	a data structure that allows to enumerate $\sem{\query}(D)$ with delay $O(|\freetreenodes|)$.
\end{theorem}

It is known that for every schema $\sigma$ there is an algorithm that receives as input
a query $Q\in\fcACQ$ and computes in time $O(|Q|)$ a free-connex width 1 GHD of $Q$
(\cite{Bagan_PhD}; see also Section~5.1 in~\cite{BGS-tutorial}).
When applying to this GHD the construction used in the proof of~\cite[Lemma~3.3]{BGS-tutorial},
and afterwards performing the completion construction from~\cite[Remark~3.1]{BGS-tutorial},
one can compute in time $O(|Q|)$ a GHD $\HD=(\Tree,\Bag,\Cover)$ of $\query$ along with a witness
$\freetreenodes \subseteq\Nodes(\Tree)$ that satisfy the assumptions of
Theorem~\ref{thm:BGS-refined} and for which, additionally,
the following is true: $|U|\leq |\free(Q)|$ and $|\Nodes(\Tree)|\in O(|Q|)$.
In summary, this provides a proof of Theorem~\ref{thm:BGS-enum}.

\section{Details Omitted in Section~\ref{sec:Reductions}}\label{appendix:Reductions}

\subsection{Details Omitted in Section~\ref{sec:ArbitrarySchemas}}\label{appendix:FromArbitraryToBinary}

\subsubsection{An example concerning a non-binary schema}\label{appendix:BinaryToArbitraryIsNontrivial}

At a first glance one may be tempted to believe that
Theorem~\ref{thm:ReductionArbSchemaToBinarySchema} can be proved in a
straightforward way.
However, the notion of fc-ACQs is quite subtle, and it is not so
obvious how to translate an fc-ACQ of an arbitrary schema $\sigma$ 
into an \emph{fc-ACQ} of a suitably chosen binary schema $\sigma''$.
Here is a concrete example.

Let us consider a schema $\sigma$ consisting of a single, ternary relation symbol $R$.
A straightforward way to represent a database $D$ of this schema by
a database $D''$ of a binary schema is as follows.
$D''$ has 3 edge-relations called $E_1, E_2, E_3$.
Every element in the active domain of $D$ is a node of $D''$.
Furthermore, every tuple $t = (a_1,a_2,a_3)$ in the $R$-relation of $D$ serves as a node of $D''$,
and we insert into $D''$ an $E_1$-edge $(t, a_1)$, an $E_2$-edge $(t,a_2)$,
and an $E_3$-edge $(t,a_3)$.
Now, a CQ $Q$ posed against $D$ translates into a CQ $Q''$ posed against $D''$ as follows:
$Q''$ has the same head as $Q$.
For each atom $R(x,y,z)$ in the body of $Q$ we introduce a new variable $u$ and insert
into the body of $Q''$ the atoms $E_1(u,x), E_2(u,y), E_3(u,z)$.

The problem is that a \emph{free-connex acyclic} query $Q$ against $D$ does \emph{not} necessarily translate
into a \emph{free-connex acyclic} query $Q''$ against $D''$ --- and therefore,
proving Theorem~\ref{thm:ReductionArbSchemaToBinarySchema} is not so easy.
Here is a specific example of such a query $Q$:
\[
	\Ans(x,y,z) \ \leftarrow \ R(x,y,z), \ R(x,x,y), \ R(y,y,z), \ R(z,z,x).
\]
This query is a free-connex acyclic (as a witness, take the join-tree whose root is labeled with $R(x,y,z)$
and has 3 children labeled with the remaining atoms in the body of the query).
But the associated query $Q''$ is
\[
\Ans(x,y,z) \ \leftarrow
	\begin{array}[t]{l}
		E_1(u_1,x), E_2(u_1,y), E_3(u_1,z), \\
		E_1(u_2,x), E_2(u_2,x), E_3(u_2,y), \\
		E_1(u_3,y), E_2(u_3,y), E_3(u_3,z), \\
		E_1(u_4,z), E_2(u_4,z), E_3(u_4,x).
	\end{array}
\]
Note that the Gaifman-graph of $Q''$ is not acyclic (it contains the cycle $x-u_2-y-u_3-z-u_4-x$).
Therefore, by Proposition~\ref{prop:binaryfcacqs}, $Q''$ is \emph{not}
free-connex acyclic.
This simple example illustrates that the straightforward encoding of the database $D$
as a database over a binary schema won't help to easily prove Theorem~\ref{thm:ReductionArbSchemaToBinarySchema}.

\subsubsection{Proof Details Omitted in Section~\ref{sec:ArbitrarySchemas}}\label{appendix:ReductionFromArbitraryToBinary}
\constructDbBinary*
\begin{proof}
  By definition, we have $|\sigmaBinary|= |\sigma| + 2{\cdot}k^2+k+1$.
  For every tuple $\at\in\tD$ with $r\deff \ar(\at)$ we have
\[
  |\Projections(\at)|
  \ \leq \
  \sum_{m = 0}^{r} {r \choose m } \cdot m!
  \ \leq \ 
  r \cdot r!
  \ \leq \
  k \cdot k!\,.
\]
Thus, $|\Projections(\tD)|\leq k\cdot k! \cdot |\tD|\leq k\cdot
k!\cdot \dbsize{D}=2^{O(k\cdot \log k)}\cdot \dbsize{D}$.

Clearly, by a single pass over all $R\in\sigma$ and all tuples $\dt\in
R^D$, we can construct the following sets: 
\begin{itemize}
\item
  $(U_R)^{\dbBinary}$,
  for all $R\in\sigma$,
\item
  $\tD$ and
  $\setc{w_{\dt}}{\dt\in\tD}$,
\item
  $\Projections(\dt)$, for all $\dt\in\tD$,
\item
  $\Projections(\tD)$ and
  $\setc{v_{\projection}}{\projection\in\Projections(\tD)}$,
\item  
  $(\ArS{i})^{\dbBinary}$, for all $i\in [0,k]$,
\item
  $(E_{i,j})^{\dbBinary}$, for all $i,j\in [k]$.
\end{itemize}

\noindent
All this can be achieved in time $\poly(k)\cdot k!\cdot \dbsize{D} = 2^{O(k\cdot \log k)}\cdot \dbsize{D}$;
and afterwards we also have available data structures 
that allow us to
enumerate with output-linear delay the elements of
the respective set,
to test in time $O(k)$ whether a given item belongs to one of these
sets, 
to switch in time $O(k)$ between a tuple $\dt\in\tD$ and the
associated node $w_{\dt}$, and to switch in time $O(k)$ between a projection
$\projection\in\Projections(\tD)$ and the associated node $v_{\projection}$.

A brute-force way to construct the sets $(F_{i,j})^{\dbBinary}$
for $i,j\in[k]$ is as follows. 
Initialize them as the empty set $\emptyset$, and then
loop over all $\projection\in\Projections(\tD)$, and let
$X\deff\tset(\projection)$. Loop over all tuples
$\projectionAlt=(q_1,\ldots,q_m)$ with $1\leq m\leq k$ and $\tset(\projectionAlt)\subseteq
X$, check if $\projectionAlt\in\Projections(\tD)$, and if so,
insert the tuple $(v_\projection, v_\projectionAlt)$ into
$(F_{i,j})^{\dbBinary}$ and insert 
$(v_\projectionAlt, v_\projection)$  into $(F_{j,i})^{\dbBinary}$
for all those $i,j\in[k]$ where $1\leq i\leq \ar(\projection)$,
$1\leq j\leq \ar(\projectionAlt)$ and
$\proj_i(\projection)=\proj_j(\projectionAlt)$.

Note that after having performed this algorithm, for all $i,j\in[k]$ the set
$(F_{i,j})^{\dbBinary}$ consists of exactly the intended tuples.
The runnning time is in
$|\Projections(\tD)|\cdot \poly(k)\cdot k! =2^{O(k\cdot \log k)}\cdot \dbsize{D}$.
This completes the proof of Claim~\ref{claim:Construct-dbBinary}.
\end{proof}  

\propComputeGHD*
\begin{proof}
From Bagan \cite{Bagan_PhD}
we obtain an algorithm that upon input of a CQ $\query$ decides in time
$\bigOh(\size{\query})$ whether or not $\query$ is free-connex acyclic, and if
so, outputs an fc-1-GHD for $\query$ (cf.\ also
\cite[Remark~5.3]{BGS-tutorial}).

When applying to this fc-1-GHD the construction used in the proof of~\cite[Lemma~3.3]{BGS-tutorial},
and afterwards performing the completion construction from~\cite[Remark~3.1]{BGS-tutorial},
one can compute in time $O(\size{\query})$ a complete fc-1-GHD $H'=(\Tree,\Bag,\Cover,\Witness)$ of $\query$ with
$|\Witness|\leq |\free(\query)|$.

Finally, we modify $H'$ as follows by considering 
all edges $\smallset{t,p}\in E(T)$. In case that
$\Bag(t)\not\subseteq\Bag(p)$ and $\Bag(t)\not\supseteq\Bag(p)$,
subdivide the edge $\smallset{t,p}$ by introducing a new node
$n_{\smallset{t,p}}$,  replace the edge $\smallset{t,p}$ by two new edges
$\set{t,n_{\smallset{t,p}}}$ and $\set{n_{\smallset{t,p}},p}$, and let
$\Bag(n_{\smallset{t,p}})\deff\Bag(t)\cap\Bag(p)$ and
$\Cover(n_{\smallset{t,p}})\deff\Cover(t)$.
If $\smallset{t,p}\subseteq\Witness$, then insert $n_{\smallset{t,p}}$ into
$\Witness$.

It is straightforward to 
verify that this results in a complete
fc-1-GHD of $\query$ with the desired properties.
This completes the proof of Proposition~\ref{prop:ComputeGHD}.
\end{proof}

\claimConstructQBinary*
\begin{proof}
The above definition obviously yields an algorithm for constructing
$\QBinary$ in time $\bigOh(\size{H})$.
Clearly, $\QBinary$ is a conjunctive query of (the binary) schema
$\sigmaBinary$.
In order to prove that  $\QBinary\in\fcACQBinary$, by
Proposition~\ref{prop:binaryfcacqs} it suffices to show that the Gaifman graph
$G(\QBinary)$ is acyclic and for every connected component $C$ of
$G(\QBinary)$, the subgraph of $C$ induced by the set
$\free(\QBinary)\cap V(C)$ is connected or empty. 

Let $\widetilde{T}$ be the graph obtained from $T$ as follows: we
rename every node $t$ into $\singvarv{t}$, and for every $t\in \myAtoms(T)$
we add to $\singvarv{t}$ a new leaf node called $\singvarw{t}$.
Since $T$ is a tree,  $\widetilde{T}$ is a tree as well.
It can easily be verified that
the Gaifman graph $G(\QBinary)$ is the subgraph of $\widetilde{T}$ 
obtained from $\widetilde{T}$ by deleting those edges
$\smallset{\singvarv{t},\singvarv{t'}}$ where
$\Bag(t)\cap\Bag(t')=\emptyset$.
In particular, $G(\QBinary)$ is acyclic.

According to our definition of $\QBinary$, we have
$\free(\QBinary)=\setc{\singvarv{t}}{t\in \Witness}$.
Since $\Witness$ is a witness for the free-connexness of $H$, the set
$\Witness$ induces a connected subtree of $T$.
Thus, the set
$\setc{\singvarv{t}}{t\in\Witness}$ also induces a connected subtree of $\widetilde{T}$.
This implies that for every connected component $C$ of $G(\QBinary)$,
the subgraph of $C$ induced 
by the set $\setc{\singvarv{t}}{t\in\Witness}\cap V(C)$ 
is connected or empty.
From Proposition~\ref{prop:binaryfcacqs} we obtain that $\QBinary$ is free-connex
acyclic.
This completes the proof of Claim~\ref{claim:ConstructQBinary}.
\end{proof}  

\claimPropertiesOfBeta*
\begin{proof}[Proof of Claim~\ref{claim:PropertiesOfBeta}\ref{item:BetahIsConsistent}] \ \\
Consider an arbitrary $h\in\Hom(\QBinary,\dbBinary)$.
Our proof proceeds in three steps. 
\medskip

\noindent
\emph{Step~1:} \
For all $t\in V(T)$ there exists a 
$\projectionAlt_t\in\Projections(\tD)$ such that
$h(\singvarv{t})=v_{\projectionAlt_t}$ and $\ar(\projectionAlt_t)=\ar(\vtup{t})$. Furthermore, for all
$y\in\Vars(\query)$, all $t,t'\in V(T)$ with
$y\in\Bag(t)\cap\Bag(t')$, and for those
$i\leq \ar(\vtup{t})$ and $i'\leq\ar(\vtup{t'})$ with
$y=\proj_{i}(\vtup{t})=\proj_{i'}(\vtup{t'})$ we have $\proj_i(\projectionAlt_t)=\proj_{i'}(\projectionAlt_{t'})$.
\smallskip

\noindent
\emph{Proof of Step~1:} \
Concerning the first statement, consider an arbitrary $t\in V(T)$ and
note that $\Atoms(\QBinary)$ contains the atom $\ArS{|\Bag(t)|}(\singvarv{t})$.
Since $h\in\Hom(\QBinary,\dbBinary)$, we have $h(\singvarv{t})\in(\ArS{|\Bag(t)|})^{\dbBinary}$.
From the definition of $\dbBinary$ we obtain that
$h(\singvarv{t})=v_{\projectionAlt_t}$ for some
$\projectionAlt_t\in\Projections(\tD)$ with
$\ar(\projectionAlt_t)=|\Bag(t)|=\ar(\vtup{t})$.
\smallskip

Concerning the second statement, consider an arbitrary variable
$y\in\Vars(Q)$ and arbitrary nodes $t,t'\in
V(T)$ with $y\in\Bag(t)\cap\Bag(t')$. Let $i\leq\ar(\vtup{t})$ and
$i'\leq\ar(\vtup{t'})$ such that $y=\proj_i(\vtup{t})=\proj_{i'}(\vtup{t'})$.
We have to show that
$\proj_i(\projectionAlt_t)=\proj_{i'}(\projectionAlt_{t'})$.

Recall that $H$ is an fc-1-GHD and thus, in particular, fulfills the
\emph{path condition}. Thus, $y\in\Bag(t)\cap\Bag(t')$ implies that
$y\in \Bag(t'')$ for every node $t''$ that lies on the path from $t$
to $t'$ in $T$.
Therefore, it suffices to prove the statement for the special case
where $t$ and $t'$ are neighbors in $T$, i.e., $\smallset{t,t'}\in
E(T)$.
From $\smallset{t,t'}\in E(T)$ we obtain
that either $t'$ is the parent of $t$ in $\rootedTree$
or $t$ is the parent of $t'$ in $\rootedTree$.
Therefore, according to our definition of the query $\QBinary$, the
set $\Atoms(\QBinary)$ contains the atom
$F_{i,i'}(\singvarv{t},\singvarv{t'})$ or the atom
$F_{i',i}(\singvarv{t'},\singvarv{t})$.
Since $h\in\Hom(\QBinary,\dbBinary)$, we therefore have:
$(h(\singvarv{t}),h(\singvarv{t'})) \in (F_{i,i'})^{\dbBinary}$ or 
$(h(\singvarv{t'}),h(\singvarv{t}))\in (F_{i',i})^{\dbBinary}$.
Recalling from the first statement of Step~1 that
$h(\singvarv{t})=v_{\projectionAlt_{t}}$ and
$h(\singvarv{t'})=v_{\projectionAlt_{t'}}$, this yields:
$(v_{\projectionAlt_{t}},v_{\projectionAlt_{t'}}) \in (F_{i,i'})^{\dbBinary}$ or 
$(v_{\projectionAlt_{t'}},v_{\projectionAlt_{t}})\in (F_{i',i})^{\dbBinary}$.
From the definition of $\dbBinary$ we obtain: $\proj_{i}(\projectionAlt_{t})=\proj_{i'}(\projectionAlt_{t'})$.
This completes the proof of Step~1.
\medskip

\noindent
Recall from the formulation of Claim~\ref{claim:PropertiesOfBeta}\ref{item:BetahIsConsistent} that $h'\deff\beta(h)$.
\medskip

\noindent
\emph{Step~2:} \ $h(\singvarv{t}) = v_{h'(\vtup{t})}$, for every $t\in V(T)$.
\smallskip

\noindent
\emph{Proof of Step~2:} \
Consider an arbitrary $t\in V(T)$. According to Step~1, there exists a
$\projectionAlt_t\in\Projections(D)$ such that
$\ar(\projectionAlt_t)=\ar(\vtup{t})$ and $h(\singvarv{t})=v_{\projectionAlt_t}$.
We have to show that $\projectionAlt_t=h'(\vtup{t})$.

Let $m\deff\ar(\vtup{t})$, consider an arbitrary $i\in[m]$, and let
$y\deff\proj_i(\vtup{t})$ (i.e., $y$ is the variable occurring at
position $i$ in $\vtup{t}$). We have to show that $\proj_i(\projectionAlt_t)=h'(y)$.

According to our definition of $h'=\beta(h)$ we know that
$h'(y)=\proj_{j_y}(\projection_{h,y})$, where
$\projection_{h,y}\in\Projections(\tD)$ such that
$h(\singvarv{t_y})=v_{\projection_{h,y}}$ and
$\ar(\projection_{h,y})=|\Bag(t_y)|=\ar(\vtup{t_y})$.
Furthermore, by our choice of $j_y$ we know that $y=\proj_{j_y}(\vtup{t_y})$.
Using Step~1 for $t'\deff t_y$  and $i'\deff j_y$, and noting that
$\projectionAlt_{t'}=\projection_{h,y}$, we obtain:
$\proj_i(\projectionAlt_t)=\proj_{i'}(\projectionAlt_{t'})=\proj_{j_y}(\projection_{h,y})=h'(y)$.
This completes the proof of Step~2.
\medskip

\noindent
\emph{Step~3:} \
For every $t\in \myAtoms(T)$ and for $R(\tup{z})\deff \Cover(t)$ we have
$h(\singvarw{t})=w_{h'(\tup{z})}$ and $h'(\tup{z})\in R^D$.\smallskip

\noindent
\emph{Proof of Step~3:} \ Let $t\in \myAtoms(T)$ and $R(\tup{z})\deff\Cover(t)$. 
Let $r\deff\ar(R)$ and $(z_1,\ldots,z_r)\deff \tup{z}$.
By definition of $\myAtoms(T)$ we have $\smallset{z_1,\ldots,z_r}=\Bag(t)$.
Thus,
there is a surjective mapping $f\colon [r]\to [m]$ for $m\deff |\Bag(t)|$, such
that for $(x_1,\ldots,x_m)\deff\vtup{t}$ we have: $(z_1,\ldots,z_r) = (x_{f(1)},\ldots,x_{f(r)})$.

By the definition of $\QBinary$, the set $\Atoms(\QBinary)$ contains
the atoms $(U_R)(\singvarw{t})$ and 
$E_{\nu,f(\nu)}(\singvarw{t},\singvarv{t})$ for all $\nu\in[r]$.
Since $h\in\Hom(\QBinary,\dbBinary)$, we have
$h(\singvarw{t})\in (U_R)^{\dbBinary}$ and
$(h(\singvarw{t}),h(\singvarv{t}))\in (E_{\nu,f(\nu)})^{\dbBinary}$
for all $\nu\in[r]$.

By the definition of $\dbBinary$, there is a tuple
$\at=(a_1,\ldots,a_r)\in R^D$ such that $h(\singvarw{t})=w_{\at}$.
In order to complete the proof, it therefore suffices to show that $\at=h'(\tup{z})$.

From \emph{Step~2} we already know that $h(\singvarv{t}) =
v_{h'(\vtup{t})}$.
Using the definition of $\dbBinary$ and the fact that $h(\singvarw{t})=w_{\at}$
and 
$(h(\singvarw{t}),h(\singvarv{t}))\in (E_{\nu,f(\nu)})^{\dbBinary}$, we obtain that
$\proj_{\nu}(\at)=\proj_{f(\nu)}(h'(\vtup{t}))$,
i.e., $a_\nu = h'(x_{f(\nu)})$, 
for all $\nu\in[r]$.

Using that $(z_1,\ldots,z_r) = (x_{f(1)},\ldots,x_{f(r)})$, we obtain:
$h'(\tup{z})=(h'(x_{f(1)}),\ldots,h'(x_{f(r)}))=(a_1,\ldots,a_r)=\at$.
This completes the proof of Step~3.
In summary, the proof of Claim~\ref{claim:PropertiesOfBeta}\ref{item:BetahIsConsistent} is complete.
\end{proof}

\begin{proof}[Proof of
  Claim~\ref{claim:PropertiesOfBeta}\ref{item:BetahIsAHomomorphism},
  \ref{item:betaIsSufficientlyInjective}, and \ref{item:betaIsSurjective}] \
  
\ref{item:BetahIsAHomomorphism}: \  
Let $h\in\Hom(\QBinary,\dbBinary)$ and let $h'\deff\beta(h)$.
Consider an arbitrary atom $R(\tup{z})\in\Atoms(\query)$. We have to
show that $h'(\tup{z})\in R^{D}$.

Let $\myatom\deff R(\tup{z})$ and consider the particular node $t\deff
t_\myatom\in \myAtoms(T)$. 
From Claim~\ref{claim:PropertiesOfBeta}\ref{item:BetahIsConsistent} we obtain that
$h'(\tup{z})\in R^D$.
This completes the proof of Claim~\ref{claim:PropertiesOfBeta}\ref{item:BetahIsAHomomorphism}.
\medskip

\ref{item:betaIsSufficientlyInjective}: \
For each $i\in\smallset{1,2}$ let $h'_i\deff\beta(h_i)$.
  
  First, consider a $t\in V(T)$ such that $h_1(\singvarv{t})\neq h_2(\singvarv{t})$.
  From Claim~\ref{claim:PropertiesOfBeta}\ref{item:BetahIsConsistent} we obtain that
  $h_i(\singvarv{t})=v_{h'_i(\vtup{t})}$, for each $i\in\smallset{1,2}$.
  Thus, we have
  $v_{h'_1(\vtup{t})}=h_1(\singvarv{t})\neq
  h_2(\singvarv{t})=v_{h'_2(\vtup{t})}$.
  This implies that 
  $h'_1(\vtup{t})\neq h'_2(\vtup{t})$.
  I.e., for $(x_1,\ldots,x_m)\deff \vtup{t}$ we have:
  $(h'_1(x_1),\ldots,h'_1(x_m))\neq (h'_2(x_1),\ldots,h'_2(x_m))$.
  Hence, for some $\nu\in[m]$ we have $h'_1(x_\nu)\neq h'_2(x_\nu)$.
  Choosing $y\deff x_\nu$ 
  completes the proof of the first statement.

Now, consider a $t\in \myAtoms(T)$ such that $h_1(\singvarw{t})\neq
h_2(\singvarw{t})$.
Let $R(\tup{z})\deff \Cover(t)$. From Claim~\ref{claim:PropertiesOfBeta}\ref{item:BetahIsConsistent} we obtain that
  $h_i(\singvarw{t})=w_{h'_i(\tup{z})}$, for each $i\in\smallset{1,2}$.
  Thus, we have
  $w_{h'_1(\tup{z})}=h_1(\singvarw{t})\neq
  h_2(\singvarw{t})=w_{h'_2(\tup{z})}$.
  This implies that 
  $h'_1(\tup{z})\neq h'_2(\tup{z})$.
  I.e., for $(z_1,\ldots,z_r)\deff \tup{z}$ we have:
  $(h'_1(z_1),\ldots,h'_1(z_r))\neq (h'_2(z_1),\ldots,h'_2(z_r))$.
  Hence, for some $\nu\in[r]$ we have $h'_1(z_\nu)\neq h'_2(z_\nu)$.
  Choosing $y\deff z_\nu$ and noting that $y\in\Bag(t)$
  (because from $t\in \myAtoms(T)$ we know that $\Bag(t)=\Vars(\Cover(t))=\smallset{z_1,\ldots,z_r}$)
  completes the proof of the second statement and 
 the proof of Claim~\ref{claim:PropertiesOfBeta}\ref{item:betaIsSufficientlyInjective}.
\medskip

\ref{item:betaIsSurjective}: \
 From Claim~\ref{claim:PropertiesOfBeta}\ref{item:betaIsSufficientlyInjective} and the fact that
$\Vars(\QBinary)=\setc{\singvarv{t}}{t\in
  V(T)}\cup\setc{\singvarw{t}}{t\in \myAtoms(T)}$, we immediately obtain
that the mapping $\beta$ is injective.
To prove that it is surjective, we proceed as follows.

Let $h''\in\Hom(\query,D)$. Our aim is to find a
$h\in\Hom(\QBinary,\dbBinary)$ such that $h''=\beta(h)$.
Based on $h''$, we define a mapping $h\colon\Vars(\QBinary)\to\Adom(\dbBinary)$ as follows.
Recall that $\Vars(\QBinary)=\setc{\singvarv{t}}{t\in
  V(T)}\cup\setc{\singvarw{t}}{t\in \myAtoms(T)}$ and $\Adom(\dbBinary)=\setc{w_{\at}}{\at\in\tD}\cup
\setc{v_{\projection}}{\projection\in\Projections(\tD)}$.
For every $t\in V(T)$ we let $R_t(\tup{z}_t)\deff\Cover(t)$.
Since $h''\in\Hom(\query,D)$ and $R_t(\tup{z}_t)\in\Atoms(\query)$, we have $\at_t\deff h''(\tup{z}_t)\in
(R_t)^D\subseteq\tD$. From $\tset(\tup{z}_t)\supseteq\Bag(t)$ we
obtain that $h''(\vtup{t})\in\Projections(\at_t)\subseteq\Projections(\tD)$.
We let
\[
  h(\singvarv{t})\deff v_{h''(\vtup{t})},
  \quad\text{for every $t\in V(T)$,}
  \qquad\text{and}\qquad
  h(\singvarw{t})\deff w_{h''(\tup{z}_t)},
  \quad\text{for every $t\in \myAtoms(T)$.}
\]
Clearly, $h$ is a mapping $h\colon\Vars(\QBinary)\to\Adom(\dbBinary)$.
\medskip

\noindent
\textit{Step~1:} \ $h\in\Hom(\QBinary,\dbBinary)$.
\smallskip

\noindent
\textit{Proof of Step~1:} \
We systematically consider all the atoms in
$\Atoms(\QBinary)$. Consider an arbitrary $t\in
V(T)$. Recall that $R_t(\tup{z}_t)=\Cover(t)$ and $h(\singvarv{t})\deff
v_{h''(\vtup{t})}$; and in case that $t\in \myAtoms(T)$ we also have
$h(\singvarw{t})\deff w_{h''(\tup{z}_t)}$.

We first consider the atom $\ArS{|\Bag(t)|}(\singvarv{t})$.
Since $\ar(\vtup{t})=|\Bag(t)|$, from our definition of $\dbBinary$ we
obtain that $h(\singvarv{t})=v_{h''(\vtup{t})} \in (\ArS{|\Bag(t)|})^{\dbBinary}$. 

Next, in case that $t\in \myAtoms(T)$, we consider the atom
$(U_{R_t})(\singvarw{t})$ of $\QBinary$.
We already know that $ h''(\tup{z}_t)\in
(R_t)^D$.
According to our definition of $\dbBinary$, we 
hence obtain that $h(\singvarw{t}) = w_{h''(\tup{z}_t)} \in (U_{R_t})^{\dbBinary}$.

Furthermore, in case that $t\in \myAtoms(T)$, we consider the atoms $E_{i,j}(\singvarw{t},\singvarv{t})$ for
$i\leq\ar(\tup{z}_t)$ and $j\leq\ar(\vtup{t})$ 
where $\proj_i(\tup{z}_t)=\proj_j(\vtup{t})$, i.e., the $i$-th entry of
$\tup{z}_t$ contains the same variable as the $j$-th entry of
$\vtup{t}$.
Thus, also the $i$-th entry of $h''(\tup{z}_t)$ contains the same
value as the $j$-th entry of $h''(\vtup{t})$.
Recall that we already know
that $\at_t\deff h''(\tup{z}_t)\in \tD$ and
$\projection_t\deff h''(\vtup{t})\in\Projections(\at_t)$.
Thus, by our definition of $\dbBinary$ we have:
$(h(\singvarw{t}),h(\singvarv{t}))=(w_{\at_t},v_{\projection_t})\in (E_{i,j})^{\dbBinary}$.

Finally, for arbitrary $t\in V(T)$, in case that $t$ is not the root of $\rootedTree$, we also
have to consider the parent $p$ of $t$ in $\rootedTree$ and the atoms
$F_{i,j}(\singvarv{t},\singvarv{p})$ for $i\leq\ar(\vtup{t})$ and
$j\leq \ar(\vtup{p})$ where $\proj_i(\vtup{t})=\proj_j(\vtup{p})$,
i.e., the $i$-th entry of $\vtup{t}$ contains the same variable as the
$j$-th entry of $\vtup{p}$. Thus, also the $i$-th entry of
$h''(\vtup{t})$ contains the same value as the $j$-th entry of
$h''(\vtup{p})$.

Recall that we already know that 
$h''(\vtup{t})\in\Projections(\tD)$ and $h''(\vtup{p})\in\Projections(\tD)$.
Furthermore, by our choice of $H$ according to Proposition~\ref{prop:ComputeGHD}
we know that $\Bag(t)\subseteq\Bag(p)$ or $\Bag(t)\supseteq\Bag(p)$,
i.e., $\tset(\vtup{t})\subseteq\tset(\vtup{p})$ or
$\tset(\vtup{t})\supseteq\tset(\vtup{p})$.
This implies that $\tset(h''(\vtup{t}))\subseteq\tset(h''(\vtup{p}))$
or $\tset(h''(\vtup{t}))\supseteq\tset(h''(\vtup{p}))$.
By our definition of $\dbBinary$, we therefore have:
$(h''(\vtup{t}),h''(\vtup{p}))\in (F_{i,j})^{\dbBinary}$.
This completes the proof of Step~1.
\medskip

\noindent
\textit{Step~2:} \ $h''=\beta(h)$.
\smallskip

\noindent
\textit{Proof of Step~2:} \
By definition of $h$, for every $t\in V(T)$ we have
$h(\singvarv{t})=v_{h''(\vtup{t})}$.
On the other hand, since $h\in\Hom(\QBinary,\dbBinary)$, we obtain
from Claim~\ref{claim:PropertiesOfBeta}\ref{item:BetahIsConsistent} for $h'\deff \beta(h)$ and for
every $t\in V(T)$ that $h(\singvarv{t})=v_{h'(\vtup{t})}$.
Hence, for every $t\in V(T)$ we have
$v_{h''(\vtup{t})}=v_{h'(\vtup{t})}$, and hence we also have
$h''(\vtup{t})=h'(\vtup{t})$.
Since for every variable $y\in\Vars(\query)$ there is a $t\in V(T)$
such that $y$ occurs as an entry in the tuple $\vtup{t}$, we have
$h''(y)=h'(y)$, for every $y\in\Vars(\query)$.
This proves that $h''=h'$, i.e., $h''=\beta(h)$, and completes the proof of Step~2 as
well as the proof of
Claim~\ref{claim:PropertiesOfBeta}\ref{item:betaIsSurjective}.
\smallskip

In summary, the proof of Claim~\ref{claim:PropertiesOfBeta} is complete.
\end{proof}

In the remainder of
Appendix~\ref{appendix:ReductionFromArbitraryToBinary}, we explain how
$\beta$ yields an easy-to-compute bijection from 
$\smallsem{\QBinary}(\dbBinary)$ to $\smallsem{\query}(D)$.

\begin{claim}\label{claim:BetaYieldsBijectionBetweenQueryResults}
There is a bijection $f\colon \smallsem{\QBinary}(\dbBinary)\to\smallsem{\query}(D)$.
Furthermore, when given a tuple
$\at\in\smallsem{\QBinary}(\dbBinary)$, the tuple
$f(\at)\in\smallsem{\query}(D)$ can be computed in
time $\bigOh(|\free(\query)| \cdot k))$. 
\end{claim}
\begin{proof}
Let $\ell\deff |\free(\query)|$ and let $\Ans(z_1,\ldots,z_\ell)$ be the
head of $\query$ (i.e., $\free(\query)=\smallset{z_1,\ldots,z_\ell}$).
Recall that the head of $\QBinary$ is
$\Ans(\singvarv{t_1},\ldots,\singvarv{t_{|\Witness|}})$, where
$\set{t_1,\ldots,t_{|\Witness|}}=\Witness$
and $|W|<2 {\cdot} |\free(Q)|$.
By definition of the semantics of CQs, we know that
\begin{eqnarray*}
  \smallsem{\QBinary}(\dbBinary)
& =
&
      \setc{(h(\singvarv{t_1}),\ldots,h(\singvarv{t_{|\Witness|}}))}{h\in\Hom(\QBinary,\dbBinary)}\
      , \quad\text{and}    
\\
  \smallsem{\query}(D)
& =
& \setc{(h'(z_1),\ldots,h'(z_\ell))}{h'\in\Hom(\query,D)}\ .
\end{eqnarray*}
From Claim~\ref{claim:PropertiesOfBeta}\ref{item:betaIsSurjective} we know that $\beta$ is a bijection
from $\Hom(\QBinary,\dbBinary)$ to $\Hom(\query,D)$. Hence,
\begin{eqnarray*}
  \smallsem{\query}(D)
& =
& \setc{(\beta(h)(z_1),\ldots,\beta(h)(z_\ell))}{h\in\Hom(\QBinary,\dbBinary)}\ .
\end{eqnarray*}
Furthermore,
since $\free(\query)=\bigcup_{t\in\Witness}\Bag(t)$, we obtain
from Claim~\ref{claim:PropertiesOfBeta}\ref{item:betaIsSufficientlyInjective} 
that for any two $h_1,h_2\in\Hom(\QBinary,\dbBinary)$ with
\[
  \big( h_1(\singvarv{t_1}),\ldots,h_1(\singvarv{t_{|\Witness|}})\big)
  \quad\neq\quad
  \big(
  h_2(\singvarv{t_1}),\ldots,h_2(\singvarv{t_{|\Witness|}})\big)\ ,
\]  
we have
\[
  \big(\beta(h_1)(z_1),\ldots,\beta(h_1)(z_\ell)\big)
  \quad\neq\quad
  \big(\beta(h_2)(z_1),\ldots,\beta(h_2)(z_\ell)\big)\ .
\]  
This shows that there exists a bijection $f\colon\smallsem{\QBinary}(\dbBinary)\to  \smallsem{\query}(D)$.
\medskip

In the following, we provide an algorithm that, when given a tuple
$\at=(a_1,\ldots,a_{|\Witness|})\in\smallsem{\QBinary}(\dbBinary)$
computes the tuple $f(\at)\in \smallsem{\query}(D)$.

First, we read the tuple $\at$ and build in time $O(\ar(\at))=O(|\free(Q)|)$ a data structure that
provides $O(1)$-access to $a_i$ upon input of an $i\in [\ar(\at)]$.

Afterwards, we proceed as follows for every $y\in \free(\query)$.
Let $i\in\set{1,\ldots,|\Witness|}$ be such that $t_i=t_y$. Recall
that we fixed $t_y$ to be a node in $\Witness$ such that
$y\in\Bag(t_y)$, and we fixed $j_y\in[\ar(t_y)]$ such that
$y=\pi_{j_y}(\vtup{t_y})$.
Consider the $i$-th entry $a_i$ of the tuple $\at$. From
Claim~\ref{claim:PropertiesOfBeta}\ref{item:BetahIsConsistent} we know that $a_i=v_{\projection}$ for
some tuple $\projection\in\Projections(\tD)$ of arity $\ar(\projection)=\ar(\vtup{t_y})$.
We let $b_y\deff \proj_{j_y}(\projection)$.

Finally, we output the tuple $\bt\deff (b_{z_1},\ldots,b_{z_\ell})$.
Note that upon input of $\at$, the tuple $\bt$ is constructed within time
$\bigOh(|\ar(\at)|+|\free(\query)|\cdot k)=O(|\free(\query)|\cdot k)$.

All that remains to be done to complete the proof is to show that if
$\at=\big(h(\singvarv{t_1}),\ldots,h(\singvarv{t_{|\Witness|}})\big)$ for
some $h\in\Hom(\QBinary,\dbBinary)$, then $\bt=\big(\beta(h)(z_1),\ldots,\beta(h)(z_\ell)\big)$.
Note that in order to show the latter, it suffices to show that
$b_y=\beta(h)(y)$ for all $y\in\free(\query)$.
Consider an arbitrary $y\in\free(\query)$. As above, we let
$i\in\set{1,\ldots,|\Witness|}$ be such that $t_i=t_y$. According to
our choice of $b_y$ we have: $b_y = \proj_{j_y}(\projection)$, where
$\projection$ is such that $a_i=v_{\projection}$. From
$a_i=h(\singvarv{t_y})$ we obtain that
$h(\singvarv{t_y})=v_{\projection}$. Thus, according to our definition
of $\beta(h)$ we have: $\beta(h)(y)=\proj_{j_y}(\projection)$, i.e.,
$\beta(h)(y)=b_y$.
This completes the proof of Claim~\ref{claim:BetaYieldsBijectionBetweenQueryResults}.
\end{proof}  

\noindent
Finally, the proof of
Theorem~\ref{thm:ReductionArbSchemaToBinarySchema} is complete:
\begin{itemize}
\item
statement~\ref{item:one:ReductionArbToBinary} is obtained from Claim~\ref{claim:Construct-dbBinary},
\item
statement~\ref{item:two:ReductionArbToBinary} is obtained from
Proposition~\ref{prop:ComputeGHD}, Claim~\ref{claim:ConstructQBinary} and the fact
that $|\free(\QBinary)|=|\Witness|<2\cdot|\free(Q)|$,
\item
statement~\ref{item:three:ReductionArbToBinary} is obtained from Claim~\ref{claim:BetaYieldsBijectionBetweenQueryResults}.
\end{itemize}

\subsection{Details Omitted in Section~\ref{sec:ReductionToOneBinaryRelation}}\label{appendix:ReductionToOneBinaryRelation}

\begin{claim}\label{claim:QSimple_arityAndFCACQ}
	$|\free(\QSimple)| < 3 \cdot |\free(Q)|$ and $\QSimple \in \fcACQSimple$.
\end{claim}
\begin{proof}
According to the definition of the head of $\QSimple$ we have
$|\free(\QSimple)|=|\free(Q)|+ 2{\cdot} \ell$, where $\ell$ is the
number of edges of the subgraph of $G(Q)$ induced by the set $\free(Q)$.
Since this subgraph is a forest on $|\free(Q)|$ nodes, its number of edges
is  $<|\free(Q)|$. Hence, $|\free(\QSimple)|<3\cdot |\free(Q)|$.

Note that the Gaifman graph $G(\QSimple)$ of $\QSimple$ is obtained
from $\vec{G}(Q)$ by subdividing each edge $(x,y)$ of $\vec{G}(Q)$
into three edges $(x,z_{xy})$, $(z_{xy},z_{yx})$ and $(z_{yx},y)$, by
attaching a new leaf $z_{xx}$ to every $x\in S$, and
by then forgetting about the orientation of the edges. Clearly,
$G(\QSimple)$ is a forest, and for every connected component
$\CSimple$ of $G(\QSimple)$ the subgraph of $\CSimple$ induced by the
set $\free(\QSimple)\cap V(\CSimple)$ is connected or empty. 
Thus, by
Proposition~\ref{prop:binaryfcacqs}, the query $\QSimple$ is
free-connex acyclic, i.e., $\QSimple\in\fcACQSimple$.
Note that upon input of a query $Q\in\fcACQ$, the query $\QSimple$ can be
constructed in time $O(\size{Q})$.
This completes the proof of Claim~\ref{claim:QSimple_arityAndFCACQ}.
\end{proof}

\claimReductionBinaryToGraph*
\begin{proof}
\ref{claim:reduction:three}: \ Let $\nu$ be a homomorphism from $Q$
to $D$, and  let $\nuSimple\colon\vars(\QSimple)\to\Dom$ be the
mapping defined as follows:
for all $x\in\vars(Q)$ we let $\nuSimple(x)\deff\nu(x)$;
for all $x\in S$ we let $\nuSimple(z_{xx})\deff w_{aa}$ for $a\deff\nu(x)$;
and for all edges $(x,y)$ of $\vec{G}(Q)$ we let $\nuSimple(z_{xy})\deff w_{ab}$ and
$\nuSimple(z_{yx})\deff w_{ba}$ for $a\deff\nu(x)$ and $b\deff\nu(y)$.
We have to  show that $\nuSimple$ is a homomorphism from $\QSimple$ to $\dbSimple$.

First, consider a unary atom of $\QSimple$ that is of the form $V(x)$, or $X(x)$ with $X(x) \in\atoms(Q)$.
In both cases we have $x\in\vars(Q)$, and hence $\nuSimple(x)=a$ for $a\deff\nu(x)$. In particular, $a\in\adom{D}=V^{\dbSimple}$. Furthermore, if $X(x)$ is an atom of $\QSimple$ and of $Q$, then (since $\nu$ is a homomorphism from $Q$ to $D$) we have $a\in X^D = X^{\dbSimple}$.

Next, consider $x\in S$ and the atoms $W(z_{xx})$,
$E(x,z_{xx})$, $E(z_{xx},z_{xx})$ of $\QSimple$.
Since $x\in S$, there exists an $F\in\sigma_{|2}$ such that $F(x,x)\in\Atoms(Q)$.
Since $\nu\in\Hom(Q,D)$, we know that $(a,a)\in F^D$ for
$a\deff\nu(x)$. By definition of $\nuSimple$ we have
$\nuSimple(z_{xx})=w_{aa}$, and by definition of $\dbSimple$ we have
$w_{aa}\in W^{\dbSimple}$ and $(a,w_{aa})\in \RSimple^{\dbSimple}$ and
$(w_{aa},w_{aa})\in \RSimple^{\dbSimple}$.
Furthermore, for every $F\in \sigma_{|2}$ such that $U_F(z_{xx})$ is
an atom of $\QSimple$, we know that $F(x,x)$ is an atom of $Q$, and
hence $(a,a)\in F^D$. By definition of $\dbSimple$, this implies that
$w_{aa}\in (U_F)^{\dbSimple}$. Hence, all the atoms of $\QSimple$ that
involve $z_{xx}$ are handled correctly by $\nuSimple$.

Now, consider an arbitrary edge $(x,y)$ of $\vec{G}(Q)$ and the atoms
$W(z_{xy})$, $W(z_{yx})$, $\RSimple(x,z_{xy})$,
$\RSimple(z_{xy},z_{yx})$, and $\RSimple(z_{yx},y)$.
We know that $x,y\in \vars(Q)$. Let $a\deff\nu(x)$ and $b\deff\nu(y)$. 
Since $(x,y)$ is an edge of $\vec{G}(Q)$, there exists an $F\in
\sigma_{|2}$ such that $\atoms(Q)$ contains at least one of the atoms
$F(x,y)$ and $F(y,x)$.
Since $\nu$ is a homomorphism from $Q$ to $D$, we have $a,b\in\adom{D}$ and, furthermore, we have $(a,b)\in F^D$ or $(b,a)\in F^D$ (note that we either have $a=b$ or $a\neq b$).
By the definition of $\nuSimple$ we have $\nuSimple(z_{xy})=w_{ab}$ and $\nuSimple(z_{yx})=w_{ba}$.
By the definition of $\dbSimple$, the relation $\RSimple^{\dbSimple}$ contains each of the following tuples:
$(a,w_{ab})$, $(w_{ab},w_{ba})$, $(w_{ba},b)$ (independently of
whether or not $a{=}b$).
Furthermore, the relation $W^{\dbSimple}$ contains $w_{ab}$ and
$w_{ba}$.
Finally, for every $F\in\sigma_{|2}$ such that $U_F(z_{xy})$ ($U_F(z_{yx})$,
resp.) is an atom of $\QSimple$, we know that $F(x,y)$ ($F(y,x)$,
resp.) is an atom of $Q$, and hence $(a,b)\in F^D$ ($(b,a)\in F^D$,
resp.). By definition of $\dbSimple$, this implies that $w_{ab}\in
(U_F)^{\dbSimple}$ ($w_{ba}\in (U_F)^{\dbSimple}$, resp.).
Hence, all the atoms of $\QSimple$ that
involve $z_{xy}$ or $z_{yx}$ are handled correctly by $\nuSimple$.

In summary, we have verified that $\nuSimple$ is a homomorphism from
$\QSimple$ to $\dbSimple$. 
Hence, the proof of part~\ref{claim:reduction:three} of
Claim~\ref{claim:ReductionBinaryToGraph} is complete.
\medskip

\noindent\ref{claim:reduction:QSimpleToQ}: \
Let $\nuSimple$ be a homomorphism from $\QSimple$ to $\dbSimple$.
	\smallskip

	\ref{claim:reduction:one}: \
	Consider an arbitrary edge $(x,y)$ of $\vec{G}(Q)$, and let $a\deff
	\nuSimple(x)$ and $b\deff\nuSimple(y)$ (note that either $a=b$ or
	$a\neq b$).
	Since $x,y\in\vars(Q)$, the query $\QSimple$ contains the atoms
	$V(x)$ and $V(y)$. Thus, according to the definition of $\dbSimple$
	we have $a,b\in\adom{D}$. Since $(x,y)$ is an edge of $\vec{G}(Q)$,
	we know that $x\neq y$ and, by the construction of $\QSimple$, the
	query $\QSimple$ contains the atoms  $W(z_{xy})$, $W(z_{yx})$, $\RSimple(x,z_{xy})$,
	$\RSimple(z_{xy},z_{yx})$, $\RSimple(z_{yx},y)$.
	Since $\nuSimple$ is a homomorphism from $\QSimple$ to $\dbSimple$,
	for $c\deff\nuSimple(z_{xy})$ and $d\deff\nuSimple(z_{yx})$ the
	relation $W^{\dbSimple}$ contains $c$ and $d$, and  
	the relation
	$\RSimple^{\dbSimple}$ contains each of the tuples $(a,c),
	(c,d), (d,b)$ (note that
	either $c=d$ or $c\neq d$). By the
	definition of $\dbSimple$, every element in $W^{\dbSimple}$ has
	exactly one neighbor that belongs to  $W^{\dbSimple}$ and exactly
	one neighbor that belongs to $\adom{D}$, and 
	$(a,c), (c,d), (d,b) \in \RSimple^{\dbSimple}$ necessarily implies
	that $c=w_{ab}$ and $d=w_{ba}$ (cf.\
	Fig.~\ref{fig:BinToGraph_gadgets}).
	Hence, we have:  
	$\nuSimple(z_{xy})=w_{ab}$ and $\nuSimple(z_{yx})=w_{ba}$.
	
	Now, consider an arbitrary
	$x\in S$ and let $a\deff\nuSimple(x)$.
	Since $x\in\vars(Q)$, the query $\QSimple$ contains the atom
	$V(x)$. Thus, according to the definition of $\dbSimple$ we have
	$a\in\adom{D}$. Since $x\in S$, by the construction of $\QSimple$,
	the query $\QSimple$ contains the atoms $W(z_{xx}), E(x,z_{xx}), E(z_{xx},z_{xx})$.
	Since $\nuSimple$ is a homomorphism from $\QSimple$ to $\dbSimple$,
	for $c\deff\nuSimple(z_{xx})$ we have 
	$c\in W^{\dbSimple}$, and
	the relation
	$\RSimple^{\dbSimple}$ contains the tuples $(a,c)$ and $(c,c)$.
	By our construction of $\dbSimple$, only nodes of the form $w_{bb}$
	(for some $b$ with $(b,b)\in \myTuplesSym$) form a self-loop $(w_{bb},w_{bb})$ in
	$\RSimple^{\dbSimple}$; and the particular node $w_{aa}$ is the only
	such node that is a neighbor of $a$ in the sense that $(a,w_{aa})\in
	\RSimple^{\dbSimple}$.
	Thus, $c=w_{aa}$. Hence, we have: $\nuSimple(z_{xx})=w_{aa}$.

	This completes the proof of item~\ref{claim:reduction:one} of part~\ref{claim:reduction:QSimpleToQ} of Claim~\ref{claim:ReductionBinaryToGraph}.
\medskip

\ref{claim:reduction:two}: \ 
First, consider an arbitrary unary atom $X(x)$ of $Q$. By construction, $X(x)$ is an atom of $\QSimple$.
Thus, since $\nuSimple$ is a homomorphism from $\QSimple$ to $\dbSimple$, we have for $a\deff\nuSimple(x)$ that $a$ belongs to $X^{\dbSimple}$. By construction of $\dbSimple$ we have $X^{\dbSimple}=X^D$. Thus, we have: $\nu(x)=a\in X^D$.

Next, consider an arbitrary atom of $Q$ of the form $F(x,x)$. Then,
$x\in S$, and from \ref{claim:reduction:one} we obtain that 
$a\deff\nuSimple(x)\in\Adom(D)$ and
$\nuSimple(z_{xx})=w_{aa}$.
Furthermore, by construction, $\QSimple$ contains the atom $U_F(z_{xx})$.
Thus, since $\nuSimple$ is a homomorphism from $\QSimple$ to
$\dbSimple$, we have: $w_{aa}=\nuSimple(z_{xx})\in (U_F)^{\dbSimple}$.
 By construction of $\dbSimple$, this implies that $(a,a)\in F^D$. Thus, we have:
$(\nu(x),\nu(x))=(a,a)\in F^D$.

Finally, consider an arbitrary atom of $Q$ of the form $F(u,v)$ with $u,v\in \vars(Q)$ and $u\neq v$.
Then, the Gaifman graph $G(Q)$ of $Q$ contains the edge $\smallset{u,v}$, and its oriented version $\vec{G}(Q)$ contains either the edge $(u,v)$ or the edge $(v,u)$.

Let us first consider the case that $\vec{G}(Q)$ contains the edge $(u,v)$. Let $a\deff \nuSimple(u)$ and $b\deff\nuSimple(v)$.
From item~\ref{claim:reduction:one} we obtain that $a,b\in\adom{D}$ and $\nuSimple(z_{uv})=w_{ab}$.
By construction, $\QSimple$ contains the atom $U_F(z_{uv})$.
Thus, since $\nuSimple$ is a homomorphism from $\QSimple$ to $\dbSimple$, we have $w_{ab}=\nuSimple(z_{uv})\in (U_F)^{\dbSimple}$. By construction of $\dbSimple$ we have $(a,b)\in F^D$. Thus, we have:
$(\nu(u),\nu(v))=(a,b)\in F^D$.

Let us now consider the case that $\vec{G}(Q)$ contains the edge $(v,u)$. Let $a\deff\nuSimple(v)$ and $b\deff\nuSimple(u)$.
From item~\ref{claim:reduction:one} we obtain that $a,b\in\adom{D}$
and $\nuSimple(z_{vu})=w_{ab}$
and $\nuSimple(z_{uv})=w_{ba}$.
By construction, $\QSimple$ contains the atom $U_F(z_{uv})$ (because
$Q$ contains the atom $F(u,v)$).
Thus, since $\nuSimple$ is a homomorphism from $\QSimple$ to $\dbSimple$, we have $w_{ba}=\nuSimple(z_{uv})\in (U_F)^{\dbSimple}$. By construction of $\dbSimple$ we have $(b,a)\in F^D$. Thus, we have:
$(\nu(u),\nu(v))\in F^D$.

In summary, we have shown that $\nu$ is a homomorphism from $Q$ to $D$.
	This completes the proof of item~\ref{claim:reduction:two} of
	part~\ref{claim:reduction:QSimpleToQ} of
	Claim~\ref{claim:ReductionBinaryToGraph}.
Finally, the proof of  Claim~\ref{claim:ReductionBinaryToGraph} is complete.
\end{proof}

\section{Details Omitted in Section~\ref{sec:eval}}\label{sec:undirected-evalApp}

\subsection{Proof of Lemma~\ref{lemma:homomorphisms}}\label{app:proof:homomorphisms}
\lemmaHomomorphisms*
\begin{proof}%
\noindent\enquote{$\Longrightarrow$}: \
Let $\val\colon \vars(\QOne) \to \Adom(\DOne)$ be a homomorphism from $\QOne$ to $\DOne$.
We have to show that
\begin{enumerate}[(1)]
	\item\label{lem:homomorphisms:forward:1}
	$\vl(\val(x))\supseteq \lambda_x$, for every $x\in\vars(\QOne)$, \ and
	\item\label{lem:homomorphisms:forward:2}
	$\set{ \val(x),\val(y) }[] \in \EOne$ for every edge $\set{x,y}[]$ of $G(\QOne)$.
\end{enumerate}

\noindent
Consider a variable $x \in \vars(\QOne)$.
By definition, $\lambda_x = \set{ U \in \sigmaOne \mid U(x) \in \atoms(\QOne) }$.
Thus, consider $U \in \lambda_x$; we have to show that $U \in \vl(\val(x))$.
Since $U(x) \in \atoms(\QOne)$ and $\val$ is a homomorphism, it must hold that $(\val(x)) \in U^{\DOne}$.
By definition of $\GOne$, $(\val(x)) \in U^{\DOne}$ implies that $U \in \vl(\val(x))$.
This proves that~\ref{lem:homomorphisms:forward:1} holds.
Consider an edge $\set{ x,y }[]$ of $G(\QOne)$.
Then, $E(x,y) \in \atoms(\QOne)$ or $E(y,x) \in \atoms(\QOne)$ must hold (or both).
Since $\val$ is a homomorphism, this means that at least one of $(\val(x), \val(y)) \in E^{\DOne}$, $(\val(y), \val(x)) \in E^{\DOne}$ must be true.
Hence, by definition of $\GOne$, we have that $\set{ \val(x), \val(y) }[] \in \EOne$.
This proves that~\ref{lem:homomorphisms:forward:2} holds.
\medskip

\noindent\enquote{$\Longleftarrow$}: \
Let $\val\colon \vars(\QOne) \to \Adom(\DOne)$ be an assignment such that the following holds.
\begin{enumerate}[(1)]
	\item\label{lem:homomorphisms:backward:1}
	$\vl(\val(x))\supseteq \lambda_x$, for every $x\in\vars(\QOne)$, \ and
	\item\label{lem:homomorphisms:backward:2}
	$\set{ \val(x),\val(y) }[] \in \EOne$ for every edge $\set{x,y}[]$ of $G(\QOne)$.
\end{enumerate}
We must show that $\val$ is a homomorphism from $\QOne$ to $\DOne$, i.e., that we have $(\val(x_1), \dots, \val(x_r)) \in R^{\DOne}$ for all atoms $R(x_1, \dots, x_r) \in \atoms(\QOne)$.
Recall that by definition, $\sigmaOne$ consists of a single binary symbol $E$, the unary symbol $L$ and possibly further unary symbols.
Therefore, we only have to distinguish two forms of atoms that may appear in $\QOne$ --- unary and binary.
Consider a binary atom $E(x, y)$ in $\atoms(\QOne)$. Then, $x, y \in
\vars(\QOne)$, and
by definition of $\QOne$ we have $x \neq y$.
Then $\set{ x,y }[]$ must be an edge of $G(\QOne)$ and due to~\ref{lem:homomorphisms:backward:2} we know that $\set{ \val(x), \val(y) }[] \in \EOne$. By definition, this is the case if $(\val(x), \val(y)) \in E^{\DOne}$ or if $(\val(y), \val(x)) \in E^{\DOne}$. Since $E^{\DOne}$ is symmetric, this means $(\val(x), \val(y)) \in E^{\DOne}$ either way.

Consider a unary atom $U(x)$ in $\atoms(\QOne)$. Then, $x \in \vars(\QOne)$.
Since $x$ is a node of $G(\QOne)$, the set $\lambda_x$ contains the symbol $U$ by definition, and~\ref{lem:homomorphisms:backward:1} yields that $U \in \vl(\val(x))$ holds as well.
Plugging in the definition of $\vl$ for $\GOne$ yields that $(\val(x)) \in U^{\DOne}$.
This completes the proof of Lemma~\ref{lemma:homomorphisms}.
\end{proof}

\subsection{Fundamental Observations for the Results of Section~\ref{sec:eval}}%
\label{app:eval:observations}
\begin{observation}\label{obs:appendix:main-lemma:unary_predicates}
	For every color $c \in C$ and for every $u \in \Adom(\DOne)$ with $\col(u) = c$ we have: \ \
	$\set{ U \mid (c) \in U^{\ciD} } = \vl(u)$.
\end{observation}
\begin{proof}
	This trivially follows from the following facts:
	By definition, $(c) \in U^{\ciD}$ iff there exists a $v \in \Adom(\DOne)$ with
	$\col(v) = c$ and $(v)\in U^{\DOne}$.
	Furthermore, $(v)\in U^{\DOne}$ iff $U \in \vl(v)$.
	Finally, the coloring $\col$ refines $\vl$, i.e.,
	for all $u, v \in \Adom(\DOne)$ with $\col(u) = \col(v)$ we have $\vl(u) = \vl(v)$.
\end{proof}

\begin{corollary}\label{cor:homomorphisms}
	A mapping $\nu: \vars(\QOne) \to \Adom(\DOne)$ is a homomorphism from $\QOne$ to $\DOne$ if, and only if:
	\begin{enumerate}[(1)]
		\item for all $x \in \Vars(\QOne)$ we have:\; $\lambda_x \subseteq \vl(\nu(x))$,\; and
		\item for all $x \in \Vars(\QOne)$ with $x \neq x_1$ we have:\ \,
		$\nu(x) \in \N{\nu(y)}{d}$, where $y = \Parent(x)$ and $d = \col(\nu(x))$.
	\end{enumerate}
\end{corollary}
\begin{proof}
	Condition (1) is the same as in Lemma~\ref{lemma:homomorphisms}.
	It remains to show that condition~(2) is equivalent to the
        following condition~(2) of
	Lemma~\ref{lemma:homomorphisms}: $\set{ \val(x),\, \val(y) }[] \in \EOne$ for every edge $\set{ x, y }[]$ of $G(\QOne)$.

	Note that $\set{ x,y }[]$ is an edge in $G(\QOne)$ iff either
        $x = \Parent(y)$ or $y = \Parent(x)$. Thus,
        condition~(2) of Lemma~\ref{lemma:homomorphisms} holds iff $\set{ \val(x),\val(y) }[] \in \EOne$ for every $x \in \vars(\QOne)$ with $x \neq x_1$ and $y = \Parent(x)$.

	By definition, $\set{ \nu(x), \nu(y) }[] \in \EOne$ holds iff
        $\nu(x) \in \N{\nu(y)}{\col(\nu(x))}$. Thus, condition~(2) of
        Lemma~\ref{lemma:homomorphisms} holds iff $\nu(x) \in
        \N{\nu(y)}{\col(\nu(x))}$ for every $x \in \vars(\QOne)$ with
        $x \neq x_1$ and $y = \Parent(x)$, which is equivalent to
        item~(2) of Corollary~\ref{cor:homomorphisms}. This completes
        the proof of Corollary~\ref{cor:homomorphisms}.
\end{proof}

\begin{lemma}\label{claim:appendix:main-lemma:neighbors}
	Let $\mu\colon \Vars(\QOne) \to C$ be a homomorphism from $\QOne$ to $\ciD$.
	For every $x \in \Vars(\QOne)$ that is not a leaf of $T$,
	for every $v \in \VOne$ with $\col(v) = \mu(x)$, and for every $z \in \Children(x)$ we have: \
	$\N{v}{d} \neq \emptyset$ for $d \isdef \mu(z)$.
\end{lemma}
\begin{proof}
  Let $x \in \Vars(\QOne)$ with $\Children(x) \neq \emptyset$.
  Let $z \in \Children(x)$. Let $c\isdef \mu(x)$ and $d \isdef
  \mu(z)$.
  Let $v \in \VOne$ such that $\col(v) = \mu(x) = c$.

  From $z\in\Children(x)$ we obtain that  $\Atoms(\QOne)$ contains at
  least one of the atoms $E(x, z)$ and $E(z,x)$.
  Since $\mu$ is a homomorphism from $\QOne$ to $\ciD$, this implies that
  $(\mu(x), \mu(z)) \in E^{\ciD}$ or $(\mu(z), \mu(x)) \in E^{\ciD}$. I.e., $(c, d) \in E^{\ciD}$ or $(d, c) \in E^{\ciD}$ holds.
  Since $E^{\ciD}$ is symmetric, this means that $\numN{c}{d} > 0$ is
  true in any case. This implies that for every $u \in \VOne$ with
  $\col(u) = c$ we have $\N{u}{d} \neq \emptyset$. This, in
  particular, yields that $\N{v}{d} \neq \emptyset$.
  This completes the proof of Lemma~\ref{claim:appendix:main-lemma:neighbors}.
\end{proof}

\begin{lemma}\label{prop:appendix:main-lemma:hom_correspondence}
	Let $\mu\colon \Vars(\QOne) \to C$ and $\nu\colon \Vars(\QOne) \to \VOne$ be mappings such that for all $x \in \Vars(\QOne)$ we have
	\begin{enumerate}[(a)]
		\item\label{prop:appendix:main-lemma:hom_correspondence:a}
		$\col(\nu(x)) = \mu(x)$, and
		\item\label{prop:appendix:main-lemma:hom_correspondence:b}
		$x = x_1$ or $\nu(x) \in \N{\nu(y)}{d}$
		where $y = \Parent(x)$ and $d = \mu(x)$.
	\end{enumerate}
	Then, $\mu$ is a homomorphism from $\QOne$ to $\ciD$ iff $\nu$ is a homomorphism from $\QOne$ to $\DOne$.
\end{lemma}

\begin{proof}
	\noindent\enquote{$\Longrightarrow$}: \
	We use Corollary~\ref{cor:homomorphisms} to verify that $\nu$ is a homomorphism from $\QOne$ to $\DOne$:

	\noindent (1):\; Let $x \in \Vars(\QOne)$ and let $\nu(x) = u$ and $\mu(x) = c$.
	We must show that $\vl(u) \supseteq \lambda_x$.
	Since $\mu$ is a homomorphism from $\QOne$ to $\ciD$, we have
        $\lambda_x \subseteq A$ for $A\isdef \set{ U \mid (c) \in U^{\ciD} }$.
	Since $\col(u) = c$
        by~\ref{prop:appendix:main-lemma:hom_correspondence:a}, we
        obtain from
        Observation~\ref{obs:appendix:main-lemma:unary_predicates} that
	$A = \vl(u)$.
	Thus, $\lambda_x \subseteq \vl(u)$.

	\noindent (2):\;
	This follows directly
        from~\ref{prop:appendix:main-lemma:hom_correspondence:b} and \ref{prop:appendix:main-lemma:hom_correspondence:a}. 
 	\medskip
        
	\noindent\enquote{$\Longleftarrow$}: \
	First, consider an arbitrary unary atom $U(x)$ in $\atoms(\QOne)$. We must show that $(\mu(x)) \in U^{\ciD}$.
	Since $\nu$ is a homomorphism from $\QOne$ to $\DOne$, we have
        $(\nu(x)) \in U^{\DOne}$, i.e., $U \in \vl(\nu(x))$. 
	Using (a) and
        Observation~\ref{obs:appendix:main-lemma:unary_predicates}
        yields that $(\mu(x)) \in U^{\ciD}$. 

	Now, consider an arbitrary binary atom $E(x,y)$ in
        $\atoms(\QOne)$. We must show that $(\mu(x), \mu(y)) \in
        E^{\ciD}$. Since $\nu$ is a homomorphism from $\QOne$ to $\DOne$, we have $(\nu(x), \nu(y)) \in E^{\DOne}$.
	Let $c\isdef \col(\nu(x))$ and $d\isdef \col(\nu(y)) $.
	Then, by definition of $\ciD$ it holds that $(c, d) \in E^{\ciD}$.
	Because of (a), we know that $\mu(x) = c$ and $\mu(y) = d$,
        which shows that $(\mu(x), \mu(y)) \in E^{\ciD}$.

        This completes the proof of Lemma~\ref{prop:appendix:main-lemma:hom_correspondence}.
\end{proof}

\subsection{Proof of Lemma~\ref{lemma:bool}}\label{app:proof:bool}
\lemmaBool*
\begin{proof}
  By the definition of $\DOne$ and $\QOne$ we already know that
  $\sem{Q}(D) = \sem{\QOne}(\DOne)$. In the following, we show that
  $\sem{\QOne}(\DOne) = \sem{\QOne}(\ciD)$.

  First, consider the case that $\sem{\QOne}(\DOne) =\Yes$. Then, 
  there is a homomorphism $\nu\colon \vars(\QOne) \to \VOne$ from
  $\QOne$ to $\DOne$.  Let $\mu\colon \vars(\QOne) \to C$ defined via $\mu(x) \isdef \col(\nu(x))$ for all $x \in \vars(\QOne)$.
Then, according to Corollary~\ref{cor:homomorphisms}, the mappings
$\mu$ and $\nu$ match the requirements (a) and (b) of
Lemma~\ref{prop:appendix:main-lemma:hom_correspondence}. Hence,
Lemma~\ref{prop:appendix:main-lemma:hom_correspondence} yields that
$\mu$ is a homomorphism from $\QOne$ to $\ciD$. I.e., we have  $\sem{\QOne}(\ciD) =\Yes$.

Now, consider the case that $\sem{\QOne}(\ciD) =\Yes$.
Then, there is a homomorphism $\mu\colon \vars(\QOne) \to C$ from
$\QOne$ to $\ciD$.
We can now combine
Lemma~\ref{prop:appendix:main-lemma:hom_correspondence} with
Lemma~\ref{claim:appendix:main-lemma:neighbors} to obtain a
homomorphism $\nu\colon \Vars(\QOne) \to \VOne$ from $\QOne$ to
$\DOne$ as follows.
Do a top-down pass on $T$. For the root node $r$ pick an arbitrary node $v$
  in $\VOne$ of color $\col(v)=\mu(r)$ and let $\nu(r)\deff v$. In the
  induction step, we consider an edge $(y,x)$ of $T$; by the induction
hypothesis we have already chosen a node $v=\nu(y)\in \VOne$ of color
$\mu(y)$. We let $d\deff\mu(x)$. From Lemma~\ref{claim:appendix:main-lemma:neighbors}
we know that there exists a node $v'\in N(v,d)$. We pick an arbitrary
such node $v'$ and let $\nu(x)\deff v'$. It can now be verified that
this mapping $\nu$ satisfies the conditions (a) and (b) of
Lemma~\ref{prop:appendix:main-lemma:hom_correspondence}. Thus, by
Lemma~\ref{prop:appendix:main-lemma:hom_correspondence} we
obtain that $\nu$ is a homomorphism from $\QOne$ to $\DOne$.
Hence, $\sem{\QOne}(\DOne) =\Yes$.
This completes the proof of Lemma~\ref{lemma:bool}.
\end{proof}

\subsection{Proof of Lemma~\ref{lemma:mainLemma}}\label{app:proof:mainLemma}
For the next two lemmas, consider a fixed $\bar{c} \isdef (c_1, \dots,
c_k) \in \sem{\QOne}(\ciD)$, and a fixed homomorphism $\mu\colon
\Vars(\QOne) \to C$ from $\QOne$ to $\ciD$ witnessing that $\bar{c} \in \sem{\QOne}(\ciD)$, i.e., $\mu(x_i) = c_i$ for all $i \in [k]$.

For $\ell \in [k]$, we call an $\ell$-tuple $(v_1, \dots, v_{\ell})\in\Adom(\DOne)^\ell$ \emph{consistent with} $\bar{c} = (c_1, \dots, c_k)$, if 
\begin{enumerate}
	\item 
	for all $i \in [\ell]$ $\col(v_i) = c_i$,\; and
	\item
	for all $j < i$ where $\set{ x_i , x_j }[]$ is an edge in $G(\QOne)$ we have $v_i \in \N{v_j}{c_i}$.
\end{enumerate}
Note that $(v_1)$ is consistent with $\bar{c}$ for every $v_1 \in \Adom(\DOne)$ with $\col(v_1) = c_1$.

\begin{lemma}\label{app:lemma:neighbors}
	Let $\ell \in [k{-}1]$, let $(v_1, \dots, v_{\ell})$ be consistent with $\bar{c}$ and let $x_j = \Parent(x_{\ell+1})$.
	Then, $(v_1, \dots, v_{\ell}, v_{\ell+1})$ is consistent with $\bar{c}$ for all $v_{\ell+1} \in \N{v_j}{c_{\ell+1}}$.
\end{lemma}

\begin{proof}
	Let $v_{\ell+1} \in \N{v_j}{c_{\ell+1}}$.
	We have to show that $(v_1, \dots, v_{\ell}, v_{\ell+1})$ is
        consistent with $\bar{c}$.

	Let $i \in [\ell+1]$.
	Clearly, $\col(v_i) = c_{i}$ holds.
	Let $j < i$ such that $\set{ x_i, x_j }[]$ is an edge in $G(\QOne)$.
	If $i \leq \ell$, then $v_i \in \N{v_j}{c_i}$ since, by assumption, $(v_1, \dots, v_{\ell})$ is consistent with $\bar{c}$.
	If $i = \ell+1$, then $x_j$ is the unique parent of $x_i$ since $G(\QOne)$ is acyclic, i.e., $x_j = \Parent(x_i)$.
	By choice of $v_i = v_{\ell+1}$ we have: $v_{i} \in \N{v_j}{c_i}$.
	Thus, in summary, $(v_1, \dots, v_{\ell}, v_{\ell+1})$ is consistent with $\bar{c}$.
\end{proof}

\begin{lemma}\label{claim:appendix:main-lemma:consistency}
	For every $\ell \in [k]$ and every $\ell$-tuple $\bar{v} \in \VOne^\ell$, the following is true:
	\begin{quote}
		$\bar{v}$ is consistent with $\bar{c}$ \;$\iff$\; $\bar{v}$ is a partial output of $\QOne$ over $\DOne$	of color $\bar{c}$.
	\end{quote}
\end{lemma}
\begin{proof}
	\enquote{$\Longrightarrow$}:\;
	Let $\bar{v}=(v_1, \dots, v_{\ell})\in\VOne^\ell$ be consistent with $\bar{c}$.
	We show the stronger statement that there exists a homomorphism $\nu\colon \Vars(\QOne) \to \VOne$ from $\QOne$ to $\DOne$ such that $\nu(x_i) = v_i$ for all $i \in [\ell]$ and $\col(\nu(z)) = \mu(z)$ for all $z \in \Vars(\QOne)$.
	From this, it directly follows that $(v_1, \dots, v_{\ell})$ is a partial output of $\QOne$ over $\DOne$ of color $\bar{c}$.
	
	We define $\nu$ along the order $<$ on $\Vars(\QOne)$ (which
        corresponds to a top-down pass over the edges of $T$).
	For all $i \in [\ell]$, let $\nu(x_i) \isdef v_i$.
	Now consider $z \in \Vars(\QOne)$ with $z \neq x_i$ for every
        $i \in [\ell]$ but where we have already considered $y =
        \Parent(z)$, i.e., where we have already picked $\nu(y)$ with $\col(\nu(y)) = \mu(y)$.
	Let $u \isdef \nu(y)$ and $d \isdef \mu(z)$.
	Because $\mu$ is a homomorphism from $\QOne$ to $\ciD$, we
        obtain from Lemma~\ref{claim:appendix:main-lemma:neighbors} that $\N{u}{d}$ is not empty.
	We choose an arbitrary $w \in \N{u}{d}$ and let $\nu(z) \isdef w$.
	Clearly, $\col(\nu(w)) = d = \mu(z)$.

	Once we have considered all variables this way, $\nu$
        satisfies the conditions (a) and (b) of
        Lemma~\ref{prop:appendix:main-lemma:hom_correspondence}. Hence,
        from Lemma~\ref{prop:appendix:main-lemma:hom_correspondence}
        we obtain that $\nu$ is a homomorphism from $\QOne$ to $\DOne$.
	\medskip
        
	\enquote{$\Longleftarrow$}:\;
	Let $(v_1, \dots, v_{\ell})$ be a partial output of $\QOne$ over $\DOne$ of color $\bar{c}$.
	We have to show that it is consistent with $\bar{c}$.
	Clearly, $\col(v_i) = c_i$ for all $i \in [\ell]$.
	Let $j < i$ such that $\set{ x_j, x_i }[]$ is an edge in $G(\QOne)$.
	Since $(v_1, \dots, v_{\ell})$ is a partial output of $\QOne$
        over $\DOne$ of color $\bar{c}$, there exists a homomorphism $\nu$ from $\QOne$ to $\DOne$ such that $\nu(x_i) = v_i$ and $\nu(x_j) = v_j$.
	Notice that $x_j = \Parent(x_i)$ and $\col(v_i) = c_i$.
	Since $\nu$ is a homomorphism,
        Corollary~\ref{cor:homomorphisms}(2) yields that $v_i \in
        \N{v_j}{c_i}$ holds.
        In summary, this shows that $\bar{v}$ is consistent with
        $\bar{c}$. This completes the proof of Lemma~\ref{claim:appendix:main-lemma:consistency}.
\end{proof}

\lemmaMainLemma*
\begin{proof}~
	\begin{enumerate}[(a)]
		\item
		Since $(v_1, \dots, v_k) \in \sem{\QOne}(\DOne)$, there exists a homomorphism $\nu\colon \Vars(\QOne) \to \VOne$ from $\QOne$ to $\DOne$ such that $\nu(x_i) = v_i$ for all $i \in [k]$.
		Let $\mu\colon \Vars(\QOne) \to C$ with $\mu(x) = \col(\nu(x))$ for all $x \in \Vars(\QOne)$.
		Using Corollary~\ref{cor:homomorphisms}, we can apply Lemma~\ref{prop:appendix:main-lemma:hom_correspondence} to $\mu$ and $\nu$.
		Thus, $\mu$ is a homomorphism witnessing that $(\col(v_1), \dots, \col(v_k)) \in \sem{\QOne}(\ciD)$.
		\item
		For every $v_1 \in \Adom(\DOne)$ with $\col(v_1) = c_1$, the tuple $(v_1)$ is consistent with $\bar{c}$ and hence, by Lemma~\ref{claim:appendix:main-lemma:consistency}, it is a partial output of $\QOne$ over $\DOne$ of color $\bar{c}$.

		If $(v_1,\ldots,v_i)$ is a partial output of $\QOne$ over $\DOne$ of color $\bar{c}$ then, by Lemma~\ref{claim:appendix:main-lemma:consistency}, the tuple $(v_1,\ldots,v_i)$ is consistent with $\bar{c}$.
		Let $x_j = \Parent(x_{i{+}1})$.
		Since $(v_1,\ldots,v_i)$ is a partial output of $\QOne$ over $\DOne$ of color $\bar{c}$, we obtain from Lemma~\ref{claim:appendix:main-lemma:neighbors} that $\N{v_j}{c_{i{+}1}} \neq \emptyset$.
		From Lemma~\ref{app:lemma:neighbors} we obtain for every $v_{i{+}1} \in \N{v_j}{c_{i{+}1}}$ that the tuple $(v_1, \dots, v_i, v_{i{+}1})$ is consistent with $\bar{c}$; and applying Lemma~\ref{claim:appendix:main-lemma:consistency}, we obtain that $(v_1, \dots, v_i, v_{i{+}1})$ also is a partial output of $\QOne$ over $\DOne$ of color $\bar{c}$.
		\qedhere
	\end{enumerate}
\end{proof}

\subsection{Proof of Lemma~\ref{lemma:counting}}\label{app:proof:counting}
\lemmaCounting*
\begin{proof}
The proof proceeds by induction, bottom-up over the tree $T$.
\medskip

\noindent\emph{Base case}:\;
Let $v \in \VOne$, let $c=\col(v)$, and let $x$ be a leaf of $T$.
Then \ref{item:b:lemma:counting} trivially holds because $\Children(x) = \emptyset$.
To prove \ref{item:a:lemma:counting}, notice that $V(T_x) = \set{ x }[]$ and $E(T_x) = \emptyset$.
Thus, there is only one mapping $\nu\colon V(T_x) \to \VOne$ with
$\nu(x) = v$.
This mapping $\nu$ satisfies the condition formulated in (a) if and only if
$\vl(\nu(x)) \supseteq \lambda_x$, i.e., if and only if
$\vl(v)\supseteq\lambda_x$.
Since $\col(v)=c=\col(v_c)$, we obtain
from Observation~\ref{obs:appendix:main-lemma:unary_predicates} that
$\vl(v)=\vl(v_c)$.
Thus, the mapping $\nu$ satisfies the condition formulated in \ref{item:a:lemma:counting} iff
$\vl(v_c)\supseteq\lambda_x$ which, by definition, is the case iff
$f_1(c,x) = 1$ (and otherwise, $f_1(c,x)=0$).
Since $x$ is a leaf of $T$, we have $\fd(c, x) = f_1(c, x)$.
In summary, this proves that $\fd(c,x)$ is exactly the number of
mappings $\nu$ that satisfy the condition formulated in \ref{item:a:lemma:counting}.

\bigskip
\noindent\emph{Inductive step}:\;
Let $x \in V(T)$ be an inner node of $T$, let $c \in C$ and let $v \in \VOne$ with $\col(v) = c$.

\medskip
\noindent\emph{Induction hypothesis}:\;
For every $y \in \Children(x)$, every $d \in C$ and every $w \in \VOne$ with $\col(w) = d$,
the value $\fd(d, y)$ is the number of mappings $\nu\colon V(T_y) \to \VOne$ satisfying $\nu(y) = w$ and
\begin{enumerate}[(1)]
	\item
	for every $x' \in V(T_y)$ we have $\vl(\val(x')) \supseteq \lambda_{x'}$, and
	\item
	for every edge $\set{x',y'}[]$ in $T_y$ we have $\set{ \val(x'),\val(y') }[] \in \EOne$.
\end{enumerate}
\medskip

\noindent
We first show the following Claim~\ref{app:claim:lemma:counting:item:b}, which corresponds to the lemma's statement \ref{item:b:lemma:counting}.

\begin{claim}\label{app:claim:lemma:counting:item:b} 
	For every $y \in \Children(x)$, the value $g(c, y)$ is the number of mappings $\nu\colon \set{ x }[] \cup V(T_y) \to \VOne$ satisfying $\nu(x) = v$ and
	\begin{enumerate}[(1)]
		\item
		for every $x' \in V(T_y)$ we have $\vl(\val(x')) \supseteq \lambda_{x'}$, and
		\item
		for every edge $\set{x',y'}[]$ in $T_y$ we have $\set{ \val(x'),\val(y') }[] \in \EOne$, and
		\item
		we have $\set{ \val(x),\val(y) }[] \in \EOne$.\endClaim
	\end{enumerate}
\end{claim}

\begin{claimproof}
	Since every considered mapping $\nu$ has to map $x$ to $v$,
        every child $y \in \Children(x)$ has to be mapped to a $w$
        such that $\set{ v, w }[] \in \EOne$ according to (3), i.e.,
        we know that $y$ has to be mapped to a neighbor $w$ of $v$ in $\DOne$.
	For this, any $w \in \N{v}{c'}$ for any $c' \in C$ is a valid choice, i.e.,
	we have $|\N{v}{c'}| = \numN{c}{c'}$ choices for $w$ for all $c' \in C$.
	Once we have chosen a $w$ in this way and let $\nu(y) = w$, we can use the induction hypothesis to get the number of choices available to map
	the remaining variables such that $\nu(y) = w$ and (1) and (2) hold, which is $\fd(\col(w), y)$.
	Thus, we get that the total number of considered mappings $\nu$ is $\sum_{c' \in C}
        \fd(c', y) \cdot \numN{c}{c'}$, which is $g(c, y)$ by
        definition. This completes the proof of Claim~\ref{app:claim:lemma:counting:item:b}.
\end{claimproof}
\medskip

\noindent
Next, we use Claim~\ref{app:claim:lemma:counting:item:b} to prove the following Claim~\ref{app:claim:lemma:counting:item:a}, which corresponds to the lemma's statement~\ref{item:a:lemma:counting}.

\begin{claim}\label{app:claim:lemma:counting:item:a}
	$f_{\downarrow}(c,x)$ is the number of mappings $\val\colon V(T_x) \to \VOne$ satisfying $\val(x) = v$ and
	\begin{enumerate}[(1)]
		\item
		for every $x' \in V(T_x)$ we have $\vl(\val(x')) \supseteq \lambda_{x'}$, and
		\item
		for every edge $\set{x',y'}[]$ in $T_x$ we have $\set{ \val(x'),\val(y') }[] \in \EOne$.\endClaim
	\end{enumerate}
\end{claim}
\begin{claimproof}
	In the same way as in the base case, we can see that the
        number of considered mappings $\nu$ is 0 unless $f_1(c, x)$ is 1.
	If $f_1(c, x)$ is 1, then, according to~(2), a valid mapping must put every $y \in \Children(x)$ on a neighbor $w$ of $v$.

	Thus, every considered mapping must fulfill the requirements~(1)--(3) of Claim~\ref{app:claim:lemma:counting:item:b} for every $y \in \Children(x)$ if we restrict it to $\set{ x }[] \cup V(T_{y})$.
	On the other hand, the trees $T_y$, for $y\in\Children(x)$,
        only intersect in node $x$.
	Thus, we can combine every map $\nu'\colon \set{ x }[] \cup
        V(T_{y'}) \to \VOne$ of a child $y'$ that adheres to these
        requirements with every other map $\nu''\colon \set{ x }[]
        \cup V(T_{y''})$ of every other child $y''$ that adheres to
        these requirements, and in total we will get a map $\nu$ that
        satisfies the conditions (1) and (2) of Claim~\ref{app:claim:lemma:counting:item:a}.
	Using Claim~\ref{app:claim:lemma:counting:item:b}, this implies that we get a total of $\prod_{y \in \Children(x)} g(c,y)$ choices for $\nu$ if $f_1(c,x)$ is 1,
	and 0 choices if $f_1(c,x)$ is 0, which is precisely $\fd(c, x)$.
	This completes the proof of Claim~\ref{app:claim:lemma:counting:item:a}.
\end{claimproof}
\medskip

\noindent
In summary, this completes the proof of the inductive step.
Hence, the proof of Lemma~\ref{lemma:counting} is complete.
\end{proof}  
\section{Details Omitted in Section~\ref{sec:SizeOfIndex}}\label{app:size-of-index}

In this section we will use the notation from~\cite{ScheidtSchweikardt_MFCS25} and assume that the reader is familiar with it.
For all $\projection \in \Projections(\tD)$ let 
\begin{align*}
	N_{\projection} &\isdef \set[\big]{ \projectionAlt \in \Projections(\tD) \mid 
	\tset(\projectionAlt) \subseteq \tset(\projection)
         \text{ or }
	\tset(\projectionAlt) \supseteq \tset(\projection)
        }
        \quad\text{and}\\
	M_{\projection} &\isdef \set[\big]{ \tup{c} \in \tD \mid \projection \in \Projections(\tup{c}) }.
\end{align*}

To prove Theorem~\ref{thm:size-of-index} we show that the following coloring $g$ of $\dbBinary$ based on CR's coloring of $\mathcal{H}_{D}$ is stable.
Let $h$ be the stable coloring produced by CR on $\mathcal{H}_D$ and let $n_h$ be the number of colors $h$ uses.
From \cite[Theorem~3.6 and Section~3.2]{ScheidtSchweikardt_MFCS25} we know that $h$ restricted to $\set{ w_{\at} \mid \at \in \tD }$ is a stable coloring on $\mathcal{G}_D$.
We define $g$ as follows:
\begin{enumerate}
	\item For all $\at \in \tD$ let $g(w_{\at}) \isdef h(w_{\at})$.
	\item For all $\slice \in \slices(\tD)$ let $g(v_{\slice}) \isdef h(v_{\slice})$.
	\item For all $\projection \in \Projections(\tD) \setminus \slices(\tD)$ let
	\begin{align*}\SwapAboveDisplaySkip
		g(v_{\projection}) \;\;\isdef\;\; \Bigl(
			\stp(\projection),
			\mset[\big]{ 
				\bigl(
					\stp(\tup{c}, \projection),
					h(w_{\tup{c}})
				\bigr)
				\mid
				\tup{c} \in M_{\projection}
			 }
		\Bigr)\;.
	\end{align*}
\end{enumerate}
Let $n_g$ be the number of colors that $g$ uses.
The main goal of this section will be to show that $g$ is a stable
coloring of $\dbBinary$, i.e. to show the following:
\begin{lemma}\label{app:lemma:g-is-stable}
	~
	\begin{enumerate}
		\item 
		For all $\at, \bt \in \tD$ with $g(w_{\at}) = g(w_{\bt})$ it holds that
		\[
			\mset[\big]{ (\stp(\at, \projection), g(v_{\projection})) \mid \projection \in \Projections(\at) }
			=
			\mset[\big]{ (\stp(\bt, \projection), g(v_{\projection})) \mid \projection \in \Projections(\bt) }.
		\]

		\item
		For all $\slice, \sliceAlt \in \slices(\tD)$ with $g(v_{\slice}) = g(v_{\sliceAlt})$ 
		it holds that
		\begin{enumerate}
			\item
			$
				\mset[\big]{ (\stp(\tup{c}, \slice), g(w_{\tup{c}})) \mid \tup{c} \in M_{\slice} }
				=
				\mset[\big]{ (\stp(\tup{c}, \sliceAlt), g(w_{\tup{c}})) \mid \tup{c} \in M_{\sliceAlt} },
			$
			\quad and
			\item
			$
				\mset[\big]{ (\stp(\slice, \projection'), g(v_{\projection'})) \mid
					\projection' \in N_{\slice} }
				=
				\mset[\big]{ (\stp(\sliceAlt, \projection'), g(v_{\projection'})) \mid
					\projection' \in N_{\sliceAlt} }.
			$
		\end{enumerate}

		\item
		For all $\projection, \projectionAlt \in \Projections(\tD) \setminus \slices(\tD)$ with $g(v_{\projection}) = g(v_{\projectionAlt})$ 
		it holds that
		\begin{enumerate}
			\item 
			$
				\mset[\big]{ (\stp(\tup{c}, \projection), g(w_{\tup{c}})) \mid \tup{c} \in M_{\projection} }
				=
				\mset[\big]{ (\stp(\tup{c}, \projectionAlt), g(w_{\tup{c}})) \mid \tup{c} \in M_{\projectionAlt} },
			$ \quad and
			\item
			$
				\mset[\big]{ (\stp(\projection, \projection'), g(v_{\projection'})) \mid
					\projection' \in N_{\projection} }
				=
				\mset[\big]{ (\stp(\projectionAlt, \projection'), g(v_{\projection'})) \mid
					\projection' \in N_{\projectionAlt} }.
			$
		\end{enumerate}
	\end{enumerate}
\end{lemma}

Since $h$ and $g$ restricted to $\set{ w_{\at} \mid \at \in \tD }$ are
equivalent to the coloring that RCR produces on $D$, it is easy to
see that $n_h \in \bigOh(|\rcrcols|)$, and since $g$ is stable on $\dbBinary$, $n_g
\in \bigOh(|\rcrcols|)$ ---
to see that this is true, note that for any two tuples $\at$, $\bt$
with $g(\wat) = g(\wbt)$ it must hold that $\Projections(\at)$ and
$\Projections(\bt)$ are colored in the same way. Since
$|\Projections(\at)| \leq 2^{\bigOh(k \log k)}$ for all $\at \in \tD$,
this yields $n_g \leq |\rcrcols| + |\rcrcols| \cdot 2^{\bigOh(k \log
  k)}$. Since $k=\ar(\sigma)$ and $\sigma$ is fixed, the factor
$2^{\bigOh(k \log k)}$ is assumed to be constant and we obtain that
$n_g = \bigOh(|\rcrcols|)$.
Finally, by using standard methods (check
e.g.~\cite{BBG-ColorRefinement}), $g$ can be transformed into a stable
coloring of $\dbSimple$ with $O(n_g)$ colors. Thus, since running CR on $\dbSimple$
produces a coarsest stable coloring of $\dbSimple$, this proves
Theorem~\ref{thm:size-of-index}. 
Therefore, all that remains to be done in order to prove
Theorem~\ref{thm:size-of-index} is to prove Lemma~\ref{app:lemma:g-is-stable}.
The remainder of this section is devoted to the proof of Lemma~\ref{app:lemma:g-is-stable}.

\medskip\noindent
We know the following since $h$ is the coloring produced by CR on $\mathcal{H}_D$:
\begin{fact}\label{app:fact:h-encodes-stp}
	For all $\at, \bt \in \tD$ with $h(w_{\at}) = h(w_{\bt})$ it holds that $\stp(\at) = \stp(\bt)$; and for all $\slice, \sliceAlt \in \slices(\tD)$ with $h(v_{\slice}) = h(v_{\sliceAlt})$ it holds that $\stp(\slice) = \stp(\sliceAlt)$.
\end{fact}

\smallskip\noindent
We know the following since $h$ is stable on $\mathcal{G}_D$ and $\mathcal{H}_D$:
\begin{fact}\label{app:fact:beta-ab}
	For all $\at, \bt \in \tD$ with $h(w_{\at}) = h(w_{\bt})$ the following is true:
	\begin{enumerate}
		\item\label{item:GD:app:fact:beta-ab} 
		there exists a bijection $\beta_{\at, \bt}\colon \set{ \tup{c} \in \tD \mid \stp(\at, \tup{c}) \neq \emptyset } \to \set{ \tup{d} \in \tD \mid \stp(\bt, \tup{d}) \neq \emptyset }$ such that for all $\tup{c} \in \tD$ with $\stp(\at, \tup{c}) \neq \emptyset$ and $\tup{d} \isdef \beta_{\at, \bt}(\tup{c})$ it holds that
		\begin{enumerate}
			\item $h(w_{\tup{c}}) = h(w_{\tup{d}})$ and
			\item $\stp(\at, \tup{c}) = \stp(\bt, \tup{d})$;
		\end{enumerate}

		\item\label{item:HD:app:fact:beta-ab}
		and there exists a bijection $\tilde{\beta}_{\at, \bt}\colon \slices(\at) \to \slices(\bt)$ such that for all $\slice \in \slices(\at)$ with $\sliceAlt \isdef \tilde{\beta}_{\at, \bt}(\slice)$ it holds that
		\begin{enumerate}
			\item $h(v_{\slice}) = h(v_{\sliceAlt})$ and
			\item $\stp(\at, \slice) = \stp(\bt, \sliceAlt)$.
		\end{enumerate}
	\end{enumerate}
\end{fact}
\begin{fact}\label{app:fact:beta-st}
	For all $\slice, \sliceAlt \in \slices(\tD)$ with $h(v_{\slice}) = h(v_{\sliceAlt})$ there is a bijection $\tilde{\beta}_{\slice, \sliceAlt}\colon M_{\slice} \to M_{\sliceAlt}$ such that for all $\tup{c} \in M_{\slice}$ with $\tup{d} \isdef \tilde{\beta}_{\slice, \sliceAlt}(\tup{c})$ it holds that
	\begin{enumerate}
		\item $h(w_{\tup{c}}) = h(w_{\tup{d}})$ and
		\item $\stp(\tup{c}, \slice) = \stp(\tup{d}, \sliceAlt)$.
	\end{enumerate}
\end{fact}

\smallskip\noindent
With these facts we are already able to prove parts (2a) and (3a) of Lemma~\ref{app:lemma:g-is-stable}: %
\begin{proposition}\label{app:prop:e-neighbors-of-projections}
	\begin{enumerate}[(a)]
		\item
		For all $\slice, \sliceAlt \in \slices(\tD)$ with $g(v_{\slice}) = g(v_{\sliceAlt})$ 
		it holds that
		\begin{equation*}
				\mset[\big]{ (\stp(\tup{c}, \slice), g(w_{\tup{c}})) \mid \tup{c} \in M_{\slice} }
				=
				\mset[\big]{ (\stp(\tup{c}, \sliceAlt), g(w_{\tup{c}})) \mid \tup{c} \in M_{\sliceAlt} }.
		\end{equation*}

		\item
		For all $\projection, \projectionAlt \in \Projections(\tD) \setminus \slices(\tD)$ with $g(v_{\projection}) = g(v_{\projectionAlt})$ 
		it holds that
		\begin{equation*}
				\mset[\big]{ (\stp(\tup{c}, \projection), g(w_{\tup{c}})) \mid \tup{c} \in M_{\projection} }
				=
				\mset[\big]{ (\stp(\tup{c}, \projectionAlt), g(w_{\tup{c}})) \mid \tup{c} \in M_{\projectionAlt} }.\endProposition
		\end{equation*}
	\end{enumerate}
\end{proposition}
\begin{proof}
	\begin{enumerate}[(a)]
		\item This follows from Fact~\ref{app:fact:beta-st}
                  since, by definition, $g(v_{\bar{x}})=h(v_{\bar{x}})$
                  for all $\bar{x}\in \tD\cup\slices(\tD)$.
		\item This follows directly from the definition of $g$
                  on $\Projections(\tD)\setminus\slices(\tD)$.\qedhere
	\end{enumerate}
\end{proof}

\medskip\noindent
To prove the rest of Lemma~\ref{app:lemma:g-is-stable}, we need a bunch of technical lemmas:
\begin{lemma}\label{app:lemma:stp-identifies-proj}
	For all $\at \in \tD$ and all $\projection, \projection' \in \Projections(\at)$ it holds that:\; $\stp(\at, \projection) = \stp(\at, \projection') \iff \projection = \projection'$.
\end{lemma}
\begin{proof}
	\enquote{$\Longleftarrow$} is obvious.
	For \enquote{$\Longrightarrow$} let $(a_1, \dots, a_k) = \at$, $(p_1, \dots, p_{\ell}) = \projection$, $(p'_1, \ldots, p'_{\ell'}) = \projection'$ and assume that $\stp(\at, \projection) = \stp(\at, \projection')$.
	For every $j \in [\ell]$ there exists an $i \in [k]$ such that $(i, j) \in \stp(\at,\projection) = \stp(\at,\projection')$, and hence $p_j = a_i = p'_j$.
	Thus, $\ell' \geq \ell$ and $p'_j = p_j$ for all $j \leq \ell$.
	Furthermore, in particular for $j = \ell'$ we know that there is an $i \in [k]$ such that $(i, \ell') \in \stp(\at, \projection')$, i.e., $(i, \ell') \in \stp(\at, \projection)$ must hold as well.
	This implies that $\ell \geq \ell'$, and hence $\ell = \ell'$ and $\projection = \projection'$.
\end{proof}

\begin{lemma}\label{app:lemma:bijection}
	For all $\at, \bt \in \tD$ we have: \ $\stp(\at) = \stp(\bt)$ iff there is a bijection $\pproj\colon \Projections(\at) \to \Projections(\bt)$ such that for all $\projection \in \Projections(\at)$ we have $\stp(\at, \projection) = \stp(\bt, \pproj(\projection))$.
\end{lemma}
\begin{proof}
	Let $\at = (a_1, \dots, a_k)$ and $\bt = (b_1, \dots, b_\ell)$ for $k,\ell \in \NNpos$.
	
	\enquote{$\Longleftarrow$}:\;
	By assumption there is a bijection $\pproj\colon \Projections(\at) \to \Projections(\bt)$ such that $\stp(\at, \projection) = \stp(\bt, \pproj(\projection))$ holds for all $\projection \in \Projections(\at)$.

	Consider an arbitrary slice $\sliceAlt'$ with $\sliceAlt' \in \slices(\bt)$ and $\tset(\sliceAlt') = \tset(\bt)$ (obviously, such a slice exists).
	In particular, there exists a $j \in [\ar(\sliceAlt')]$ such that $b_{\ell} = t'_{j}$, and hence $(\ell,j) \in \stp(\bt,\sliceAlt')$.
	Let $\slice' \isdef \pproj^{-1}(\sliceAlt')$ and note that, by the choice of $\pproj$ we have $\stp(\at, \slice') = \stp(\bt, \sliceAlt')$.
	Hence, from $(\ell, j) \in \stp(\bt, \sliceAlt')$ we obtain that $(\ell, j) \in \stp(\at,\slice')$, and hence $k \geq \ell$. 

	Now, let us fix an arbitrary slice $\slice \in \slices(\at)$ with $\tset(\slice) = \tset(\at)$.
	Let $\sliceAlt \isdef \pproj(\slice)$.
	By the choice of $\pproj$ we have $\stp(\at, \slice) = \stp(\bt, \sliceAlt)$.
	From $\tset(\slice) = \tset(\at)$ and $\slice \in \slices(\at)$, we obtain that for each $i \in [k]$ there exists exactly one $j_i \in [\ar(\slice)]$ with $(i,j_i) \in \stp(\at, \slice)$.
	From $\stp(\at, \slice) = \stp(\bt, \sliceAlt)$ we obtain that $(i, j_i) \in \stp(\bt, \sliceAlt)$ for every $i \in [k]$.
	For $i = k$ this in particular implies that $\ell \geq k$.
	In summary, we have shown that $k = \ell$.

	Finally, let us fix arbitrary $i, \hati \in [k]$.
	We have $(i,\hati) \in \stp(\at) \iff a_i = a_{\hati} \iff j_i=j_{\hati} \iff (i, j_i), (\hati, j_i) \in \stp(\at, \slice)$.
	From $\stp(\at, \slice) = \stp(\bt, \sliceAlt)$ we obtain that $(i, j_i), (\hati, j_i) \in \stp(\at, \slice) \iff (i, j_i), (\hati, j_i) \in \stp(\bt, \sliceAlt) \iff$ $b_i = t_{j_i} = b_{\hati} \iff b_i = b_{\hati} \iff (i,\hati) \in \stp(\bt)$.

	In summary, we obtain that $\stp(\at)=\stp(\bt)$. This completes the proof of \enquote{$\Longleftarrow$}.
	\medskip

	\enquote{$\Longrightarrow$}:\;
	By assumption we have $\stp(\at) = \stp(\bt)$.
	From \cite[Lemma~3.7(b)]{ScheidtSchweikardt_MFCS25} we obtain that $k = \ell$ and the function $\beta\colon \tset(\at) \to \tset(\bt)$ with $\beta(a_i) \isdef b_i$ for all $i \in [k]$ is well-defined and bijective.
	For any $\projection \in \Projections(\at)$ of the form $(p_1, \ldots, p_n) = \projection$ (in particular, $1 \leq n = \ar(\projection)$), we let $\pproj(\projection)\isdef (\beta(p_1), \ldots, \beta(p_n))$.

	We first show that $\pproj(\projection) \in \Projections(\bt)$:
	Let $\projectionAlt \isdef \pproj(\projection)$, i.e., $\projectionAlt = (q_1, \ldots, q_n)$ and $q_i = \beta(p_i)$ for all $i \in [n]$.
	Since $\projection \in \Projections(\at)$, there is a tuple $(i_1, \dots, i_{n})$ of pairwise distinct elements witnessing this, i.e., $\projection = \pi_{(i_1, \dots, i_{n})}(\at)$. That also means that $\projection = (a_{i_1}, \dots, a_{i_n})$, which means that $\projectionAlt = (\beta(a_{i_1}), \dots, \beta(a_{i_n})) = (b_{i_1}, \dots, b_{i_n}) = \pi_{(i_1, \dots, i_n)}(\bt)$. Thus, $\pi_{(i_1, \dots, i_n)}(\bt)$ witnesses that $\projectionAlt \in \Projections(\bt)$.

	Next, we show that $\stp(\at, \projection) = \stp(\bt, \projectionAlt)$:
	For arbitrary $i \in [k]$ and $j \in [n]$ we have $(i,j) \in \stp(\at, \projection) \iff a_i = p_j \iff \beta(a_i) = \beta(p_j) \iff b_i = q_j \iff (i,j) \in \stp(\bt, \projectionAlt)$.
	Hence, we have $\stp(\at, \projection) = \stp(\bt, \projectionAlt)$.

	In summary, we have shown that $\pproj$ is a mapping  $\pproj\colon \Projections(\at) \to \Projections(\bt)$ that satisfies $\stp(\at, \projection) = \stp(\bt, \pproj(\projection))$ for all $\projection \in \Projections(\at)$.

	Next, we show that $\pproj$ is \emph{injective}:
	Consider arbitrary $\projection, \projection' \in \Projections(\at)$ of the form $(p_1, \dots, p_n)$ and $(p'_1, \dots, p'_m)$ such that $\pproj(\projection) = \pproj(\projection')$.
	By definition of $\pproj$ we have $n = m$, and $\beta(p_i) = \beta(p'_i)$ for all $i \in [n]$.
	Since $\beta$ is injective, we obtain that $p_i = p'_i$ holds for all $i \in [n]$.
	Thus, $\projection = \projection'$.
	Hence, $\pproj$ is injective.

	Next, we show that $\pproj$ is \emph{surjective}:
	Consider an arbitrary $\projectionAlt \in \Projections(\bt)$ of the form $(q_1, \ldots, q_n)$.
	Let $(i_1, \dots, i_n)$ be pairwise distinct indices witnessing $\projectionAlt \in \Projections(\bt)$, i.e., $\projectionAlt = \pi_{(i_1, \dots, i_n)}(\bt)$.
	For $i \in [n]$ let $p_i \isdef \beta^{-1}(q_i)$, and let $\projection \isdef (p_1, \ldots, p_n)$.
	Since, $\projection = (\beta^{-1}(b_{i_1}), \dots, \beta^{-1}(b_{i_n})) = (a_{i_1}, \dots, a_{i_n})$, it holds that $\pi_{(i_1, \dots, i_n)}(\at) = \projection$.
	Therefore, $\projection \in \Projections(\at)$.
	We are done by noting that $\pproj(\projection) = \projectionAlt$.
	This completes the proof of \enquote{$\Longrightarrow$} and the proof of Lemma~\ref{app:lemma:bijection}.
\end{proof}

\begin{lemma}\label{app:lemma:bijection-props}
	Let $\at, \bt \in \tD$ with $\stp(\at) = \stp(\bt)$ and let $\pproj$ be a bijection according to Lemma~\ref{app:lemma:bijection}. The following is true:
	\begin{enumerate}
		\item $\pproj$ is unique,
		\item for all $\projection \in \Projections(\at)$ it holds that $\ar(\projection) = \ar(\pproj(\projection))$ and $\stp(\projection) = \stp(\pproj(\projection))$ and
		\item for all $\slice \in \slices(\at)$ it holds that
                  $\pproj(\slice) = \tilde{\beta}_{\at,
                    \bt}(\slice)$, where $\tilde{\beta}_{\at,
                    \bt}(\slice)$ is the bijection from Fact~\ref{app:fact:beta-ab}\,\eqref{item:HD:app:fact:beta-ab}. 
	\end{enumerate}
\end{lemma}
\begin{proof}
	(1) follows from Lemma~\ref{app:lemma:stp-identifies-proj}.
	
	Next, we show (2).
	Let $\projection \in \Projections(\at)$ and $\projectionAlt = \pproj(\projection)$.
	For $i \isdef \ar(\projection)$ there must be a $j_i$ such that $p_i = a_{j_i}$.
	I.e., $(j_i,i) \in \stp(\at,\projection) = \stp(\bt,\projectionAlt)$.
	Hence, $b_{j_i} = q_i$ and, in particular, $i \leq \ar(\projectionAlt)$, i.e., $\ar(\projection) \leq \ar(\projectionAlt)$.
	By a similar reasoning we obtain that $\ar(\projectionAlt) \leq \ar(\projection)$:
	By assumption, $\projectionAlt \in \Projections(\bt)$.
	Hence, in particular for $i \isdef \ar(\projectionAlt)$ there exists a $j_i$ such that $q_i = b_{j_i}$.
	I.e., $(j_i,i) \in \stp(\bt, \projectionAlt) = \stp(\at, \projection)$.
	Hence, $a_{j_i} = p_i$ and, in particular, $i \leq \ar(\projection)$, i.e., $\ar(\projectionAlt) \leq \ar(\projection)$.
	This proves that $\ar(\projection) = \ar(\pproj(\projection))$.
	
	Let $(i,j) \in \stp(\projection)$, i.e., $p_i = p_j$. Then there is an $i'$ such that $(i', i), (i', j) \in \stp(\at, \projection) = \stp(\bt, \projectionAlt)$. Thus, $b_{i'} = q_i = q_j$, i.e., $(i,j) \in \stp(\projectionAlt)$. Thus, $\stp(\projection) \subseteq \stp(\projectionAlt)$.
	Let $(i,j) \in \stp(\projectionAlt)$, i.e., $q_i = q_j$. Then there is an $i'$ such that $(i', i), (i', j) \in \stp(\bt, \projectionAlt) = \stp(\at, \projection)$. Thus, $a_{i'} = p_i = p_j$, i.e., $(i,j) \in \stp(\projection)$. Thus, $\stp(\projectionAlt) \subseteq \stp(\projection)$.

	(3) follows from (1) and the uniqueness of $\pslices$ according to~\cite[page 88:8]{ScheidtSchweikardt_MFCS25}.
\end{proof}

\begin{lemma}\label{app:lemma:projections-dont-matter}
	Let $\at, \bt \in \tD$ with $\stp(\at) = \stp(\bt)$.
	If there is a $\projection \in \Projections(\at)$ with $\projectionAlt \isdef \pproj(\projection)$ such that
	\[
		\mset[\big]{ (\stp(\tup{c}, \projection), g(w_{\tup{c}})) \mid \tup{c} \in M_{\projection} }
		\neq
		\mset[\big]{ (\stp(\tup{c}, \projectionAlt), g(w_{\tup{c}})) \mid \tup{c} \in M_{\projectionAlt} }
	\]
	then for every slice $\slice \in \slices(\at)$ with $\tset(\slice) = \tset(\projection)$ and $\sliceAlt \isdef \pproj(\slice)$ it holds that
	\[
		\mset[\big]{ (\stp(\tup{c}, \slice), g(w_{\tup{c}})) \mid \tup{c} \in M_{\slice} }
		\neq
		\mset[\big]{ (\stp(\tup{c}, \sliceAlt), g(w_{\tup{c}})) \mid \tup{c} \in M_{\sliceAlt} }.
	\]
\end{lemma}
\begin{proof}
Let $\projection \in \Projections(\at)$ and $\projectionAlt \isdef \pproj(\projection)$ such that
\begin{equation*}
		\mset[\big]{ (\stp(\tup{c}, \projection), g(w_{\tup{c}})) \mid \tup{c} \in M_{\projection} }
		\neq
		\mset[\big]{ (\stp(\tup{c}, \projectionAlt), g(w_{\tup{c}})) \mid \tup{c} \in M_{\projectionAlt} }.
\end{equation*}
Then, $\projection\neq\emptytuple$ since $M_{\emptytuple}=\tD$ and $\pproj(\emptytuple)=\emptytuple$.
              
Let $\projection = (p_1, \dots, p_k)$ with $\ell\deff
|\tset(\projection)| \geq 1$. Let $\slice = (s_1, \dots, s_{\ell}) \in
\slices(\at)$ be a slice such that $\tset(\projection) =
\tset(\slice)$, and let $\beta\colon [k] \to [\ell]$ be the unique map such that for all $i \in [k]$ we have $p_i = s_{\beta(i)}$.

	According to Lemma~\ref{app:lemma:bijection-props}, $\ar(\projection) = \ar(\projectionAlt)$, thus $\projectionAlt = (q_1, \dots, q_k)$.
	Let $\sliceAlt = \pproj(\slice)$.

	\begin{claim*}
		$\tset(\sliceAlt) = \tset(\projectionAlt)$ and for all $j \in [k]$ we have $q_j = t_{\beta(j)}$.
	\end{claim*}
	\noindent\emph{Proof of Claim}:\;
	Let $j \in [k]$. Choose $i$ such that $(i,j) \in \stp(\at, \projection)$, which also means that $(i,j) \in \stp(\bt, \projectionAlt)$, i.e., $a_i = p_j$ and $b_i = q_j$.
	Since $p_j = s_{\beta(j)}$, we have $(i, \beta(j)) \in \stp(\at, \slice)= \stp(\bt, \sliceAlt)$. Thus, $b_i = t_{\beta(j)}$, i.e., $q_j = t_{\beta(j)}$.

	Since $\img(\beta) = [\ell]$, $\tset(\sliceAlt) = \tset(\projectionAlt)$ must hold.
	\hfill$\blacksquare$\smallskip

	We know that there is a $\lambda$ such that
	\begin{equation}\label{eq:Ungleich}
		\mset[\big]{ g(w_{\tup{c}}) \mid \tup{c} \in
                  M_{\projection}, \ \stp(\tup{c}, \projection) = \lambda } 
		\neq 
		\mset[\big]{ g(w_{\tup{c}}) \mid \tup{c} \in
                  M_{\projectionAlt}, \ \stp(\tup{c}, \projectionAlt) = \lambda }.
	\end{equation}
	Hence, it suffices to show that there is a $\lambda'$ such that
	\[
		\mset[\big]{ g(w_{\tup{c}}) \mid \tup{c} \in
                  M_{\slice}, \stp(\tup{c}, \ \slice) = \lambda' } 
		\neq 
		\mset[\big]{ g(w_{\tup{c}}) \mid \tup{c} \in
                  M_{\sliceAlt}, \stp(\tup{c}, \ \sliceAlt) = \lambda' }.
	\]
	Let $\lambda' \isdef \set{ (i,\beta(j)) \mid (i,j) \in \lambda }$.

	\begin{claim*}
		For all $\tup{c} \in M_{\projection}$ with $\stp(\tup{c}, \projection) = \lambda$, we have $\tup{c} \in M_{\slice}$ and $\stp(\tup{c}, \slice) = \lambda'$. Analogously, for all $\tup{c} \in M_{\projectionAlt}$ with $\stp(\tup{c}, \projectionAlt) = \lambda$ we have $\tup{c} \in M_{\sliceAlt}$ and $\stp(\tup{c}, \sliceAlt) = \lambda'$.
	\end{claim*}
	\noindent\emph{Proof of Claim}:\;
	Clearly, $\slice \in \slices(\tup{c})$ since $\slice$ is a slice and $\tset(\slice) = \tset(\projection)$.
	Thus, it remains to show that $\stp(\tup{c}, \slice) = \lambda'$.
	It is easy to see that $\lambda' \subseteq \stp(\tup{c}, \slice)$:
	let $(i,j) \in \lambda$, then $(i, \beta(j)) \in \lambda'$.
	By definition, $(i,j) \in \lambda$ implies that $c_i = p_j$. 
	By definition of $\beta$ we have $p_j = s_{\beta(j)}$, i.e., $c_i = s_{\beta(j)}$.
	Hence, $(i, \beta(j)) \in \stp(\tup{c}, \slice)$.

	To show that $\stp(\tup{c}, \slice) \subseteq \lambda'$, let
        $(i,j) \in \stp(\tup{c}, \slice)$. I.e., $c_i = s_j$. Since
        $\tset(\slice) = \tset(\projection)$, there must be a $j' \in
        [k]$ such that $p_{j'} = s_j$. Thus, $j=\beta(j')$ and $c_i =
        p_{j'}$, which by definition means that $(i, j') \in
        \lambda$. By definition of $\lambda'$, this yields $(i,
        \beta(j')) \in \lambda'$. Since $j=\beta(j')$, we obtain that $(i,j) \in \lambda'$.

	For $\projectionAlt$ and $\sliceAlt$ this works analogously.%
	\hfill$\blacksquare$\smallskip
	\begin{claim*}
		For all $\tup{c} \in M_{\slice}$ with $\stp(\tup{c}, \slice) = \lambda'$, we have $\tup{c} \in M_{\projection}$ and $\stp(\tup{c}, \projection) = \lambda$. Analogously, 
		for all $\tup{c} \in M_{\sliceAlt}$ with $\stp(\tup{c}, \sliceAlt) = \lambda'$, we have $\tup{c} \in M_{\projectionAlt}$ and $\stp(\tup{c}, \projectionAlt) = \lambda$.
	\end{claim*}
	\noindent\emph{Proof of Claim}:\;
	Clearly, $\tset(\projection) \subseteq \tset(\tup{c})$. We have to show that $\stp(\tup{c}, \projection) = \lambda$ and $\projection \in \Projections(\tup{c})$.

	Let $(i,j) \in \stp(\tup{c}, \projection)$, i.e., $c_i =
        p_j$. Since $p_j = s_{\beta(j)}$, we have $(i, \beta(j)) \in \stp(\tup{c}, \slice)= \lambda'$.
	Thus, there is a $j'$ such that $(i,j') \in \lambda$ and $\beta(j') = \beta(j)$.
	That means $p_j = s_{\beta(j)} = s_{\beta(j')} = p_{j'}$, i.e., $p_{j} = p_{j'}$.
	From $(i, j') \in \lambda$ and $p_j=p_{j'}$ we obtain that
        also $(i,j) \in \lambda$. This proves that $\stp(\tup{c}, \projection) \subseteq \lambda$.

        For proving ``$\supseteq$'', consider an arbitrary
	 $(i,j) \in \lambda$. Then, $(i, \beta(j)) \in
         \lambda'=\stp(\ct,\slice)$, i.e., $c_i = s_{\beta(j)}$. Since
         $p_j = s_{\beta(j)}$, this also means that $c_i = p_j$, i.e.,
         $(i,j) \in \stp(\tup{c}, \projection)$. Thus, $\lambda
         \subseteq \stp(\tup{c}, \projection)$. In total, we have
         shown that $\stp(\tup{c}, \projection) = \lambda$.

It remains to show that $\projection \in \Projections(\tup{c})$.
Note that by our choice of $\lambda$, there exist pairwise distinct
indices $m_1,\ldots,m_k$ such that $(m_j,j)\in\lambda$ for all
$j\in[k]$ (this holds because due to equation \eqref{eq:Ungleich}
there exists a $\ct'\in M_\projection$ (i.e., $\projection\in\Projections(\ct')$)
with $\stp(\ct',\projection)=\lambda$ or a $\ct'\in
M_{\projectionAlt}$ (i.e., $\projectionAlt\in\Projections(\ct')$) with $\stp(\ct',\projectionAlt)=\lambda$).
We have already shown that $\stp(\ct,\projection)=\lambda$, and hence
we obtain that $\projection=\pi_{(m_1,\ldots,m_k)}(\ct)$, i.e.,
$\projection\in\Projections(\ct)$.
This finishes the proof of the first statement of the claim.

The claim's second statement (i.e., the statement for
$\projectionAlt$ and $\sliceAlt$) can be shown analogously.
	\hfill$\blacksquare$\smallskip

	Combining the last two claims finishes the proof of Lemma~\ref{app:lemma:projections-dont-matter}.
\end{proof} 

\medskip\noindent
This now allows us to show part (1) of Lemma~\ref{app:lemma:g-is-stable}:
\begin{proposition}\label{app:prop:e-neigbors-of-tuples}
		For all $\at, \bt \in \tD$ with $g(w_{\at}) = g(w_{\bt})$ and the bijection $\pproj$ according to Lemma~\ref{app:lemma:bijection} and every $\projection \in \Projections(\at)$ with $\projectionAlt \isdef \pproj(\projection)$ it holds that $g(v_{\projection}) = g(v_{\projectionAlt})$.
		In particular, this shows that
		\[
			\mset[\big]{ (\stp(\at, \projection), g(v_{\projection})) \mid \projection \in \Projections(\at) }
			=
			\mset[\big]{ (\stp(\bt, \projection), g(v_{\projection})) \mid \projection \in \Projections(\bt) }.\endProposition
		\]
\end{proposition}
\begin{proof}
		Let $\at, \bt \in \tD$ with $g(w_{\at}) = g(w_{\bt})$.
From Fact~\ref{app:fact:h-encodes-stp} we know that $\stp(\at)=\stp(\bt)$.
               Let $\pproj$ be the bijection according to Lemma~\ref{app:lemma:bijection}.

		For contradiction, assume that there are $\projection \in \Projections(\at)$ and $\projectionAlt \isdef \pproj(\projection)$ such that $g(v_{\projection}) \neq g(v_{\projectionAlt})$.
		If $\projection$ is not a slice, i.e., $\projection\in
                \Projections(\tD)\setminus\slices(\tD)$, then
                $g(v_{\projection}) \neq g(v_{\projectionAlt})$ by
                definition means that $\stp(\projection) \neq \stp(\projectionAlt)$ or
		\[
			\mset[\big]{ 
				\bigl(
					\stp(\tup{c}, \projection),
					h(w_{\tup{c}})
				\bigr)
				\mid
				\tup{c} \in M_{\projection}
			}
			\neq
			\mset[\big]{ 
				\bigl(
					\stp(\tup{c}, \projectionAlt),
					h(w_{\tup{c}})
				\bigr)
				\mid
				\tup{c} \in M_{\projectionAlt}
			}.
		\]
		By Lemma~\ref{app:lemma:bijection-props}(2) we know that $\stp(\projection) = \stp(\projectionAlt)$ holds.
		Thus, by Lemma~\ref{app:lemma:projections-dont-matter} there is a slice $\slice \in \slices(\at)$ with $\sliceAlt \isdef \pproj(\slice)$ such that 
		\[
			\mset[\big]{ (\stp(\tup{c}, \slice), g(w_{\tup{c}})) \mid \tup{c} \in M_{\slice} }
			\neq
			\mset[\big]{ (\stp(\tup{c}, \sliceAlt), g(w_{\tup{c}})) \mid \tup{c} \in M_{\sliceAlt} }.
		\]
		Since Lemma~\ref{app:lemma:bijection-props}(3) yields that $\sliceAlt$ is a slice as well, $g(v_{\slice}) \neq g(v_{\sliceAlt})$ must hold according to Proposition~\ref{app:prop:e-neighbors-of-projections}(a).

		Thus, we can assume that $\projection$ and $\projectionAlt$ are slices.
		Since $\projectionAlt = \tilde{\beta}_{\at, \bt}(\projection)$ according to Lemma~\ref{app:lemma:bijection-props}(3) and $h(v_{\projection}) = g(v_{\projection}) \neq g(v_{\projectionAlt}) = h(v_{\projectionAlt})$, Fact~\ref{app:fact:beta-ab}(2) yields that $h(w_{\at}) \neq h(w_{\bt})$ must hold.
		By definition, this means that $g(w_{\at}) \neq g(w_{\bt})$ must hold as well.
		This is a contradiction to our assumption that
                $g(w_{\at}) = g(w_{\bt})$.
                This completes the proof of Proposition~\ref{app:prop:e-neigbors-of-tuples}.
\end{proof}

\medskip\noindent
We show part (2b) of Lemma~\ref{app:lemma:g-is-stable} next:
\begin{proposition}\label{app:prop:f-neighbors-slices}
	For all $\slice, \sliceAlt \in \slices(\tD)$ with $g(v_{\slice}) = g(v_{\sliceAlt})$ it holds that
	\begin{equation}\label{eq:GoalLemmaD9}
		\mset[\big]{ (\stp(\slice, \projection'), g(v_{\projection'})) \mid
			\projection' \in N_{\slice} }
		=
		\mset[\big]{ (\stp(\sliceAlt, \projection'), g(v_{\projection'})) \mid
			\projection' \in N_{\sliceAlt} }.\endProposition
	\end{equation}
\end{proposition}
\begin{proof}
By assumption we are given $\slice,\sliceAlt\in\slices(\tD)$ with
$g(v_\slice)=g(v_\sliceAlt)$, i.e., $h(v_\slice)=h(v_\sliceAlt)$.
Consider any $\lambda$ and $c$ such that
$(\lambda,c) = \bigl(\stp(\slice,\projection),g(v_\projection)\bigr)$
for some $\projection\in N_\slice$ or $(\lambda,c) = \bigl(\stp(\sliceAlt,\projectionAlt), g(v_\projection)\bigr)$ for some $\projectionAlt\in N_\sliceAlt$.
Let
\begin{eqnarray*}
  P_{\lambda,c} &\isdef & \set[\big]{ \projection \in N_{\slice} \mid 
	\big(\stp(\slice, \projection),g(v_{\projection})\big) = (\lambda,c) }\,,\\
  Q_{\lambda,c} &\isdef & \set[\big]{ \projection \in N_{\sliceAlt} \mid 
	\big(\stp(\sliceAlt, \projection),g(v_{\projection})\big) = (\lambda, c)}\,.
\end{eqnarray*}
Note that in order to prove equation \eqref{eq:GoalLemmaD9}, it suffices to
prove that $|P_{\lambda,c}|=|Q_{\lambda,c}|$.

In the following, we consider the case where $(\lambda,c)=\big(\stp(\slice,\projection),g(v_\projection)\big)$
for some $\projection\in N_\slice$, i.e., the case where we know that $P_{\lambda,c}\neq
\emptyset$ (the case where we know that $Q_{\lambda,c}\neq \emptyset$
can be handled analogously).

By definition of $N_\slice$ we know that $\projection\in N_\slice$
implies that either $\tset(\projection)\subseteq\tset(\slice)$ or
$\tset(\projection)\supset\tset(\slice)$.
Thus, $\lambda$ either
enforces that all $\projection\in N_\slice$ with
$\lambda=\stp(\slice,\projection)$ satisfy
$\tset(\projection)\subseteq\tset(\slice)$, or it enforces that
all $\projection\in N_\slice$ with
$\lambda=\stp(\slice,\projection)$ satisfy
$\tset(\projection)\supset\tset(\slice)$.

\bigskip

\emph{Case 1:} \ $\lambda$ enforces that
$\tset(\projection)\supset\tset(\slice)$ for all $\projection$ with $\lambda=\stp(\slice,\projection)$.
\\
Note that for all $\projection, \projection' \in P_{\lambda,c}$ and
$\projectionAlt,\projectionAlt'\in Q_{\lambda,c}$ we
have $g(v_\projection)=g(v_{\projection'})=g(v_\projectionAlt)=g(v_{\projectionAlt'})$, and therefore
\begin{equation}\label{eq:LemmaD9_UseStability}
  \begin{array}{llll}
& \mset[\big]{ 
	\bigl(
		\stp(\at, \projection),
		g(\wat)
	\bigr)
	\mid
	\at \in M_{\projection}
     }
 & =
 &  \mset[\big]{ 
	\bigl(
		\stp(\at', \projection'),
		g(w_{\at'})
	\bigr)
	\mid
	\at' \in M_{\projection'}
    }
\\
 =
& \mset[\big]{ 
	\bigl(
		\stp(\bt, \projectionAlt),
		g(\wbt)
	\bigr)
	\mid
	\bt \in M_{\projectionAlt}
     }
 & =
 &  \mset[\big]{ 
	\bigl(
		\stp(\bt', \projectionAlt'),
		g(w_{\bt'})
	\bigr)
	\mid
	\bt' \in M_{\projectionAlt'} 
    }
\end{array}
\end{equation}
(in case that $\projection\in\Projections(\tD)\setminus\slices(\tD)$,
this follows by the definition of $g$; and in case that
$\projection\in\slices(\tD)$, this follows from Fact~\ref{app:fact:beta-st}).

Let us fix arbitrary $\lambda'$ and $d$ such that the tuple
$(\lambda',d)$ is included in the above multisets.
For $\projection\in P_{\lambda,c}$ and $\projectionAlt\in
Q_{\lambda,c}$ let
\begin{eqnarray*}
  A_{\lambda',d}(\projection) & \deff & \setc{\at\in M_\projection}{\big(\stp(\at,\projection),g(w_\at)\big)=(\lambda',d)}\,,
\\
  B_{\lambda',d}(\projectionAlt) & \deff & \setc{\bt\in M_\projectionAlt}{\big(\stp(\bt,\projectionAlt),g(w_\bt)\big)=(\lambda',d)}\,.
\end{eqnarray*}
From equation~\eqref{eq:LemmaD9_UseStability} we obtain for all
$\projection,\projection'\in P_{\lambda,c}$ and all
$\projectionAlt,\projectionAlt'\in Q_{\lambda,c}$ that
\begin{equation}\label{eq:LemmaD9_ell}
  |A_{\lambda',d}(\projection)|
  \ = \ 
  |A_{\lambda',d}(\projection')|
  \ = \ 
  |B_{\lambda',d}(\projectionAlt)|
  \ = \ 
  |B_{\lambda',d}(\projectionAlt')|
  \ =: \ \ell_{\lambda',d}\,,
\end{equation} 
and clearly, $\ell_{\lambda',d}\geq 1$.

\begin{claim*} For all $\projection,\projection'\in P_{\lambda,c}$ with
  $\projection\neq\projection'$ we
  have $A_{\lambda',d}(\projection)\cap
  A_{\lambda',d}(\projection')=\emptyset$. \\
  Analogously, for all $\projectionAlt,\projectionAlt'\in Q_{\lambda,c}$ with
  $\projectionAlt\neq\projectionAlt'$ we
  have $B_{\lambda',d}(\projectionAlt)\cap
  B_{\lambda',d}(\projectionAlt')=\emptyset$. 
\end{claim*}
\begin{claimproof}
  Let $\projection,\projection'\in P_{\lambda,c}$.
  Consider an arbitrary $\at\in A_{\lambda',d}(\projection)$. Since $\at\in M_\projection$, for every $i\in
\set{1,\ldots,\ar(\projection)}$ there exists a
$j_i\in\set{1,\ldots,\ar(\at)}$ such that $p_i=a_{j_i}$, and hence
$(j_i,i)\in\stp(\at,\projection)=\lambda'$.

If $\at$ also belongs to $A_{\lambda',d}(\projection')$, then
$\stp(\at,\projection')=\lambda'$ implies that $p'_i=a_{j_i}=p_i$ for
every $i\in\set{1,\ldots,\ar(\projection)}$.
Furthermore, we know that $g(v_{\projection'})=g(v_{\projection})$ and
hence, in particular, $\stp(\projection')=\stp(\projection)$ and thus
$\ar(\projection')=\ar(\projection)$. In summary, $\at\in
A_{\lambda',d}(\projection)\cap A_{\lambda',d}(\projection')$ implies
that $\projection'=\projection$.
This proves the first statement of the claim. The second statement can
be shown analogously.
\end{claimproof}

From this claim and from equation~\eqref{eq:LemmaD9_ell} we obtain
that
\begin{equation}\label{eq:LemmaD9_sizeUnion}
\Big|\, \bigcup_{\projection\in P_{\lambda,c}}
  A_{\lambda',d}(\projection) \ \Big|
  \ \ = \ \ |P_{\lambda,c}|\cdot \ell_{\lambda',d}
  \qquad\text{and}\qquad
\Big|\, \bigcup_{\projectionAlt\in Q_{\lambda,c}}
  B_{\lambda',d}(\projectionAlt) \ \Big|
  \ \ = \ \ |Q_{\lambda,c}|\cdot \ell_{\lambda',d}\,.
\end{equation}

\begin{claim*}
There is a $\lambda''$ such that for
all $\projection \in P_{\lambda,c}$,
all $\at\in A_{\lambda',d}(\projection)$,
all $\projectionAlt\in Q_{\lambda,c}$,
and
all $\bt\in B_{\lambda',d}(\projectionAlt)$
we have $\lambda''=\stp(\at, \slice) = \stp(\bt,\sliceAlt)$.
\end{claim*}
\begin{claimproof}
We fix an arbitrary $\projection\in P_{\lambda,c}$ and an arbitrary
$\at\in A_{\lambda',d}(\projection)$ and let
$\lambda''\deff\stp(\at,\slice)$.
Since we are in Case~1, we have
$\tset(\projection)\supset\tset(\slice)$.
\smallskip

Now, consider an arbitrary $\projection'\in P_{\lambda,c}$ and an
arbitrary $\at'\in A_{\lambda',d}(\projection')$. Our aim is to show
that $\stp(\at',\slice)=\lambda''$.

For ``$\supseteq$'' consider an arbitrary tuple
$(m,j)\in\lambda''$. Then, by our choice of $\lambda''$ we have
$a_m=s_j$. Since $\tset(\projection)\supset\tset(\slice)$, there is an
$i\in[\ar(\projection)]$ such that $s_j=p_i$. Thus,
$(j,i)\in\stp(\slice,\projection)=\lambda=\stp(\slice,\projection')$,
and hence $s_j=p'_i$. Furthermore, $a_m=s_j=p_i$, and hence $(m,i)\in
\stp(\at,\projection)=\lambda'=\stp(\at',\projection')$. Thus,
$a'_m=p'_i=s_j$, and therefore, $(m,j)\in\stp(\at', \slice)$.

For ``$\subseteq$'' consider an arbitrary tuple
$(m,j)\in\stp(\at',\slice)$. Then we have $a'_m=s_j$. Since
$\tset(\projection')\supset\tset(\slice)$, there is an
$i\in[\ar(\projection')]$ such that $s_j=p'_i$. Thus,
$(j,i)\in\stp(\slice,\projection')=\lambda=\stp(\slice,\projection)$,
and hence $s_j=p_i$. Furthermore, $a'_m=s_j=p'_i$, and hence $(m,i)\in
\stp(\at',\projection')=\lambda'=\stp(\at,\projection)$. Thus,
$a_m=p_i=s_j$, and therefore,
$(m,j)\in\stp(\at,\slice)=\lambda''$.
This proves that $\stp(\at',\slice)=\lambda''$ for all
$\projection'\in P_{\lambda,c}$ and all $\at'\in
P_{\lambda',d}(\projection')$.
\smallskip

Now, consider an arbitrary $\projectionAlt\in Q_{\lambda,c}$ and an
arbitrary $\bt\in B_{\lambda',d}(\projectionAlt)$. Our aim is to show
that $\stp(\bt,\sliceAlt)=\lambda''$.

For ``$\supseteq$'' consider an arbitrary tuple
$(m,j)\in\lambda''$. Then, by our choice of $\lambda''$ we have
$a_m=s_j$. Since $\tset(\projection)\supset\tset(\slice)$, there is an
$i\in[\ar(\projection)]$ such that $s_j=p_i$. Thus,
$(j,i)\in\stp(\slice,\projection)=\lambda=\stp(\sliceAlt,\projectionAlt)$,
and hence $t_j=q_i$. Furthermore, $a_m=s_j=p_i$, and hence $(m,i)\in
\stp(\at,\projection)=\lambda'=\stp(\bt,\projectionAlt)$. Thus,
$b_m=q_i=t_j$, and therefore, $(m,j)\in\stp(\bt,\sliceAlt)$.

For ``$\subseteq$'' consider an arbitrary tuple
$(m,j)\in\stp(\bt,\sliceAlt)$. Then we have $b_m=t_j$. Since
$\tset(\projectionAlt)\supset\tset(\sliceAlt)$, there is an
$i\in[\ar(\projectionAlt)]$ such that $t_j=q_i$. Thus,
$(j,i)\in\stp(\sliceAlt,\projectionAlt)=\lambda=\stp(\slice,\projection)$,
and hence $s_j=p_i$. Furthermore, $b_m=t_j=q_i$, and hence $(m,i)\in
\stp(\bt,\projectionAlt)=\lambda'=\stp(\at,\projection)$. Thus,
$a_m=p_i=s_j$, and therefore,
$(m,j)\in\stp(\at,\slice)=\lambda''$.
This proves that $\stp(\bt,\sliceAlt)=\lambda''$ for all
$\projectionAlt\in Q_{\lambda,c}$ and all $\bt\in
Q_{\lambda',d}(\projectionAlt)$.
\end{claimproof}

Choose $\lambda''$ according to the previous claim and consider the sets
\begin{eqnarray*}
P'_{\lambda'',d} &\isdef & \set[\big]{ \at \in M_{\slice} \mid 
			\big(\stp(\at, \slice),g(\wat) \big) = (\lambda'',d) }\,,\\
Q'_{\lambda'',d} &\isdef &\set[\big]{ \bt \in M_{\sliceAlt} \mid 
			\big(\stp(\bt, \sliceAlt),g(\wbt)\big) = (\lambda'', d)}\,.
\end{eqnarray*}
From Fact~\ref{app:fact:beta-st}
we obtain that
\begin{equation}\label{eq:LemmaD9_sizeOfPStrichQStrich}
  |P'_{\lambda'',d}| \ \ = \ \ |Q'_{\lambda'',d}|\,.
\end{equation}  
From the above claim we obtain that
\begin{equation}\label{eq:LemmaD9_UnionIsContained}
   \bigcup_{\projection\in P_{\lambda,c}} A_{\lambda',d}(\projection)
   \ \ \subseteq \ \ P'_{\lambda'',d}
   \qquad\text{and}\qquad
   \bigcup_{\projectionAlt\in Q_{\lambda,c}} B_{\lambda',d}(\projectionAlt)
   \ \ \subseteq \ \ Q'_{\lambda'',d}\,.
\end{equation}

\begin{claim*}
  For every $\ct\in P'_{\lambda'',d}$ there exists a
$\projection'\in P_{\lambda,c}$ such that $\ct\in A_{\lambda',d}(\projection')$.

Analogously,  for every $\dt\in Q'_{\lambda'',d}$ there exists a
$\projectionAlt'\in Q_{\lambda,c}$ such that $\dt\in B_{\lambda',d}(\projectionAlt')$.
\end{claim*}  
\begin{claimproof}
Fix an arbitrary $\projection\in P_{\lambda,c}$ and an $\at\in
A_{\lambda',d}(\projection)$.
Hence, $\projection\in\Projections(\at)$,
$g(v_{\projection})=c$, $g(w_{\at})=d$,
$\stp(\slice,\projection)=\lambda$, 
$\stp(\at,\projection)=\lambda'$,
and $\stp(\at,\slice)=\lambda''$.
\medskip

Consider an arbitrary $\ct\in P'_{\lambda'',d}$, i.e.,
$\stp(\ct,\slice)=\lambda''$ and $g(w_\ct)=d$.
Note that in order to prove the claim's first statement, it suffices to find a
$\projection'\in\Projections(\ct)$ such that
$\stp(\ct,\projection')=\lambda'$ and $g(v_{\projection'})=c$ and
$\stp(\slice,\projection')=\lambda$.

From Proposition~\ref{app:prop:e-neigbors-of-tuples}, and $g(w_\at)=g(w_\ct)$ we know that
\[
\mset{
(\stp(\at,\projectionAlt'),g(v_{\projectionAlt'}))
\mid
\projectionAlt'\in\Projections(\at)
}
\ \ = \ \
\mset{
(\stp(\ct,\projectionAlt'),g(v_{\projectionAlt'}))
\mid
\projectionAlt'\in\Projections(\ct)
}\,.
\]
Since $\projection\in\Projections(\at)$ and
$\big(\stp(\at,\projection), g(v_\projection)\big)=(\lambda',c)$, we
know that $(\lambda',c)$ belongs to both multisets. Thus, there exists
a $\projection'\in\Projections(\ct)$ such that
$\stp(\ct,\projection')=\lambda'$ and $g(v_{\projection'})=c$.
All that remains to be done is show that
$\stp(\slice,\projection')=\lambda$.

For ``$\subseteq$'' consider a tuple
$(i,j)\in\stp(\slice,\projection')$, i.e., $s_i=p'_j$.
Since $\projection'\in\Projections(\ct)$, there is a $k$ such that
$p'_j=c_k$. Hence,
$(k,j)\in\stp(\ct,\projection')=\lambda'=\stp(\at,\projection)$, and therefore $a_k=p_j$.
Furthermore, $s_i=p'_j=c_k$ implies that
$(k,i)\in\stp(\ct,\slice)=\lambda''=\stp(\at,\slice)$. Hence,
$a_k=s_i$. In summary, we have $s_i=a_k=p_j$, and thus
$(i,j)\in\stp(\slice,\projection)=\lambda$. This proves that
$\stp(\slice,\projection')\subseteq\lambda$.

For ``$\supseteq$'' consider a tuple
$(i,j)\in\lambda=\stp(\slice,\projection)$, i.e., $s_i=p_j$.
Since $\projection\in\Projections(\at)$, there is a $k$ such that
$p_j=a_k$. Hence,
$(k,j)\in\stp(\at,\projection)=\lambda'=\stp(\ct,\projection')$, and
therefore $c_k=p'_j$.
Furthermore, $s_i=p_j=a_k$ implies that
$(k,i)\in\stp(\at,\slice)=\lambda''=\stp(\ct,\slice)$. Hence,
$c_k=s_i$.
In summary, we have $s_i=c_k=p'_j$, and thus
$(i,j)\in\stp(\slice,\projection')$.
This proves that $\stp(\slice,\projection')=\lambda$.

In summary, we have proven the first statement of the claim.
\medskip

The second statement can be shown analogously:
Consider an arbitrary $\dt\in Q'_{\lambda'',d}$, i.e.,
$\stp(\dt,\sliceAlt)=\lambda''$ and $g(w_\dt)=d$.
Note that in order to prove the claim's second statement, it suffices to find a
$\projectionAlt'\in\Projections(\dt)$ such that
$\stp(\dt,\projectionAlt')=\lambda'$ and $g(v_{\projectionAlt'})=c$ and
$\stp(\sliceAlt,\projectionAlt')=\lambda$.

From Lemma~\ref{app:lemma:g-is-stable}(1), which we have already
proven, and $g(w_\at)=g(w_\dt)$ we know that
\[
\mset{
(\stp(\at,\projectionAlt'),g(v_{\projectionAlt'}))
\mid
\projectionAlt'\in\Projections(\at)
}
\ \ = \ \
\mset{
(\stp(\dt,\projectionAlt'),g(v_{\projectionAlt'}))
\mid
\projectionAlt'\in\Projections(\dt)
}\,.
\]
Since $\projection\in\Projections(\at)$ and
$\big(\stp(\at,\projection), g(v_\projection)\big)=(\lambda',c)$, we
know that $(\lambda',c)$ belongs to both multisets. Thus, there exists
a $\projectionAlt'\in\Projections(\dt)$ such that
$\stp(\dt,\projectionAlt')=\lambda'$ and $g(v_{\projectionAlt'})=c$.
All that remains to be done is show that
$\stp(\sliceAlt,\projectionAlt')=\lambda$.

For ``$\subseteq$'' consider a tuple
$(i,j)\in\stp(\sliceAlt,\projectionAlt')$, i.e., $t_i=q'_j$.
Since $\projectionAlt'\in\Projections(\dt)$, there is a $k$ such that
$q'_j=d_k$. Hence,
$(k,j)\in\stp(\dt,\projectionAlt')=\lambda'=\stp(\at,\projection)$, and therefore $a_k=p_j$.
Furthermore, $t_i=q'_j=d_k$ implies that
$(k,i)\in\stp(\dt,\sliceAlt)=\lambda''=\stp(\at,\slice)$. Hence,
$a_k=s_i$. In summary, we have $s_i=a_k=p_j$, and thus
$(i,j)\in\stp(\slice,\projection)=\lambda$. This proves that
$\stp(\sliceAlt,\projectionAlt')\subseteq\lambda$.

For ``$\supseteq$'' consider a tuple
$(i,j)\in\lambda=\stp(\slice,\projection)$, i.e., $s_i=p_j$.
Since $\projection\in\Projections(\at)$, there is a $k$ such that
$p_j=a_k$. Hence,
$(k,j)\in\stp(\at,\projection)=\lambda'=\stp(\dt,\projectionAlt')$, and
therefore $d_k=q'_j$.
Furthermore, $s_i=p_j=a_k$ implies that
$(k,i)\in\stp(\at,\slice)=\lambda''=\stp(\dt,\sliceAlt)$. Hence,
$d_k=t_i$.
In summary, we have $t_i=d_k=q'_j$, and thus
$(i,j)\in\stp(\sliceAlt,\projectionAlt')$.
This proves that $\stp(\sliceAlt,\projectionAlt')=\lambda$.
In summary, we have proved both statements of the claim.
\end{claimproof}  

From the above claim and
equation~\eqref{eq:LemmaD9_UnionIsContained}
we obtain that
\[
   \bigcup_{\projection\in P_{\lambda,c}} A_{\lambda',d}(\projection)
   \ \ = \ \
   P'_{\lambda'',d}
  \qquad\text{and}\qquad
   \bigcup_{\projectionAlt\in Q_{\lambda,c}} B_{\lambda',d}(\projectionAlt)
   \ \ = \ \ 
   Q'_{\lambda'',d}\,.
 \]
 Combining this with the equations~\eqref{eq:LemmaD9_sizeUnion}  and
 \eqref{eq:LemmaD9_sizeOfPStrichQStrich}, we obtain
\[
  \ell_{\lambda',d}\cdot |P_{\lambda,c}|
  \ \ = \ \
  |P'_{\lambda'',d}|
  \ \ = \ \
  |Q'_{\lambda'',d}|
  \ \ = \ \ 
  \ell_{\lambda',d}\cdot |Q_{\lambda,c}|\,.
\]   
Since $\ell_{\lambda',d}\neq 0$, this yields that
$|P_{\lambda,c}|=|Q_{\lambda,c}|$.
This completes the proof for Case~1.

\bigskip

\emph{Case 2:} \ $\lambda$ enforces that
$\tset(\projection)\subseteq\tset(\slice)$ for all $\projection$ with
$\lambda=\stp(\slice,\projection)$.
\\
In this case, $\lambda$ and $\slice$ completely determine
$\projection$, and $|P_{\lambda,c}|=1$ and $|Q_{\lambda,c}|\leq 1$.
Let $\projection$ be the
unique element in $P_{\lambda,c}$, let $n=\ar(\projection)$,
$r=\ar(\slice)=\ar(\sliceAlt)$, and let $i_1,\ldots,i_n\in[r]$ such
that $\projection=(p_1,\ldots,p_n)=(s_{i_1},\ldots,s_{i_n})$.
Let $\projectionAlt = (q_1,\ldots,q_n)\deff (t_{i_1},\ldots,t_{i_n})$.
Clearly, $\stp(\sliceAlt,\projectionAlt)=\stp(\slice,\projection)=\lambda$.
To complete the proof for Case~2, it suffices to show that
$\projectionAlt\in Q_{\lambda,c}$. I.e., we have to show that
$\projectionAlt\in\Projections(\tD)$ and $g(v_\projectionAlt)=c$.

Let us fix an arbitrary $\at\in M_\slice$, and let $\bt\deff \tilde{\beta}_{\slice,\sliceAlt}(\at)$, where $\tilde{\beta}_{\slice,\sliceAlt}$ is
the mapping from Fact~\ref{app:fact:beta-st}. From Fact~\ref{app:fact:beta-st} we know that
$\bt\in M_{\sliceAlt}$ and
 $h(w_\at)=h(w_\bt)$ and
 $\stp(\at,\slice)=\stp(\bt,\sliceAlt)$.

\begin{claim}\label{claim:Projections_LemmaD9}
  $\projectionAlt\in\Projections(\tD)$.
  Furthermore, there exist $\ct,\dt\in\tD$ with
  $h(w_\ct)=h(w_\dt)$ such that $\projection\in\Projections(\ct)$, $\projectionAlt\in\Projections(\dt)$, and $\stp(\ct,\projection)=\stp(\dt,\projectionAlt)$.
\end{claim}
\begin{claimproof}
  If $\projection=\emptytuple$, then $\projectionAlt=\emptytuple$ and we are done. If $\projection\neq\emptytuple$, then
  $\tset(\projection)\neq\emptyset$, and we proceed as follows.
  By assumption, $\projection\in\Projections(\tD)$, i.e., there is a $\ct\in\tD$ such that $\projection\in\Projections(\ct)$.
  Note that $\tset(\projection)\subseteq\tset(\at)\cap\tset(\ct)$, and thus $\stp(\at,\ct)\neq\emptyset$. We consider the mapping
  $\beta_{\at,\bt}$ from Fact~\ref{app:fact:beta-ab}\eqref{item:GD:app:fact:beta-ab} and let $\dt\deff\beta_{\at,\bt}(\ct)$.
  From Fact~\ref{app:fact:beta-ab}\eqref{item:GD:app:fact:beta-ab} we know that $\dt\in\tD$ and $h(w_{\ct})=h(w_{\dt})$ and
  $\stp(\at,\ct)=\stp(\bt,\dt)$.

  Since $\projection=(p_1,\ldots,p_n)\in\Projections(\ct)$, there exist pairwise distinct indices
  $k_1,\ldots,k_n\in \set{1,\ldots,\ar(\ct)}$ such that
  $\projection=(c_{k_1},\ldots,c_{k_n})$.
  Since $\tset(\projection)\subseteq\tset(\slice)\subseteq\tset(\at)$, there exist indices $\ell_1,\ldots,\ell_n\in\set{1,\ldots,\ar(\at)}$ such that
  $\projection=(a_{\ell_1},\ldots,a_{\ell_n})$.
  In summary, we have
  \[
    \projection
    \ = \
    (s_{i_1},\ldots,s_{i_n})
    \ = \
    (c_{k_1},\ldots,c_{k_n})
    \ = \
    (a_{\ell_1},\ldots,a_{\ell_n})
    \qquad \text{and} \qquad
    \projectionAlt
    \ = \
    (t_{i_1},\ldots,t_{i_n}).
  \]  
  For all $\nu\in[n]$ we have $c_{k_\nu}=a_{\ell_\nu}$, i.e., $(\ell_\nu,k_\nu)\in\stp(\at,\ct)=\stp(\bt,\dt)$, and hence
  $d_{k_\nu}=b_{\ell_\nu}$.
  Moreover, for all $\nu\in[n]$ we have $s_{i_\nu}=a_{\ell_\nu}$, i.e., $(\ell_\nu,i_\nu)\in\stp(\at,\slice)=\stp(\bt,\sliceAlt)$, and hence
  $t_{i_\nu}=b_{\ell_\nu}$.
  In summary, this yields that
  \[
    \projectionAlt
    \ = \
    (t_{i_1},\ldots,t_{i_n})
    \ = \
    (b_{\ell_1},\ldots,b_{\ell_n})
    \ = \
    (d_{k_1},\ldots,d_{k_n}).
  \]
  Since $\dt\in\tD$ and $k_1,\ldots,k_n$ are pairwise distinct
  elements in $\set{1,\ldots,\ar(\dt)}$, this proves that
  $\projectionAlt\in\Projections(\dt)\subseteq\Projections(\tD)$.

    To complete the proof of the claim, we have to show that $\stp(\ct,\projection)=\stp(\dt,\projectionAlt)$.
  For ``$\subseteq$'' consider an arbitrary tuple $(\mu,\nu)\in\stp(\ct,\projection)$. I.e., $c_\mu=p_\nu=s_{i_\nu}=c_{k_\nu}=a_{\ell_\nu}$. Hence, $(\ell_\nu,\mu)\in\stp(\at,\ct)=\stp(\bt,\dt)$. Thus,
  $d_\mu=b_{\ell_\nu}=t_{i_\nu}=q_\nu$, i.e., $(\mu,\nu)\in\stp(\dt,\projectionAlt)$.
  This proves that $\stp(\ct,\projection)\subseteq\stp(\dt,\projectionAlt)$. The inclusion ``$\supseteq$'' can be shown analogously.
\end{claimproof}

\begin{claim*}
$g(v_\projectionAlt)=c$.
\end{claim*}
\begin{claimproof}
Let $\ct,\dt$ be chosen according to Claim~\ref{claim:Projections_LemmaD9}.
  Since $h(w_{\ct}) = h(w_{\dt})$ it must hold that $\stp(\ct) = \stp(\dt)$.
  Let $\pproj[\ct, \dt]$ be the bijection according to Lemma~\ref{app:lemma:bijection}.
  Since $\pproj[\ct, \dt]$ is unique according to Lemma~\ref{app:lemma:bijection-props}(1), it must hold that $\projectionAlt = \pproj[\ct, \dt](\projection)$.
  From Proposition~\ref{app:prop:e-neigbors-of-tuples} we know that $g(v_{\projection'}) = g(v_{\pproj[\ct, \dt](\projection')})$ holds for all $\projection' \in \Projections(\ct)$.
  Thus, in particular, $g(v_{\projection}) = g(v_{\projectionAlt})$, i.e., $g(v_{\projectionAlt}) = c$.
\end{claimproof}

In summary, from the above two claims we obtain that
$\projectionAlt\in Q_{\lambda,c}$. Hence, $|Q_{\lambda,c}|=1=|P_{\lambda,c}|$, and the proof for Case~2 is completed.
This completes the proof of Proposition~\ref{app:prop:f-neighbors-slices}.
\end{proof} 
\medskip\noindent
It remains to show part (3b) of Lemma~\ref{app:lemma:g-is-stable}:
\begin{proposition}\label{app:prop:f-neighbors-projections}
	For all $\projection, \projectionAlt \in \Projections(\tD) \setminus \slices(\tD)$ with $g(v_{\projection}) = g(v_{\projectionAlt})$ it holds that
	\[
		\mset[\big]{ (\stp(\projection, \projection'), g(v_{\projection'})) \mid
			\projection' \in N_{\projection} }
		=
		\mset[\big]{ (\stp(\projectionAlt, \projection'), g(v_{\projection'})) \mid
			\projection' \in N_{\projectionAlt} }.\endProposition
	\]
\end{proposition}
\begin{proof}
	Let $\projection, \projectionAlt \in \Projections(\tD)
        \setminus \slices(\tD)$ with $g(v_{\projection}) =
        g(v_{\projectionAlt})$.
        If $\projection=\emptytuple$ then also
        $\projectionAlt=\emptytuple$ and we are done.
        If $\projection\neq\emptytuple$, we proceed as follows.

	By definition of $g$, $\stp(\projection) = \stp(\projectionAlt)$ holds and there is a bijection $\delta\colon M_{\projection} \to M_{\projectionAlt}$ such that for all $\tup{c} \in M_{\projection}$ it holds that $\stp(\tup{c}, \projection) = \stp(\delta(\tup{c}), \projectionAlt)$ and $h(w_{\tup{c}}) = g(w_{\tup{c}}) = g(w_{\delta(\tup{c})}) = h(w_{\delta(\tup{c})})$.
	Choose $\at \in M_{\projection}$ and let $\bt \isdef \delta(\at)$.
	Choose a slice $\slice$ with $\tset(\slice) = \tset(\projection)$ and note that $\slice \in \slices(\at)$.
	Let $\sliceAlt \isdef \pproj(\slice)$.
	Note that $\stp(\at, \slice) = \stp(\bt, \sliceAlt)$ and $h(v_{\slice}) = h(v_{\sliceAlt})$.
	It follows from Proposition~\ref{app:prop:f-neighbors-slices} that there is a bijection $\delta'\colon N_{\slice} \to N_{\sliceAlt}$ such that for all $\projection' \in N_{\slice}$ with $\projectionAlt' \isdef \delta'(\projection')$ it holds that:
	\begin{enumerate}
		\item $\stp(\slice, \projection') = \stp(\sliceAlt, \projectionAlt')$,\quad and
		\item $g(v_{\projection'}) = g(v_{\projectionAlt'})$.
	\end{enumerate}

	\begin{claim*}
		$\tset(\sliceAlt) = \tset(\projectionAlt)$ and $\stp(\projection, \slice) = \stp(\projectionAlt, \sliceAlt)$.
	\end{claim*}
	\noindent\emph{Proof of Claim}:\;
	$\tset(\slice)=\tset(\projection)$ implies that there exist
        $r$ and $\ell_1,\ldots,\ell_r$ such that
        $\slice=(p_{\ell_1},\ldots,p_{\ell_r})$ and
        $\tset(\projection)=\set{p_{\ell_1},\ldots,p_{\ell_r}}$.
        From $\stp(\projection)=\stp(\projectionAlt)$ we obtain that
        $\tset(\projectionAlt)=\set{q_{\ell_1},\ldots,q_{\ell_r}}$.
        
        On the other hand, $\slice\in\slices(\at)$ implies that there
        are $j_1,\ldots,j_r$ such that
        $\slice=(a_{j_1},\ldots,a_{j_r})$.
        Using $\stp(\at,\slice)=\stp(\bt,\sliceAlt)$ yields $\sliceAlt=(b_{j_1},\ldots,b_{j_r})$.
        And from
        $\slice=(p_{\ell_1},\ldots,p_{\ell_r})$ we obtain that
        $p_{\ell_1}=a_{j_1}$, \ldots, $p_{\ell_r}=a_{j_r}$.
        Using $\stp(\projection,\at)=\stp(\projectionAlt,\bt)$, we
        obtain that $q_{\ell_1}=b_{j_1}$, \ldots,
        $q_{\ell_r}=b_{j_r}$.
        Combining this with $\sliceAlt=(b_{j_1},\ldots,b_{j_r})$
        yields $\sliceAlt=(q_{\ell_1},\ldots,q_{\ell_r})$.
        In particular, we obtain that
        $\tset(\sliceAlt)=\set{q_{\ell_1},\ldots,q_{\ell_r}}=\tset(\projectionAlt)$.
        Using $\stp(\projection)=\stp(\projectionAlt)$ then yields
        $\stp(\projection,\slice)=\stp(\projectionAlt,\sliceAlt)$.
	\hfill$\blacksquare$\smallskip

	From $\tset(\slice) = \tset(\projection)$ and  $\tset(\sliceAlt) = \tset(\projectionAlt)$ we obtain that $N_{\projection} = N_{\slice}$ and $N_{\projectionAlt} = N_{\sliceAlt}$.
	Hence, $\delta'$ is also a bijection between $N_{\projection}$ and $N_{\projectionAlt}$.
	From $\stp(\projection, \slice) = \stp(\projectionAlt, \sliceAlt)$ we obtain that for all $\projection' \in N_{\projection}$ with $\projectionAlt' \isdef \delta'(\projection')$ it holds that $\stp(\projection, \projection') = \stp(\projectionAlt, \projectionAlt')$ and $g(v_{\projection'}) = g(v_{\projectionAlt'})$.
	This completes the proof of Proposition~\ref{app:prop:f-neighbors-projections}
\end{proof}

\medskip\noindent
To summarize, we prove Lemma~\ref{app:lemma:g-is-stable} as follows.
\begin{proof}[Proof of Lemma~\ref{app:lemma:g-is-stable}]~
	\begin{enumerate}
		\item 
		This follows directly from Proposition~\ref{app:prop:e-neigbors-of-tuples}.

		\item
		\begin{enumerate}
			\item
			This follows directly from Proposition~\ref{app:prop:e-neighbors-of-projections}(a).
			\item
			This follows directly from Proposition~\ref{app:prop:f-neighbors-slices}.
		\end{enumerate}

		\item
		\begin{enumerate}
			\item 
			This follows directly from Proposition~\ref{app:prop:e-neighbors-of-projections}(b).
			\item
			This follows directly from Proposition~\ref{app:prop:f-neighbors-projections}.\qedhere
		\end{enumerate}
	\end{enumerate}
\end{proof}  
\end{document}